\documentclass[11pt,draftcls,onecolumn]{IEEEtran}
\usepackage{graphicx}
\usepackage{amsmath}
\usepackage{mathrsfs}
\usepackage{amssymb}
\usepackage{float}
\usepackage{url}
\usepackage{verbatim}
\usepackage{dsfont}
\usepackage{enumerate}
\usepackage{layout}

\usepackage[left=20mm,right=20mm,top=30mm,bottom=30mm]{geometry}

\newtheorem{theorem}{Theorem}
\newtheorem{lemma}{Lemma}
\newtheorem{proposition}{Proposition}
\newtheorem{corollary}{Corollary}

\newtheorem{definition}{Definition}
\newtheorem{example}{Example}
\newtheorem{remark}{Remark}
\newcommand{\prob}{\ensuremath{\mathbb{P}}}

\newcommand{\Reals}{\ensuremath{\mathbb{R}}}
\newcommand{\set}{\ensuremath{\mathcal}}
\newcommand{\PU}{\ensuremath{P_{\mathtt{1}}}}
\newcommand{\PZ}{\ensuremath{P_{\mathtt{0}}}}
\newcommand{\dint}{\displaystyle\int}
\newcommand{\overbar}[1]{\mkern 1.5mu\overline{\mkern-1.5mu#1\mkern-1.5mu}\mkern 1.5mu}
\DeclareMathOperator*{\esssup}{ess\,sup}
\DeclareMathOperator*{\essinf}{ess\,inf}

\begin{document}
\title{$f$-Divergence Inequalities}

\markboth{SASON and VERD\'{U}: $\MakeLowercase{f}$-divergence inequalities}{}

\author{Igal Sason, and Sergio Verd\'{u}
\thanks{
I. Sason is  with the Department of Electrical Engineering, Technion--Israel
Institute of Technology, Haifa 32000, Israel (e-mail: sason@ee.technion.ac.il).}
\thanks{
S. Verd\'{u} is with the Department of Electrical Engineering, Princeton University,
Princeton, New Jersey 08544, USA (e-mail: verdu@princeton.edu).}
\thanks{
This manuscript is identical to a journal paper to appear in
the {\em IEEE Transactions on Information Theory}, vol.~62, 2016 \cite{SV-IT16},
apart of some additional material which includes Sections~\ref{subsec: RE-Delta-TV functional domination}
and~\ref{subsec:C/Delta}, and three technical proofs.}
\thanks{
Parts of this work have been presented at the {\em 2014 Workshop on Information
Theory and Applications}, San-Diego, Feb.~2014,  at
the {\em 2015 IEEE Information Theory Workshop}, Jeju Island, Korea, Oct.~2015, and the
{\em 2016 IEEE International Conference on the Science of Electrical Engineering},
Eilat, Israel, Nov.~2016.}
\thanks{
This work has been supported by the Israeli Science Foundation (ISF) under
Grant 12/12, by NSF Grant CCF-1016625, by the Center for Science of Information,
an NSF Science and Technology Center under Grant CCF-0939370, and by ARO under
MURI Grant W911NF-15-1-0479.}}

\maketitle

\begin{abstract}
This paper develops systematic approaches to obtain $f$-divergence inequalities,
dealing with pairs of probability measures defined on arbitrary alphabets. Functional
domination is one such approach, where special emphasis is placed on finding the
best possible constant upper bounding a ratio of $f$-divergences. Another approach
used for the derivation of bounds among $f$-divergences relies on moment inequalities
and the logarithmic-convexity property, which results in tight bounds on the relative
entropy and Bhattacharyya distance in terms of $\chi^2$ divergences.
A rich variety of bounds are shown to hold under boundedness assumptions on the
relative information. Special attention is devoted to the total variation distance and
its relation to the relative information and relative entropy, including
``reverse Pinsker inequalities," as well as on the $E_\gamma$ divergence,
which generalizes the total variation distance. Pinsker's inequality is extended
for this type of $f$-divergence, a result which leads to an inequality linking
the relative entropy and relative information spectrum.
Integral expressions of the R\'{e}nyi divergence in terms of the relative information
spectrum are derived, leading to bounds on the R\'{e}nyi divergence in terms
of either the variational distance or relative entropy.
\end{abstract}


\section{Introduction}
Throughout their development, information theory, and more generally, probability theory,
have benefitted from non-negative measures of dissimilarity, or loosely speaking,
distances, between pairs of probability measures defined on the same measurable space
(see, e.g., \cite{GibbsSu02,LieseV_book87,Vajda_1989}).
Notable among those measures are (see Section~\ref{sec:preliminaries} for definitions):
\begin{itemize}
\item
total variation distance $| P - Q|$;
\item
relative entropy $D(P \|Q)$;
\item
$\chi^2$-divergence $\chi^2 (P\|Q )$;
\item
Hellinger divergence $\mathscr{H}_{\alpha}(P \| Q)$;
\item
R\'enyi divergence $D_\alpha (P \| Q)$.
\end{itemize}
It is useful, particularly in proving convergence results,
to give bounds of one measure of dissimilarity in terms of another.
The most celebrated among those bounds is Pinsker's inequality:\footnote{The
folklore in information theory is that \eqref{eq: Pinsker} is due to Pinsker
\cite{Pinsker60}, albeit with a suboptimal constant. As explained in
\cite{Verdu_ITA14}, although no such inequality appears in \cite{Pinsker60},
it is possible to put together two of Pinsker's bounds to conclude that
$ \tfrac1{408} |P-Q|^2 \log e \leq {D(P\|Q)} $. }
\begin{align}  \label{eq: Pinsker}
\tfrac12  |P-Q|^2 \log e \leq D(P \| Q)
\end{align}
proved by  Csisz\'{a}r\footnote{Csisz\'ar derived \eqref{eq: Pinsker}
in \cite[Theorem~4.1]{Csiszar67a} after publishing a weaker version in
\cite[Corollary~1]{Csiszar66} a year earlier.} \cite{Csiszar67a}
and Kullback \cite{kullbackTV67}, with Kemperman \cite{kemperman} independently a bit later.
Improved and generalized versions of Pinsker's inequality have been studied, among others,
in \cite{FedotovHT_IT03}, \cite{Gilardoni06}, \cite{Gilardoni10}, \cite{OrdentlichW_IT2005},
\cite{ReidW11}, \cite{Sason_IT2016}, \cite{Vajda_IT1970}.

Relationships among measures of distances between probability measures have long
been a focus of interest in probability theory and statistics (e.g., for studying
the rate of convergence of measures). The reader is referred to surveys in
\cite[Section~3]{GibbsSu02}, \cite[Chapter~2]{LieseV_book87}, \cite{ReidW11} and
\cite[Appendix~3]{Reiss_book1989}, which provide several relationships among useful
$f$-divergences and other measures of dissimilarity between probability measures.
Some notable existing bounds among $f$-divergences include, in addition to \eqref{eq: Pinsker}:
\begin{itemize}
\item\cite[Lemma 1]{lecam1973convergence}, \cite[p. 25]{lecamyoung}
\begin{align} \label{eq: lecam73}
\mathscr{H}_{\frac12}^2(P \| Q) & \leq | P - Q |^2  \\
& \leq \mathscr{H}_{\frac12}(P \| Q) \, \bigl(4 - \mathscr{H}_{\frac12}(P \| Q) \bigr);
\end{align}

\item\cite[(2.2)]{BretagnolleH79}
\begin{align}  \label{eq: BretagnolleH79}
\tfrac14 |P-Q|^2 \leq 1-\exp\bigl(-D(P\|Q)\bigr);
\end{align}

\item \cite[Theorem~5]{GibbsSu02}, \cite[Theorem~4]{Dragomir00b}, \cite{suthesis}
\begin{align} \label{grout425 - introduction}
D ( P \| Q) &\leq \log \bigl( 1 + \chi^2 ( P \| Q ) \bigr);
\end{align}

\item \cite[Corollary~5.6]{GSS_IT14}
For all $\alpha \geq 2$
\begin{align} \label{eq: power divergence}
\chi^2(P \| Q) \leq \Bigl(1 + (\alpha-1) \, \mathscr{H}_{\alpha}(P \| Q) \Bigr)^{\frac{1}{\alpha-1}} - 1;
\end{align}
the inequality in \eqref{eq: power divergence} is reversed if $\alpha \in (0,1) \cup (1,2]$, and
it holds with equality if $\alpha=2$.

\item \cite{Gilardoni06,Gilardoni06-cor}, \cite[(58)]{ReidW11}
\begin{align} \label{eq: lb chi-square - TV}
\chi^2(P\|Q) \geq \left\{
\begin{array}{ll}
|P-Q|^2, & \quad \mbox{$|P-Q| \in \bigl[0, 1]$} \\[0.2cm]
\frac{|P-Q|}{2-|P-Q|}, & \quad \mbox{$|P-Q| \in \bigl(1, 2)$.}
\end{array}
\right.
\end{align}

\item \cite{Simic08}
\begin{align}
\label{eq1: Simic08}
\hspace*{-0.4cm} \frac{D^2(P\|Q)}{D(Q\|P)} & \leq \tfrac12 \, \chi^2(P \| Q) \, \log e; \\[0.2cm]
\hspace*{-0.4cm} 4\, \mathscr{H}^2_{\frac12}(P \| Q) \, \log^2 e
\label{eq2: Simic08}
& \leq D(P \| Q) \; D(Q \| P) \\
\label{eq3: Simic08}
\hspace*{-0.4cm} & \leq \tfrac14 \, \chi^2(P\|Q) \; \chi^2(Q\|P) \, \log^2 e; \\[0.2cm]
\hspace*{-0.4cm} 4\, \mathscr{H}_{\frac12}(P \| Q) \, \log e
\label{eq4: Simic08}
& \leq D(P \| Q) + D(Q \| P) \\
\label{eq5: Simic08}
\hspace*{-0.4cm} & \leq \tfrac12 \, \bigl(\chi^2(P\|Q) + \chi^2(Q\|P)\bigr) \, \log e;
\end{align}

\item \cite[(2.8)]{DiaconisS96}
\begin{align} \label{eq: DiaconisS96 - introduction}
D(P \| Q) \leq \tfrac12 \left(|P-Q| + \chi^2(P \| Q)\right) \log e \, ;
\end{align}

\item \cite{Gilardoni10}, \cite[Corollary~32]{ReidW11}, \cite{Sason_IT2015}
\begin{align}
& D(P \| Q) + D(Q \| P)
\label{eq1: symmetric f-div}
\geq |P-Q| \; \log\left( \frac{2+|P-Q|}{2-|P-Q|} \right), \\[0.2cm]
\label{eq2: symmetric f-div}
& \chi^2( P \| Q ) + \chi^2( Q \| P ) \geq \frac{8 \, |P-Q|^2}{4 - |P - Q|^2};
\end{align}

\item \cite[p.~711]{HoeffdingW58} (cf. a generalized form in \cite[Lemma~A.3.5]{Reiss_book1989})
\begin{align} \label{eq: HoeffdingW58}
\mathscr{H}_{\frac12} (P \| Q) \, \log e \leq  D (P \| Q),
\end{align}
generalized in \cite[Proposition~2.15]{LieseV_book87}:
\begin{align}
\label{eq1: Hel-RE-RD}
& \mathscr{H}_{\alpha}(P\|Q) \, \log e \leq D_{\alpha}(P\|Q) \leq D(P\|Q),
\end{align}
for $\alpha \in (0,1)$, and
\begin{align}
\label{eq2: Hel-RE-RD}
& \mathscr{H}_{\alpha}(P\|Q) \, \log e \geq D_{\alpha}(P\|Q) \geq D(P\|Q)
\end{align}
for $\alpha \in (1,\infty)$.

\item \cite[Proposition~2.35]{LieseV_book87}
If $\alpha \in (0,1)$, $\beta \triangleq \max\{\alpha, 1-\alpha\}$, then
\begin{align}
& 1- \left(1+\tfrac12 \, |P-Q| \right)^\beta \left(1-\tfrac12 \, |P-Q|\right)^{1-\beta} \nonumber \\
& \leq (1-\alpha) \, \mathscr{H}_\alpha(P\|Q)  \\
& \leq \tfrac12 \, |P-Q|.
\end{align}
\item
\cite[Theorems 3 and 16]{ErvenH14}
\begin{itemize}
\item
$D_\alpha (P \| Q)$ is monotonically increasing in $\alpha > 0$;
\item
$ \left( \frac1\alpha - 1 \right) D_\alpha (P \| Q)$ is monotonically
decreasing in $\alpha \in (0,1]$;
\item \cite[Proposition~2.7]{LieseV_book87} the same monotonicity properties
hold for $\mathscr{H}_{\alpha}(P \| Q)$.
\end{itemize}
\item \cite{Gilardoni10}
If  $\alpha \in (0,1]$, then
\begin{align}  \label{eq: Pinsker-Gillardoni}
\frac{\alpha}{2} \, |P-Q|^2  \log e  \leq D_\alpha (P \| Q);
\end{align}

\item
An inequality for $\mathscr{H}_\alpha(P\|Q)$ \cite[Lemma~1]{KumarS16}
becomes (at $\alpha=1$), the parallelogram identity \cite[(2.2)]{Csiszar75}
\begin{align}
D(P_0 \| Q) + D(P_1 \| Q)
= D(P_0 \| P_{\frac12}) + D(P_1 \| P_{\frac12}) + 2 D(P_{\frac12} \| Q).
\label{eq: parallelogram}
\end{align}
with $P_{\frac12} = \tfrac12 (P_0 + P_1)$, and
 extends a result for
the relative entropy in \cite[Theorem~2.1]{Csiszar75}.
\item
A ``reverse Pinsker inequality", providing an upper bound on the relative
entropy in terms of the total variation distance, does not exist in general
since we can find distributions which are arbitrarily close in total variation
but with arbitrarily high relative entropy.
Nevertheless, it is possible to introduce constraints under which such
reverse Pinsker inequalities hold.
In the special case of a finite alphabet $\set{A}$, Csisz\'ar and Talata
\cite[p.~1012]{CsiszarT_IT06} showed that
\begin{align}  \label{eq: CsTa}
D(P \| Q) \leq \left(\frac{\log e}{Q_{\min}} \right) \cdot |P-Q|^2
\end{align}
when $Q_{\min} \triangleq \min_{a \in \set{A}} Q(a)$ is positive.\footnote{Recent
applications of \eqref{eq: CsTa} can be found in \cite[Appendix~D]{KostinaV15} and
\cite[Lemma~7]{TomamichelT_IT13} for the analysis of the third-order asymptotics
of the discrete memoryless channel with or without cost constraints.}

\item
\cite[Theorem~3.1]{Csiszar72} if $f \colon (0, \infty) \to \Reals$ is
a strictly convex function, then there exists a real-valued function $\psi_f$
such that $\lim_{x \downarrow 0} \psi_f(x) = 0$
and\footnote{Eq.~\eqref{eq: Csiszar72} follows as a special case of
\cite[Theorem~3.1]{Csiszar72} with $m=1$ and $w_m=1$.}
\begin{align} \label{eq: Csiszar72}
|P-Q| \leq \psi_f \bigl( D_f(P \| Q) \bigr).
\end{align}
which implies
\begin{align} \label{eq2: Csiszar72}
\lim_{n \to \infty} D_f(P_n  \| Q_n) = 0 \; \Rightarrow \; \lim_{n \to \infty} |P_n  - Q_n| = 0.
\end{align}
\end{itemize}

The numerical optimization of an $f$-divergence subject to simultaneous
constraints on $f_i$-divergences $(i=1, \ldots , L)$
was recently studied in \cite{GSS_IT14}, which showed that for that purpose
it is enough to restrict attention to alphabets of cardinality $L+2$.
Earlier, \cite{HarremoesV_2011} showed that if $L=1$, then either the solution
is obtained by a pair $(P,Q)$ on a binary alphabet, or it is a deflated version
of such a point. Therefore, from a purely numerical standpoint, the minimization
of $D_f (P \| Q)$ such that $D_g(P \| Q) \geq d$ can be accomplished by a grid
search on $[0,1]^2$. Occasionally, as in the case where $D_f(P\|Q)=D(P\|Q)$ and
$D_g(P\|Q) = |P-Q|$, it is actually possible to determine analytically the locus
of $(D_f (P \| Q), D_g (P \| Q) )$ (see \cite{FedotovHT_IT03}).
In fact, as shown in \cite[(22)]{Vajda_1972}, a binary alphabet suffices
if the single constraint is on the total variation distance. The same conclusion
holds when minimizing the R\'enyi divergence \cite{Sason_IT2016}.

In this work, we find relationships among the various divergence measures
outlined above as well as a number of other measures of dissimilarity between
probability measures.
The framework of $f$-divergences, which encompasses the foregoing measures
(R\'enyi divergence is a one-to-one transformation of the Hellinger divergence)
serves as a convenient playground.

The rest of the paper is structured as follows:
\par
Section~\ref{sec:preliminaries} introduces the basic definitions needed
and in particular the various measures of dissimilarity between probability
measures used throughout.
\par
Based on \textit{functional domination}, Section~\ref{sec: functional domination}
provides a basic tool for the derivation of bounds among $f$-divergences.
Under mild regularity conditions, this approach further enables to prove
the optimality of constants in those bounds. In addition, we show instances where
such optimality can be shown in the absence of regularity conditions.
The basic tool used in Section~\ref{sec: functional domination} is
exemplified in obtaining relationships among important $f$-divergences
such as relative entropy, Hellinger divergence and total variation distance. This approach is also
useful in strengthening and providing an alternative proof of Samson's inequality
\cite{samson2000concentration} (a counterpart to Pinsker's inequality using Marton's divergence, useful
in proving certain concentration of measure results \cite{boucheron2013concentration}),
whose constant we show cannot be improved. In addition, we show several
new results in
Section~\ref{subsec: f-divergence and total variation distance}  on
the maximal ratios of  various $f$-divergences to total variation distance.

\par
Section~\ref{sec:bounded} provides an approach for bounding ratios of
$f$-divergences, assuming that the
relative information (see Definition~\ref{def:RI}) is lower
and/or upper bounded with probability one.
The approach is
exemplified in bounding ratios of relative entropy to various $f$-divergences,
and analyzing the local behavior of $f$-divergence ratios when the reference
measure is fixed. We also show that bounded relative information leads to a
strengthened version of Jensen's inequality, which, in turn, results
in upper and lower bounds on the ratio of the non-negative difference
$\log \left( 1 + \chi^2 ( P \| Q) \right) - D( P\|Q)$ to $D(Q\| P)$.
A new reverse version of Samson's inequality is another byproduct of the main
tool in this section.
\par
The rich structure of the total variation
distance as well as its importance in both fundamentals and applications merits placing
special attention on bounding the rest of the distance measures in terms of $|P-Q|$.
Section~\ref{sec: TV-RI} gives several useful identities linking
the total variation distance with the relative information spectrum, which
result in a number of upper and lower bounds on $|P-Q|$, some of which are
tighter than Pinsker's inequality in various ranges of the parameters.
It also provides refined bounds on $D(P \| Q)$ as a
function of $\chi^2$-divergences and the total variation distance.

\par
Section~\ref{sec:reverseP} is devoted to proving ``reverse Pinsker inequalities,"
namely, lower bounds on $|P - Q|$ as a function of $D(P\|Q)$ involving either
(a) bounds on the relative information,
(b) Lipschitz constants, or
(c) the minimum mass of the reference measure (in the finite alphabet case).
In the latter case, we also examine the relationship between entropy and the
total variation distance from the equiprobable distribution, as well as the
exponential decay of the probability that an independent identically distributed
sequence is not strongly typical.
\par
Section~\ref{sec:EG} focuses on the $E_\gamma$ divergence. This $f$-divergence
generalizes the total variation distance, and its utility in information theory has
been exemplified in \cite{CZK98,LiuCV1_IT15,LiuCV1_ISIT15,LiuCV2_ISIT15,PPV10,PV-Allerton10,PW15}.
Based on the operational interpretation
of the DeGroot statistical information \cite{DeGroot62} and the integral representation
of $f$-divergences as a function of DeGroot's measure,
Section~\ref{sec:EG} provides an integral representation of
$f$-divergences as a function of the $E_\gamma$ divergence; this representation
shows that $\{(E_\gamma( P \| Q ), E_\gamma(Q \| P)), \gamma \geq 1\}$ uniquely
determines $D( P \| Q )$ and $\mathscr{H}_{\alpha}(P \| Q)$, as well as any other
$f$-divergence with twice differentiable $f$. Accordingly, bounds on the $E_\gamma$
divergence directly translate into bounds on other important $f$-divergences.
In addition, we show an extension of Pinsker's inequality \eqref{eq: Pinsker} to
$E_\gamma$ divergence, which leads to a relationship between the
relative information spectrum and relative entropy.
\par
The R\'{e}nyi divergence, which has found a plethora of information-theoretic
applications, is the focus of Section~\ref{sec:RD}. Expressions of the R\'{e}nyi
divergence are derived in Section~\ref{subsec: RD-RIS} as a function of the relative
information spectrum. These expressions lead in Section~\ref{subsec: RD-TV} to the
derivation of bounds on the R\'{e}nyi divergence as a function of
the variational distance under the boundedness assumption of the relative information.
Bounds on the R\'enyi divergence of an arbitrary order are derived in Section~\ref{subsec: RD-RE}
as a function of the relative entropy when the relative information is bounded.

\section{Basic Definitions}
\label{sec:preliminaries}
\subsection{Relative Information and Relative Entropy}
We assume throughout that the probability measures $P$ and $Q$ are
defined on a common measurable space $(\set{A}, \mathscr{F})$, and $P \ll Q$
denotes that $P$ is {\em absolutely continuous} with respect to $Q$, namely
there is no event $\set{F} \in \mathscr{F}$ such that $P(\set{F}) > 0 = Q(\set{F})$.

\begin{definition} \label{def:RI}
If $P \ll Q$, the {\em relative information} provided by $a \in \set{A}$
according to $(P,Q)$ is given by\footnote{$\frac{\text{d}P}{\text{d}Q}$ denotes the
Radon-Nikodym derivative (or density) of $P$ with respect to $Q$. Logarithms
have an arbitrary common base, and $\exp(\cdot)$ indicates the inverse function
of the logarithm with that base.}
\begin{align}  \label{eq:RI}
\imath_{P\|Q}(a) \triangleq \log \frac{\text{d}P}{\text{d}Q} \, (a).
\end{align}
\end{definition}

When the argument of the relative information is distributed according to
$P$, the resulting real-valued random variable is of particular interest.
Its cumulative distribution function and expected value are known as follows.

\begin{definition} \label{def:RIS}
If $P \ll Q$,
the {\em relative information spectrum} is the cumulative distribution function
\begin{align} \label{eq:RIS}
\mathds{F}_{P \| Q}(x) = \prob\bigl[\imath_{P\|Q}(X) \leq x \bigr],
\end{align}
with\footnote{$X\sim P$ means that $\mathbb{P} [ X \in \set{F} ] = P (\set{F}) $
for any event $\set{F} \in \mathscr{F}$.} $X \sim P$.
The \textit{relative entropy} of $P$ with respect to $Q$ is
\begin{align}
\label{eq1:RE1}
D(P \| Q) &= \mathbb{E} \bigl[ \imath_{P \| Q}(X) \bigr] \\
\label{eq2:RE1}
&= \mathbb{E} \bigl[ \imath_{P \| Q}(Y) \, \exp\bigl(\imath_{P \| Q}(Y) \bigr) \bigr],
\end{align}
where $Y \sim Q$.
\end{definition}

\subsection{$f$-Divergences} \label{section:fD}
Introduced by Ali-Silvey \cite{AliS} and Csisz\'ar \cite{Csiszar63, Csiszar67a},
a useful generalization of the relative entropy, which retains some of its
major properties (and, in particular, the data processing inequality
\cite{ZakaiZiv75}), is the class of $f$-divergences. A general
definition of an $f$-divergence is given in \cite[p.~4398]{LieseV_IT2006},
specialized next to the case where $P \ll Q$.

\begin{definition} \label{def:fD}
Let $f \colon (0, \infty) \to \Reals$ be a convex function,
and suppose that $P \ll Q$. The {\em $f$-divergence} from $P$ to $Q$ is given
by
\begin{align} \label{eq:fD}
D_f(P\|Q) &= \int f \left(\frac{\text{d}P}{\text{d}Q}\right) \, \text{d}Q \\
&= \mathbb{E} \bigl[f(Z) \bigr]
\end{align}
where
\begin{align} \label{eq: Z}
Z = \exp\bigl(\imath_{P\|Q}(Y)\bigr), \quad Y \sim Q.
\end{align}
and in \eqref{eq:fD}, we took the continuous extension\footnote{The
convexity of $f \colon (0, \infty) \to \Reals$ implies its continuity on $(0, \infty)$.}
\begin{align} \label{eq: f at 0}
f(0) = \lim_{t \downarrow 0} f(t) \in (-\infty, +\infty].
\end{align}
\end{definition}

We can also define $D_f(P\|Q)$ without requiring $P \ll Q$.  Let
$f \colon (0, \infty) \to \Reals$ be a convex function with $f(1)=0$, and let
$f^\star \colon (0, \infty) \to \Reals$ be given by
\begin{align} \label{eq: fstar}
f^\star(t) = t \, f\left(\tfrac1t\right)
\end{align}
for all $t > 0$. Note that $f^\star$ is also convex (see, e.g., \cite[Section~3.2.6]{BoydV}),
$f^\star(1)=0$, and $D_f(P \| Q) = D_{f^\star}(Q \| P)$ if $P \ll \gg Q$. By definition, we take
\begin{align}
\label{eq: fstar at 0}
f^\star(0) = \lim_{t \downarrow 0} f^\star(t) = \lim_{u \to \infty} \frac{f(u)}{u}.
\end{align}
If $p$ and $q$ denote, respectively, the densities of $P$ and $Q$ with respect to a
$\sigma$-finite measure $\mu$ (i.e., $p = \frac{\text{d}P}{\text{d}\mu}$,
$q=\frac{\text{d}Q}{\text{d}\mu}$), then we can write \eqref{eq:fD} as
\begin{align} \label{eq:fD2}
D_f(P\|Q) &= \int q \; f\left(\frac{p}{q}\right) \, \text{d}\mu \\
& = \int_{\{pq > 0\}} q \, f \left( \frac{p}{q} \right) \,\mathrm{d} \mu + f(0) \, Q( p = 0 ) + f^\star(0) P (q = 0).
\label{dpqverygralalterf}
\end{align}

\begin{remark} \label{remark: equivalence-fD}
Different functions may lead to the same $f$-divergence for all $(P,Q)$: if for an
arbitrary $b \in \Reals$, we have
\begin{align}  \label{eq: fD3}
f_b(t) = f_0(t) + b \, (t-1), \quad t \geq 0
\end{align}
then
\begin{align}  \label{eq: fD4}
D_{f_0}(P\|Q) = D_{f_b}(P\|Q).
\end{align}
\end{remark}

\par
The following key property of $f$-divergences follows from Jensen's inequality.

\begin{proposition} \label{prop: fGibbs}
If $f \colon (0, \infty) \to \Reals$ is convex and $f(1)=0$,
$P \ll Q$, then
\begin{align}  \label{eq: fGibbs}
D_f(P\|Q) \geq 0.
\end{align}
If, furthermore, $f$ is strictly convex at $t=1$, then equality in \eqref{eq: fGibbs}
holds if and only if $P=Q$.
\end{proposition}

Surveys on general properties
of $f$-divergences can be found in \cite{LieseV_book87,Vajda_1989,Vajda_2009}.

The assumptions of Proposition~\ref{prop: fGibbs} are satisfied by many interesting
measures of dissimilarity between probability measures. In particular,
the following examples  receive particular attention in this paper.
As per Definition~\ref{def:fD}, in each case the function $f$ is defined on $(0, \infty)$.
\begin{enumerate}[1)]
\item {\em Relative entropy} \cite{SKRAL51}:  $f(t) = t\, \log t$,
\begin{align} \label{eq: 1st KL divergence}
D(P\|Q) &= D_f(P\|Q) \\
&= D_r(P \| Q) \label{eq: r-divergence}
\end{align}
with
$r \colon (0, \infty) \to [0, \infty)$ defined as
\begin{align}  \label{eq: r}
r(t) \triangleq t \log t + (1-t) \, \log e .
\end{align}

\item {\em Relative entropy}: ($P\ll \gg Q$)  $f(t) = -\log t$,
\begin{align} \label{eq: 2nd KL divergence}
D(Q\|P) &= D_f(P\|Q);
\end{align}

\item {\em Jeffrey's divergence} \cite{jeffreys46}: ($ P \ll \gg Q$) $f(t) = (t-1) \, \log t$,
\begin{align} \label{jeffreys}
D( P \| Q) + D(Q\|P) &= D_f(P\|Q);
\end{align}

\item {\em $\chi^2$-divergence} \cite{Pearson1900x}:
$f(t) = (t-1)^2$ or $f(t) = t^2-1$,
\begin{align}
\label{eq: chi-square 1}
\chi^2(P\|Q) &= D_f(P\|Q) \\
\label{eq: chi-square 1b}
\chi^2(P\|Q) & = \int \left(\frac{\text{d}P}{\text{d}Q} - 1 \right)^2 \, \text{d}Q \\
\label{eq: chi-square 2}
& = \int \left(\frac{\text{d}P}{\text{d}Q}\right)^2 \, \text{d}Q - 1 \\
\label{eq: chi-square 3}
& = \mathbb{E} \bigl[ \exp\bigl(2 \imath_{P\|Q}(Y)\bigr) \bigr] - 1 \\
\label{eq: chi-square 4}
& = \mathbb{E} \bigl[ \exp\bigl(\imath_{P\|Q}(X)\bigr) \bigr] - 1
\end{align}
with $X \sim P$ and $Y \sim Q$. Note that if $P \ll \gg Q$, then from the right side of
\eqref{eq: chi-square 2}, we obtain
\begin{align}
\label{eq: chi-square 5}
\chi^2(Q\|P) = D_g(P\|Q)
\end{align}
with $g(t) = \frac{1}{t}-t = f^\star (t)$, and $f(t) = t^2-1$.

\item {\em Hellinger divergence of order $\alpha \in (0,1) \cup (1, \infty)$}
\cite{jeffreys46}, \cite[Definition~2.10]{LieseV_book87}:
\begin{align} \label{eq: Hel-divergence}
\mathscr{H}_{\alpha}(P \| Q) = D_{f_\alpha}(P \| Q)
\end{align}
with
\begin{align} \label{eq: H as fD}
f_\alpha(t) = \frac{t^\alpha-1}{\alpha-1}.
\end{align}
The $\chi^2$-divergence is the Hellinger divergence of order~2, while
$\tfrac12 \mathscr{H}_{\frac12}(P \| Q)$ is usually referred to as the
\textit{squared Hellinger distance}.
The analytic extension of $\mathscr{H}_{\alpha}(P \| Q)$ at $\alpha=1$ yields
\begin{align} \label{eq: Hel-divergence order 1}
\mathscr{H}_1(P \| Q) \, \log e = D(P \| Q).
\end{align}

\item {\em Total variation distance}: Setting
\begin{align} \label{eq: TV as fD}
f(t)=|t-1|
\end{align}
results in
\begin{align}
\label{eq1: TV distance}
|P-Q| &=  D_f(P\|Q) \\[0.1cm]
\label{eq2: TV distance}
&= \int \left|\frac{\text{d}P}{\text{d}Q} - 1 \right| \, \text{d}Q \\[0.1cm]
\label{eq3: TV distance}
& = 2 \, \sup_{\set{F} \in \mathscr{F}} \bigl(P(\set{F}) - Q(\set{F})\bigr).
\end{align}

\item {\em Triangular Discrimination} \cite{LeCam86,Vincze81} (a.k.a. Vincze-Le Cam distance):
\begin{align} \label{eq:delta}
\Delta(P\|Q) = D_f(P\|Q)
\end{align}
with
\begin{align} \label{eq:tridiv}
f(t) = \frac{(t-1)^2}{t+1}.
\end{align}
Note that
\begin{align}
\label{eq1: Delta and chi^2}
\tfrac12 \, \Delta  (P \| Q) & = \chi^2 (P \, \| \, \tfrac12 P + \tfrac12 Q) \\
\label{eq2: Delta and chi^2}
& = \chi^2 (Q \, \| \, \tfrac12 P + \tfrac12 Q).
\end{align}

\item {\em Jensen-Shannon divergence} \cite{Lin91} (a.k.a. capacitory discrimination):
\begin{align} \label{eq:js1}
\mathrm{JS}(P\|Q) &= D\left(P \, \| \, \tfrac12 P + \tfrac12 Q \right)
+ D\left(Q \, \| \, \tfrac12 P + \tfrac12 Q \right) \\
\label{eq:js11}
& = D_f(P\|Q)
\end{align}
with
\begin{align}
\label{eq:js2}
f(t) = t \log t - (1+t) \log\left(\frac{1+t}{2}\right).
\end{align}

\item {\em $E_\gamma$ divergence} (see, e.g., \cite[p.~2314]{PPV10}): For $\gamma \geq 1$,
\begin{align} \label{eq:Eg f-div}
E_{\gamma}(P \| Q) = D_{f_\gamma}(P\|Q)
\end{align}
with
\begin{align} \label{eq: f for EG}
f_\gamma(t) = (t-\gamma)^+
\end{align}
where $(x)^+ \triangleq \max\{x,0\}$.
$E_\gamma$ is sometimes called  ``hockey-stick
divergence" because of the shape of $f_\gamma$.
If $\gamma = 1$, then
\begin{align} \label{eq:EG-TV}
E_1(P\|Q) = \tfrac12 \, |P-Q|.
\end{align}

\item {\em DeGroot statistical information} \cite{DeGroot62}:
For $p \in (0,1)$,
\begin{align} \label{eq:DG f-div}
\mathcal{I}_p(P\|Q) = D_{\phi_p}(P\|Q)
\end{align}
with
\begin{align} \label{eq: f for DG}
\phi_p(t) = \min \{p, 1-p\} - \min \{p, 1-pt\}.
\end{align}
Invoking \eqref{eq:Eg f-div}--\eqref{eq: f for DG}, we get (cf. \cite[(77)]{LieseV_IT2006})
\begin{align} \label{eq3: DG-TV}
\mathcal{I}_{\frac12}(P\|Q) = \tfrac12 E_1(P\|Q) = \tfrac14 \, |P-Q|.
\end{align}

This measure was first proposed by DeGroot \cite{DeGroot62} due to its
operational meaning in Bayesian statistical hypothesis testing (see Section~\ref{saens}),
and it was later identified as an $f$-divergence (see \cite[Theorem~10]{LieseV_IT2006}).

\item {\em Marton's divergence} \cite[pp.~558--559]{Marton96}:
\begin{align} \label{smarton}
d_2^2(P,Q)  &= \min \mathbb{E} \left[ \prob^2[ X \neq Y \, | \, Y ] \right] \\
&= \label{smarton as f-div}
D_s( P\| Q)
\end{align}
where the minimum is over all probability measures $P_{XY}$ with respective
marginals $P_X=P$ and $P_Y=Q$, and
\begin{align} \label{eq: s}
s(t) = (t - 1)^2 \; 1\{ t < 1 \}.
\end{align}
Note that Marton's divergence satisfies the triangle inequality \cite[Lemma~3.1]{Marton96},
and $d_2(P,Q)=0$ implies $P=Q$; however, due to its asymmetry, it is not a distance measure.
\end{enumerate}

\subsection{R\'{e}nyi Divergence}
Another generalization of relative entropy was introduced by R\'enyi \cite{Renyientropy}
in the special case of finite alphabets. The general definition (assuming\footnote{R\'{e}nyi
divergence can also be defined without requiring absolute continuity, e.g., \cite[Definition~2]{ErvenH14}.}
$P \ll Q$) is the following.

\begin{definition}  \label{def:RD}
Let $P \ll Q$. The {\em R\'{e}nyi divergence of order $\alpha \geq 0$} from $P$ to $Q$ is
given as follows:
\begin{itemize}
\item
If $\alpha \in (0,1) \cup (1, \infty) $, then
\begin{align}
D_{\alpha}(P\|Q)
\label{eq:RD1}
&= \frac1{\alpha-1} \; \log \Bigl( \mathbb{E}\bigl[\exp\bigl(\alpha
\, \imath_{P\|Q}(Y)\bigr)\bigr] \Bigr) \\[0.1cm]
\label{eq:RD3}
& = \frac1{\alpha-1} \; \log \Bigl( \mathbb{E} \bigl[ \exp \bigl((\alpha-1)
\, \imath_{P\|Q}(X)\bigr)\bigr] \Bigr)
\end{align}
with $X \sim P$ and $Y \sim Q$.
\item If $\alpha = 0$,  then\footnote{The function in \eqref{eq:RD1} is, in general, right-discontinuous
at $\alpha=0$. R\'{e}nyi \cite{Renyientropy} defined $D_0(P\|Q)=0$, while we have followed \cite{ErvenH14}
defining it instead as $\lim_{\alpha \downarrow 0} D_{\alpha}(P\|Q)$.}
\begin{align} \label{eq: d0}
D_0 (P \| Q ) = \max_{\set{F} \in \mathscr{F}\colon P(\set{F}) = 1} \log \left(\frac1{Q (\set{F})}\right).
\end{align}
\item If $\alpha =1$, then
\begin{align} \label{eq: d1}
D_1(P\|Q) = D(P\|Q)
\end{align}
which is the analytic extension of $D_{\alpha}(P \| Q)$ at $\alpha=1$.
If $D(P\|Q) < \infty$, it can be verified by L'H\^{o}pital's rule that
$D(P\|Q) = \lim_{\alpha \uparrow 1} D_{\alpha}(P \| Q)$.
\item If $\alpha = +\infty$ then
\begin{align} \label{def:dinf}
D_{\infty}(P\|Q) = \log \left(\esssup \frac{\text{d}P}{\text{d}Q} \, (Y)\right)
\end{align}
with $Y \sim Q$. If $P \not \ll Q$, we take $D_{\infty}(P\|Q)  = \infty$.
\end{itemize}
\end{definition}

R\'enyi divergence is a one-to-one transformation of Hellinger divergence
of the same order $\alpha \in (0,1) \cup (1,\infty)$:
\begin{align}\label{renyimeetshellinger}
D_\alpha(P \| Q ) = \frac1{\alpha -1} \; \log \left( 1 + (\alpha - 1)
\, \mathscr{H}_\alpha(P \| Q) \right)
\end{align}
which, when particularized to order~2 becomes
\begin{align} \label{pasavenida}
D_2 ( P \| Q) = \log \left( 1 + \chi^2 ( P \| Q ) \right).
\end{align}
Note that \eqref{eq: power divergence}, \eqref{eq1: Hel-RE-RD}, \eqref{eq2: Hel-RE-RD}
follow from \eqref{renyimeetshellinger} and the monotonicity of the R\'{e}nyi divergence
in its order, which in turn yields \eqref{eq: HoeffdingW58}.

Introduced in \cite{Bhattachryya}, the Bhattacharyya distance was popularized
in the engineering literature in \cite{Kailath67}.
\begin{definition}  \label{definition: B distance}
The \textit{Bhattacharyya distance} between $P$ and $Q$, denoted by $B(P \| Q)$, is given by
\begin{align}
\label{eq1: B distance}
B(P \| Q) & = \tfrac12 \, D_{\frac12}(P\|Q) \\
\label{eq2: B distance}
& = \log \biggl( \frac1{1-\tfrac12 \mathscr{H}_{\frac12}(P \| Q)} \biggr).
\end{align}
\end{definition}
Note that, if $P \ll \gg Q$, then $B(P\|Q) = B(Q\|P)$ and $B(P\|Q)=0$ if and
only if $P=Q$, though $B(P\|Q)$ does not satisfy the triangle inequality.

\section{Functional Domination} \label{sec: functional domination}
Let $f$ and $g$ be convex functions on $(0, \infty)$ with $f(1)=g(1)=0$, and let $P$ and $Q$
be probability measures defined on a measurable space $(\set{A}, \mathscr{F})$.
If there exists $\alpha > 0$ such that $f(t) \leq \alpha g(t)$ for all $t \in (0, \infty)$,
then it follows from Definition~\ref{def:fD} that
\begin{align} \label{000}
D_f(P \|Q) \leq \alpha\, D_g(P \| Q).
\end{align}
This simple observation leads to a proof of, for example, \eqref{eq: HoeffdingW58}
and the left inequality in \eqref{eq: lecam73} with the aid of Remark~\ref{remark: equivalence-fD}.

\subsection{Basic Tool} \label{subsec: basic tool of functional domination}

\begin{theorem} \label{theorem: tight bound}
Let $P \ll Q$, and assume
\begin{itemize}
\item $f$ is convex  on
$(0, \infty)$ with $f(1)=0$;
\item
$g$ is convex  on
$(0, \infty)$ with $g(1)=0$;
\item
$g(t) > 0$ for all $t \in (0,1) \cup (1, \infty)$.
\end{itemize}
 Denote the function $\kappa\colon (0,1) \cup (1, \infty) \to \Reals$
\begin{align} \label{kappadef-1}
\kappa(t) &= \frac{f(t)}{g(t)}, \quad t \in (0,1) \cup (1, \infty)
\end{align}
and
\begin{align}
\label{barkdef}
\bar{\kappa} &= \sup_{t  \in (0,1) \cup (1, \infty)} \kappa(t).
\end{align}
Then,
\begin{enumerate}[a)]
\item \label{theorem: tight bound: parta}
\begin{align}  \label{eq: tight bound--0}
D_f(P \| Q) \leq  \bar{\kappa} \, D_g(P \| Q).
\end{align}
\item \label{theorem: tight bound: partb}
If, in addition, $f'(1)=g'(1)=0$, then
\begin{align}  \label{eq: tight bound}
\sup_{P \neq Q} \frac{D_f(P \| Q)}{D_g(P \| Q)} = \bar{\kappa}.
\end{align}
\end{enumerate}
\end{theorem}

\medskip

\begin{proof}
\begin{enumerate}[a)]
\item
The bound in \eqref{eq: tight bound--0} follows from \eqref{000} and
$f(t) \leq \bar{\kappa} \, g(t)$ for all $t>0$.
\item
Since $g$ is positive except at $t=1$,
$D_g(P\|Q) > 0$ if $P \neq Q$. The convexity of $f, g$ on $(0, \infty)$
implies their continuity; and since $g(t)>0$ for all $t \in (0,1) \cup (1, \infty)$,
$\kappa (\cdot)$ is continuous on both  $(0,1)$ and $(1, \infty)$.
\par
To show \eqref{eq: tight bound}, we fix an arbitrary $\nu \in (0, 1) \cup (1,\infty)$
and construct a sequence of pairs of probability measures whose ratio of $f$-divergence
to $g$-divergence converges to $ \kappa(\nu)$.
To that end, for sufficiently small $\varepsilon > 0$, let $P_{\varepsilon}$ and
$Q_{\varepsilon}$ be parametric probability measures defined on the set
$\set{A} = \{0, 1\}$ with $P_{ \varepsilon}(0) = \nu \, \varepsilon$ and
$Q_{\varepsilon}(0) = \varepsilon$. Then,
\begin{align}
\lim_{\varepsilon \to 0} \, \frac{D_f(P_{\varepsilon} \|
Q_\varepsilon)}{D_g(P_{\varepsilon} \| Q_{\varepsilon})}
& =  \lim_{\varepsilon \to 0} \, \frac{\varepsilon \, f(\nu) + (1-\varepsilon) \,
f\left(\frac{1-\nu \varepsilon}{1-\varepsilon}\right)}{\varepsilon \, g(\nu) + (1-\varepsilon) \,
g\left(\frac{1-\nu \varepsilon}{1-\varepsilon}\right)} \\[0.1cm]
&= \lim_{\alpha \to 0}
\frac{f(\nu) + \frac{\nu -1}{\alpha} \; f(1 - \alpha)}{g(\nu) +
\frac{\nu -1}{\alpha} \; g(1-\alpha)} \label{arlington} \\[0.1cm]
& = \kappa(\nu) \label{cun}
\end{align}
where \eqref{arlington} holds
by change of variable $\varepsilon = \alpha / ( \nu - 1 + \alpha )$,
and \eqref{cun} holds by the assumption on the derivatives of $f$ and $g$ at 1,
the assumption that $f(1)=g(1)=0$, and the continuity of $\kappa (\cdot)$ at $\nu$.
If $\bar{\kappa} = \kappa ( \nu)$ we are done.
If the supremum in \eqref{barkdef} is not attained on $(0, 1) \cup (1,\infty)$,
then the right side of \eqref{cun} can be made
arbitrarily close to $\bar{\kappa}$ by an appropriate choice of $\nu$.
\end{enumerate}
\end{proof}

\begin{remark}
Beyond the restrictions in Theorem~\ref{theorem: tight bound}\ref{theorem: tight bound: parta}),
the only operative restriction imposed by
Theorem~\ref{theorem: tight bound}\ref{theorem: tight bound: partb}) is the differentiability
of the functions $f$ and $g$ at $t=1$. Indeed, we can invoke Remark~\ref{remark: equivalence-fD}
and add $f'(1) \, (1-t)$ to $f(t)$, without changing $D_f$ (and likewise with $g$) and thereby
satisfying the condition in Theorem~\ref{theorem: tight bound}\ref{theorem: tight bound: partb});
the stationary point at~1 must be a minimum of both $f$ and $g$ because of the assumed
convexity, which implies their non-negativity on $(0, \infty)$.
\end{remark}

\begin{remark}
It is useful to generalize Theorem~\ref{theorem: tight bound}\ref{theorem: tight bound: partb})
by dropping the assumption on the existence of the derivatives at~1. To that end, note that
the inverse transformation used for the transition to \eqref{arlington} is given by
$\nu = 1 + \alpha \left(\tfrac1\varepsilon-1\right)$ where $\varepsilon > 0$ is sufficiently
small, so if $\nu > 1$ (resp. $\nu < 1$), then $\alpha > 0$ (resp. $\alpha < 0$).
Consequently, it is easy to see from \eqref{arlington} that if $\bar{\kappa} = \sup_{t>1} \kappa(t)$,
the construction in the proof can restrict to $\nu > 1$, in which case it is enough to require that
the left derivatives of $f$ and $g$ at~1 be equal to $0$. Analogously, if
$\bar{\kappa} = \sup_{0<t<1} \kappa(t)$, it is enough to require that the right derivatives
of $f$ and $g$ at~1 be equal to $0$. When neither left nor right derivatives at $1$ are $0$, then
\eqref{eq: tight bound} need not hold as the following example shows.
\end{remark}

\begin{example}\label{example:counter}
Let $f(t) = |t-1|$ and
\begin{align}
g(t) &= 2f(t) + 1-t.
\end{align}
Then, $\bar{\kappa} = 1$, while in view of \eqref{eq: fD4}
and \eqref{eq1: TV distance} for all $(P,Q)$,
\begin{align}
D_g (P\|Q) = 2 D_f (P\|Q) = 2 |P-Q|.
\end{align}
\end{example}

\subsection{Relationships Among $D(P \| Q)$, $\chi^2(P \| Q)$ and $|P-Q|$}
\label{subsec: RE-chi2-TV functional domination}

Since the R\'{e}nyi divergence of order $\alpha > 0 $ is monotonically increasing in $\alpha$,
\eqref{pasavenida} yields
\begin{align} \label{grout425}
D ( P \| Q) &\leq \log \left( 1 + \chi^2 ( P \| Q ) \right) \\
\label{eq: CsiszarT06}
&\leq \chi^2(P \| Q) \log e.
\end{align}
Inequality \eqref{grout425}, which can be found in \cite{suthesis} and \cite[Theorem~5]{GibbsSu02},
is sharpened in Theorem~\ref{thm:d-chi} under the assumption of bounded relative information.
In view of \eqref{renyimeetshellinger}, an alternative way to sharpen \eqref{grout425} is
\begin{align}
D ( P \| Q) &\leq \frac1{\alpha-1} \, \log \left( 1 + (\alpha -1) \mathscr{H}_\alpha ( P \| Q ) \right)
\end{align}
for $\alpha \in (1,2)$, which is tight as $\alpha \to 1$.

\par
Relationships between the relative entropy, total variation distance and $\chi^2$ divergence are derived next.

\begin{theorem}
\begin{enumerate}[a)]
\item
If $P \ll Q$ and $c_1, c_2 \geq 0$, then
\begin{align} \label{eq: Improved DiaconisS96}
D(P \| Q) \leq \left(c_1 \, |P-Q| + c_2 \, \chi^2(P \| Q) \right) \log e
\end{align}
holds if $(c_1, c_2) = (0,1)$ and $(c_1, c_2) = \bigl(\tfrac14, \tfrac12\bigr)$. Furthermore,
if $c_1=0$ then $c_2=1$ is optimal, and if $c_2 = \tfrac12$ then $c_1 = \tfrac14$ is optimal.
\item
\begin{align}
\label{eq: symmetrized RE-chi^2}
\sup \frac{D(P \| Q) + D(Q \| P)}{\chi^2(P \| Q) + \chi^2(Q \| P)} = \tfrac12 \log e
\end{align}
where the supremum is over $P \ll \gg Q$ and $P \neq Q$.
\end{enumerate}
\end{theorem}

\begin{proof}
\begin{enumerate}[a)]
\item
The satisfiability of \eqref{eq3: Improved DiaconisS96} with $(c_1, c_2) = (0,1)$ is equivalent
to \eqref{eq: CsiszarT06}.

\par
Let $P \ll Q$, and $Y \sim Q$. Then,
\begin{align}
D(P \| Q)
\label{eq1: Improved DiaconisS96}
& = \mathbb{E} \left[ r \left(\frac{\text{d}P}{\text{d}Q} \, (Y) \right) \right] \\
\label{eq2: Improved DiaconisS96}
& \leq \tfrac12 \, \mathbb{E}
\left[ \left( 1 - \frac{\text{d}P}{\text{d}Q} \, (Y) \right)^+ +
\left( \frac{\text{d}P}{\text{d}Q} \, (Y) - 1 \right)^2 \right] \, \log e \\
\label{eq3: Improved DiaconisS96}
& = \left( \tfrac14 \, |P-Q| + \tfrac12 \, \chi^2(P\|Q) \right) \, \log e
\end{align}
where \eqref{eq1: Improved DiaconisS96} follows from the definition of relative entropy
with the function $r \colon (0, \infty) \to \Reals$ defined in \eqref{eq: r};
\eqref{eq2: Improved DiaconisS96} holds since for $t \in (0, \infty)$
\begin{align} \label{SV}
r(t) \leq \tfrac12 \left[ (1-t)^+ + (t-1)^2 \right] \, \log e
\end{align}
and \eqref{eq3: Improved DiaconisS96} follows from \eqref{eq: chi-square 1b}, \eqref{eq2: TV distance},
and the identity $$(1-t)^+ = \tfrac12 \, \bigl[|1-t| + (1-t) \bigr].$$ This proves \eqref{eq: Improved DiaconisS96}
with $(c_1, c_2) = \bigl(\tfrac14, \tfrac12\bigr)$.

\par
Next, we show that if $c_1=0$ then $c_2 = 1$ is the best possible constant in \eqref{eq: Improved DiaconisS96}.
To that end, let $f_2(t) = (t-1)^2$, and let $\kappa(t)$ be the continuous extension of $\frac{r(t)}{f_2(t)}$.
It can be verified that the function $\kappa$ is monotonically decreasing on $(0, \infty)$, so
\begin{align}
\bar{\kappa} = \lim_{t \downarrow 0} \kappa (t) = \log e.
\end{align}
Since $D_{r}(P\|Q) = D(P\|Q)$ and $D_{f_2}(P\|Q) = \chi^2(P\|Q)$, and
$r'(1) = f_2'(1) = 0$, the desired result follows from
Theorem~\ref{theorem: tight bound}\ref{theorem: tight bound: partb}).

\par
To show that $c_1 = \tfrac14$ is the best possible constant in
\eqref{eq: Improved DiaconisS96} if $c_2 = \tfrac12$, we let
$g_2(t) = \tfrac12 \left[ (1-t)^+ + (t-1)^2 \right]$.
Theorem~\ref{theorem: tight bound}\ref{theorem: tight bound: partb})
does not apply here since $g_2$ is not differentiable at $t=1$.
However, we can still construct probability measures for proving
the optimality of the point $(c_1, c_2) = (\tfrac14, \tfrac12)$. To that end,
let $\varepsilon \in (0,1)$, and define probability measures $P_\varepsilon$
and $Q_\varepsilon$ on the set $\set{A} = \{0,1\}$ with
$P_\varepsilon(1) = \varepsilon^2$ and $Q_\varepsilon(1)=\varepsilon$.
Since $D_r(P\|Q) = D(P\|Q)$ and $D_{g_2}(P\|Q) = \tfrac14 \, |P-Q| + \tfrac12 \, \chi^2(P \| Q)$,
\begin{align}
& \lim_{\varepsilon \downarrow 0} \frac{D(P_\varepsilon \| Q_\varepsilon)}{\tfrac14 \,
|P_\varepsilon-Q_\varepsilon| + \tfrac12 \, \chi^2(P_\varepsilon \| Q_\varepsilon)} \nonumber \\[0.2cm]
& = \lim_{\varepsilon \downarrow 0} \frac{(1-\varepsilon) \, r(1+\varepsilon) +
\varepsilon \, r(\varepsilon)}{(1-\varepsilon) \, g_2(1+\varepsilon) + \varepsilon \, g_2(\varepsilon)} \\[0.2cm]
\label{eq:limitratio2}
& = \log e
\end{align}
where \eqref{eq:limitratio2} holds since we can write numerator and denominator in the right side as
\begin{align}
(1-\varepsilon) \, r(1+\varepsilon) + \varepsilon \, r(\varepsilon)
\label{eq: calculation1}
& = (1-\varepsilon^2) \, \log(1+\varepsilon) + \varepsilon^2 \, \log \varepsilon  \\
& = \varepsilon \, \log e + o(\varepsilon) ,\\[0.1cm]
\label{eq: calculation2}
(1-\varepsilon) \, g_2(1+\varepsilon) + \varepsilon \, g_2(\varepsilon)
& = \varepsilon - \varepsilon^2.
\end{align}
\item
We have
\begin{align}
\label{eq4: symmetrized RE-chi^2}
D_{f}(P \| Q) &= D(P \| Q) + D(Q \| P), \\
\label{eq5: symmetrized RE-chi^2}
D_{g}(P \| Q)  &=  \chi^2(P \| Q) + \chi^2(Q \| P)
\end{align}
with
\begin{align}
\label{eq1: symmetrized RE-chi^2}
& f(t) \triangleq (t-1) \, \log t \\
\label{eq2: symmetrized RE-chi^2}
& g(t) \triangleq t^2-t-1+\frac1{t} \\
& \kappa(t) = \frac{f(t)}{g(t)} =  \frac{t \log t}{t^2-1} , \quad t\in (0,1) \cup (1, \infty) \\
& \lim_{t\to 1} \kappa(t) = \bar{\kappa} = \tfrac12 \log e \label{man}
\end{align}
where \eqref{man} is easy to verify since $\kappa$ is monotonically increasing
on $(0,1)$, and monotonically decreasing on $(1, \infty)$.
The desired result follows since the conditions of
Theorem~\ref{theorem: tight bound}\ref{theorem: tight bound: partb}) apply.
\end{enumerate}
\end{proof}

\begin{remark}
Inequality~\eqref{eq: Improved DiaconisS96} strengthens the bound claimed in \cite[(2.8)]{DiaconisS96},
\begin{align} \label{eq: DiaconisS96}
D(P \| Q) \leq \tfrac12 \left(|P-Q| + \chi^2(P \| Q)\right) \log e ,
\end{align}
although the short outline of the suggested proof in \cite[p.~710]{DiaconisS96} leads to the weaker upper bound
$|P-Q| + \tfrac12 \, \chi^2(P \| Q)$ nats.
\end{remark}

\begin{remark}
Note that \eqref{eq: CsiszarT06} implies a looser result where the
constant in the right side of \eqref{eq: symmetrized RE-chi^2} is doubled. Furthermore,
 \eqref{eq: Pinsker} and \eqref{eq: symmetrized RE-chi^2} result in the
bound $\chi^2(P \| Q) + \chi^2(Q \| P) \geq 2 \, |P-Q|^2$ which, although weaker than
\eqref{eq2: symmetric f-div}, has the same behavior for small values of $|P-Q|$.
\end{remark}

\subsection{Relationships Among $D(P \| Q)$, $\Delta(P \| Q)$ and $|P-Q|$}
\label{subsec: RE-Delta-TV functional domination}

The next result shows three upper bounds on the Vincze-Le Cam distance in terms of
the relative entropy (and total variation distance).
\begin{theorem} \
\label{thm: RE-Delta-TV}
Let $r \colon (0, \infty) \to \Reals$ be the function defined in \eqref{eq: r},
and let $f_{\Delta}$  denote the function $f \colon (0, \infty) \to \Reals$
defined in \eqref{eq:tridiv}.
\begin{enumerate}[a)]
\item If $P \ll Q$, then
\begin{align}
\label{eq1: RE-Delta}
\Delta(P \| Q) \, \log e \leq c_1 \, D(P \| Q)
\end{align}
where $c_1 = 1.11591 = \bar{\kappa}$, computed with
\begin{align} \label{eq2: RE-Delta}
\kappa(t) = \frac{f_{\Delta}(t) \, \log e}{r(t)}.
\end{align}
Furthermore, the constant $c_1$ in \eqref{eq1: RE-Delta} is the best possible.
\item If $P \ll Q$, then
\begin{align}
\label{eq2: RE-Delta-TV}
\Delta(P \| Q) \, \log e \leq D(P \| Q) + c_2 \, |P-Q| \, \log e
\end{align}
where $c_2 = 0.0374250$ is defined by
\begin{align}
\label{eq: c2 RE-Delta-TV}
c_2 &= \min \bigl\{c \geq 0 \colon \, \kappa_c(t) \leq 1, \; \forall \, t \in (0,1) \bigr\} \\
\label{eq2.5: RE-Delta-TV}
\kappa_c(t) &= \frac{f_{\Delta}(t) \, \log e}{r(t) + 2c \, (1-t)^+ \, \log e}.
\end{align}
Furthermore, the constant $c_2$ in \eqref{eq2: RE-Delta-TV} is the best possible.
\item If $P \ll \gg Q$, then
\begin{align}
\label{eq3: RE-Delta}
\Delta(P \| Q) \, \log e \leq \tfrac12 \, D(P \| Q) + \tfrac12 \, D(Q \| P)
\end{align}
and the bound in \eqref{eq3: RE-Delta} is tight.
\end{enumerate}
\end{theorem}

\begin{proof}
\begin{enumerate}[a)]
\item
Let $f = f_{\Delta} \, \log e$ and $g=r$. These functions satisfy the conditions in
Theorem~\ref{theorem: tight bound}\ref{theorem: tight bound: partb}), and result in
the function $\kappa$ defined in \eqref{eq2: RE-Delta}. The function $\kappa$
is maximal at $t^\star = 0.223379$ with $c_1 \triangleq \kappa(t^\star) = 1.11591$. Since
$D_f(P \| Q) = \Delta(P \| Q) \, \log e$ and $D_g(P\|Q) = D(P\|Q)$,
Theorem~\ref{theorem: tight bound}\ref{theorem: tight bound: partb})
results in \eqref{eq1: RE-Delta} with the optimality of its constant $c_1$.

\item By the definition of $c_2$ in \eqref{eq: c2 RE-Delta-TV}, it follows that for all $t > 0$
\begin{align}
\label{eq2a: RE-Delta-TV}
f_{\Delta}(t) \, \log e \leq r(t) + 2c_2 \, (1-t)^+ \log e;
\end{align}
note that \eqref{eq2a: RE-Delta-TV} holds for all $t \in (0,\infty)$ (i.e., not only in $(0,1)$
as in \eqref{eq: c2 RE-Delta-TV}) since, for $t \geq 1$, \eqref{eq2a: RE-Delta-TV} is equivalent
to $f_{\Delta}(t) \, \log e \leq r(t)$ which indeed holds for all $t \in [1,\infty)$.
Hence, the bound in \eqref{eq2: RE-Delta-TV} follows from \eqref{eq: r-divergence}, \eqref{eq:delta},
\eqref{eq:Eg f-div} and \eqref{eq:EG-TV}. It can be verified from \eqref{eq: c2 RE-Delta-TV} that
$c_2 = 0.0374250$, the maximal value of $\kappa_{c_2}$ on $(0,\infty)$ is attained at
$t^\star = 0.122463$ and $\kappa_{c_2}(t^\star)=1$. The optimality of the constant $c_2$ in
\eqref{eq2: RE-Delta-TV} is shown next (note that Theorem~\ref{theorem: tight bound}\ref{theorem: tight bound: partb})
cannot be used here, since the right side of \eqref{eq2a: RE-Delta-TV} is not differentiable at~$t=1$).
Let $P_{\varepsilon}$ and $Q_{\varepsilon}$ be probability measures defined on $\set{A} = \{0, 1\}$
with $P_{\varepsilon}(0) = t^\star \varepsilon$ and $Q_{\varepsilon}(0) = \varepsilon$ for
$\varepsilon \in (0,1)$. Then, it follows that
\begin{align}
& \lim_{\varepsilon \to 0} \frac{\Delta(P_\varepsilon \| Q_\varepsilon) \, \log e}{D(P_\varepsilon
\| Q_\varepsilon) + c_2 \, \bigl|P_\varepsilon - Q_\varepsilon\bigr| \, \log e} \nonumber \\[0.1cm]
\label{eq d2: RE-Delta-TV}
& = \lim_{\varepsilon \to 0} \frac{f_{\Delta}(t^\star) \, \varepsilon \, \log e + o(\varepsilon)}{\bigl[
r(t^\star) + 2c_2 \, (1-t^\star) \, \log e \bigr] \, \varepsilon + o(\varepsilon)} \\
\label{eq d3: RE-Delta-TV}
& = \kappa_{c_2}(t^\star) = 1
\end{align}
where \eqref{eq d2: RE-Delta-TV} follows by the construction of $P_{\varepsilon}$ and
$Q_{\varepsilon}$, and the use of Taylor series expansions;
\eqref{eq d3: RE-Delta-TV} follows from \eqref{eq2.5: RE-Delta-TV} (note that $t^\star < 1$),
which yields the required optimality result in \eqref{eq2: RE-Delta-TV}.

\item
Let $f = f_{\Delta} \, \log e$ and $g \colon (0,\infty) \to \Reals$ be defined by
$g(t)=(t-1) \, \log t$ for all $t > 0$. These functions, which satisfy the conditions
in Theorem~\ref{theorem: tight bound}\ref{theorem: tight bound: partb}), yield
$D_f(P \| Q) = \Delta(P\|Q) \, \log e$ and $D_g(P \| Q) = D(P\|Q) + D(Q\|P)$.
Theorem~\ref{theorem: tight bound}\ref{theorem: tight bound: partb}) yields
the desired result with
\begin{align}
\kappa(t) &= \frac{t-1}{(t+1) \log_e t}, \quad t \in (0,1) \cup (1, \infty) \\
\kappa(1) &= \frac1{2} = \bar{\kappa}.
\end{align}
\end{enumerate}
\end{proof}

\begin{remark}
We can generalize \eqref{eq1: RE-Delta} and \eqref{eq3: RE-Delta} to
\begin{align} \label{eqref: generalized bound}
\Delta(P \| Q) \log e \leq c_1 \, D(P \| Q) + c_2 \, D(Q \| P)
\end{align}
where $c_1, c_2 \geq 0$ and $P \ll \gg Q$. In view of Theorem~\ref{theorem: tight bound},
the locus of the allowable constant pairs $(c_1, c_2)$ in \eqref{eqref: generalized bound} can be evaluated.
To that end, Theorem~\ref{theorem: tight bound}\ref{theorem: tight bound: partb}) can be used with
$f(t) = f_{\Delta}(t) \, \log e$, and $g(t) = c_1 r(t) + c_2 \, t \, r\left(\frac1{t}\right)$ for all $t>0$.
This results in a convex region, which is symmetric with respect to the straight line $c_1 = c_2$ (since
$\Delta(P\|Q) = \Delta(Q\|P)$).
Note that \eqref{eq1: RE-Delta}, \eqref{eq3: RE-Delta} and the symmetry property of this
region identify, respectively, the points $(1.11591,0)$, $(\tfrac12, \tfrac12)$ and $(0,1.11591)$ on its
boundary; since the sum of their coordinates is nearly constant, the boundary of this subset of the
positive quadrant is nearly a straight line of slope $-1$.
\end{remark}

\subsection{An Alternative Proof of Samson's Inequality}
\label{subsec: Samson's Inequality}
An analog of Pinsker's inequality, which comes in handy for the proof of Marton's conditional
transportation inequality \cite[Lemma~8.4]{boucheron2013concentration}, is the following bound
due to Samson \cite[Lemma~2]{samson2000concentration}:
\begin{theorem} \label{thm: Samson}
If $P \ll Q$, then
\begin{align} \label{eq: Samson}
d_2^2(P,Q) + d_2^2(Q,P) \leq \tfrac2{\log e} \; D(P\|Q)
\end{align}
where $d_2(P,Q)$ is the distance measure defined in \eqref{smarton}.
\end{theorem}

We provide next an alternative proof of Theorem~\ref{thm: Samson}, in view of
Theorem~\ref{theorem: tight bound}\ref{theorem: tight bound: partb}), with the following advantages:
\begin{enumerate}[a)]
\item This proof yields the optimality of the constant in \eqref{eq: Samson}, i.e., we prove that
\begin{align}  \label{eq: Samson - tight}
\sup \frac{d_2^2(P,Q) + d_2^2(Q,P)}{D(P \| Q)} = \frac2{\log e}
\end{align}
where the supremum is over all probability measures $P,Q$ such that $P \neq Q$ and $P \ll \gg Q$.
\item A simple adaptation of this proof enables to derive a reverse inequality to \eqref{eq: Samson}, which holds
under the boundedness assumption of the relative information (see Section~\ref{subsec: Reverse Samson's Inequality}).
\end{enumerate}

\begin{proof}
\begin{align}
\label{eq1: sum of squared d_2}
d_2^2(P,Q) + d_2^2(Q,P) & = D_s(P\|Q) + D_{s^\star}(P\|Q) \\
\label{eq2: sum of squared d_2}
& = D_f(P\|Q)
\end{align}
where, from \eqref{eq: s}, $s^\star \colon (0, \infty) \to [0, \infty)$
is the convex function
\begin{align}
\label{eq: s_star}
s^\star(t) = t \, s\bigl(\tfrac1{t}\bigr) = \frac{(t - 1)^2 \; 1\{ t > 1 \}}{t}
\end{align}
and, from \eqref{eq: s} and \eqref{eq: s_star}, the non-negative and convex function
$f \colon (0, \infty) \to \Reals$ in \eqref{eq2: sum of squared d_2} is given by
\begin{align}
\label{eq: f - Samson}
f(t) = s(t) + s^{\star}(t)
= \frac{(t-1)^2}{\max\{1,t\}}
\end{align}
for all $t > 0$. Let $r$ be the non-negative and convex function
$r \colon (0, \infty) \to \Reals$ defined in \eqref{eq: r}, which
yields $D_r(P\|Q) = D(P \| Q)$. Note that $f(1)=r(1)=0$, and the
functions $f$ and $r$ are both differentiable at $t=1$ with $f'(1)=r'(1)=0$.
The desired result follows from
Theorem~\ref{theorem: tight bound}\ref{theorem: tight bound: partb}) since
in this case
\begin{align} \label{eq: kappa - Samson}
\kappa(t) &= \frac{(t-1)^2 }{r(t) \, \max\{1,t\}} , \quad t \in (0,1) \cup (1, \infty) \\
\lim_{t \to 1} \kappa(t) &= \tfrac2{\log e} = \bar{\kappa}, \label{eq: kappa1 - Samson}
\end{align}
as can be verified from the monotonicity of $\kappa$ on $(0,1)$ (increasing)
and $(1, \infty)$ (decreasing).
\end{proof}

\begin{remark}
As mentioned in \cite[p.~438]{samson2000concentration}, Samson's inequality \eqref{eq: Samson}
strengthens the Pinsker-type inequality in \cite[Lemma~3.2]{Marton96}:
\begin{align} \label{eq: Marton96 - Lamma 3.2}
& d_2^2(P,Q) \leq \tfrac2{\log e} \; \min\bigl\{D(P\|Q), \, D(Q\|P) \bigr\}
\end{align}
Nevertheless, similarly to our alternative proof of Theorem~\ref{thm: Samson},
one can verify that Theorem~\ref{theorem: tight bound}\ref{theorem: tight bound: partb})
yields the optimality of the constant in \eqref{eq: Marton96 - Lamma 3.2}.
\end{remark}

\subsection{Ratio of $f$-Divergence to Total Variation Distance}
\label{subsec: f-divergence and total variation distance}

Vajda \cite[Theorem~2]{Vajda_1972} showed that the range of an $f$-divergence
is given by (see \eqref{eq: fstar at 0})
\begin{align} \label{eq: Vajda72}
0 \leq D_f(P \| Q) \leq f(0) + f^{\star}(0)
\end{align}
where every value in this range is attainable by a suitable pair of probability
measures $P \ll Q$. Recalling  Remark~\ref{remark: equivalence-fD}, note that
$f_b(0) + f_b^{\star}(0) = f(0) + f^{\star}(0)$ with $f_b(\cdot)$ defined in \eqref{eq: fD3}.
Basu \textit{et al.} \cite[Lemma~11.1]{BasuSP} strengthened \eqref{eq: Vajda72},
showing that
\begin{align}
\label{eq: BasuSP inequality}
D_f(P \| Q) \leq \tfrac12 \left( f(0) +  f^\star(0)\right) \, |P-Q|.
\end{align}
Note that, provided $f(0)$ and $f^\star(0)$ are finite,  \eqref{eq: BasuSP inequality}
yields a counterpart to \eqref{eq: Csiszar72}.
Next, we show that the constant in \eqref{eq: BasuSP inequality} cannot be improved.

\begin{theorem} \label{theorem: strengthened BasuSP}
If $f \colon (0, \infty) \to \Reals$ is convex with $f(1)=0$, then
\begin{align}
\label{eq: BasuSP tightness}
\sup \frac{D_f(P \| Q)}{|P-Q|} &= \tfrac12 \left( f(0) +  f^\star(0) \right)
\end{align}
where the supremum is over all probability measures $P, Q$ such that $P \ll Q$ and $P \neq Q$.
\end{theorem}

\begin{proof}
As the first step, we give a simplified proof of \eqref{eq: BasuSP inequality}
(cf. \cite[pp.~344-345]{BasuSP}).
In view of Remark~\ref{remark: equivalence-fD}, it is sufficient to show that for all $t>0$,
\begin{align}\label{paris}
f(t) + \tfrac12 \left( f(0) - f^\star(0) \right) (t-1)
\leq \tfrac12 \left( f(0) +  f^\star(0) \right) |t-1|.
\end{align}
If $t\in (0,1)$, \eqref{paris} reduces to $f(t) \leq (1-t) f(0)$, which holds in view of the
convexity of $f$ and $f(1) = 0$. If $t \geq 1$, we can readily check, with the aid of \eqref{eq: fstar},
that \eqref{paris} reduces to $f^\star (\tfrac1t) \leq (1-\tfrac1t) f^\star(0)$, which, in turn,
holds because $f^\star$ is convex and $f^\star(1) = 0$.

For the second part of the proof of \eqref{eq: BasuSP tightness}, we construct a pair
of probability measures $P_\varepsilon$ and $Q_\varepsilon$ such
that, for a sufficiently small $\varepsilon > 0$,
$\frac{D_f(P_\varepsilon\|Q_\varepsilon)}{|P_\varepsilon-Q_\varepsilon|}$
can be made arbitrarily close to the right side of \eqref{eq: BasuSP tightness}.
To that end, let $\varepsilon \in (0, \tfrac12 \, (\sqrt{5}-1)]$, and let
$P_\varepsilon$ and $Q_\varepsilon$ be defined on the set
$\set{A} = \{0,1,2\}$ with $P_\varepsilon(0) = Q_\varepsilon(1) = \varepsilon$
and $P_\varepsilon(1) = Q_\varepsilon(0) = \varepsilon^2$. Then,
\begin{align}
\lim_{\varepsilon \to 0} \frac{D_f(P_\varepsilon \| Q_\varepsilon)}{|P_\varepsilon - Q_\varepsilon|}
\label{eq1: achievability BasuSP}
& = \lim_{\varepsilon \to 0} \frac{\varepsilon \, f(\varepsilon)
+ \varepsilon^2 \, f\bigl(\tfrac1\varepsilon\bigr)}{2 \varepsilon \, (1-\varepsilon)} \\
\label{eq2: achievability BasuSP}
& = \tfrac12 \, f(0) + \tfrac12 \, f^{\star}(0)
\end{align}
where \eqref{eq1: achievability BasuSP} holds since $P_\varepsilon(2) = Q_\varepsilon(2)$,
and \eqref{eq2: achievability BasuSP} follows from \eqref{eq: f at 0} and \eqref{eq: fstar at 0}.
\end{proof}

\begin{remark}
Csisz\'ar \cite[Theorem~2]{Csiszar67b} showed that if $f(0)$ and $f^\star(0)$ are finite
and $P \ll Q$, then there exists a constant $C_f > 0$ which depends only on $f$ such that
\begin{align} \label{eq1: Csiszar67}
D_f(P \| Q) \leq C_f \, \sqrt{|P-Q|}.
\end{align}
If $|P-Q| < 1$, then \eqref{eq1: Csiszar67} is superseded by \eqref{eq: BasuSP inequality}
where the constant is not only explicit but is the best possible according to
Theorem~\ref{theorem: strengthened BasuSP}.
\end{remark}

A direct application of Theorem~\ref{theorem: strengthened BasuSP} yields
\begin{corollary}
\begin{align}
\label{eq: sup1}
& \sup_{P \neq Q} \frac{\mathscr{H}_{\alpha}(P \| Q)}{|P-Q|} = \frac1{2(1-\alpha)},
\quad \forall \, \alpha \in (0,1) \\[0.1cm]
\label{eq: sup2}
& \sup_{P \neq Q} \frac{\Delta(P \| Q)}{|P-Q|} = 1, \\[0.1cm]
\label{eq: sup3}
& \sup_{P \neq Q} \frac{\mathrm{JS}(P \| Q)}{|P-Q|} = \log 2, \\[0.1cm]
\label{eq: sup4}
& \sup_{P \neq Q} \frac{d_2^2(P,Q)}{|P-Q|} = \frac12, \\[0.1cm]
\label{eq: sup5}
& \sup_{P \neq Q} \frac{d_2^2(P,Q) + d_2^2(Q,P)}{|P-Q|} = 1
\end{align}
where the suprema in \eqref{eq: sup1}--\eqref{eq: sup4} are over all $P \ll Q$ with $P \neq Q$, and
the supremum in \eqref{eq: sup5} is over all $P \ll \gg Q$ with $P \neq Q$.
\end{corollary}

\begin{remark}
The results in \eqref{eq: sup1}, \eqref{eq: sup2} and \eqref{eq: sup3} strengthen, respectively,
the inequalities in \cite[Proposition~2.35]{LieseV_book87}, \cite[(11)]{Topsoe_IT00} and
\cite[Theorem~2]{Topsoe_IT00}. The results in \eqref{eq: sup4} and \eqref{eq: sup5} form
counterparts of \eqref{eq: Samson - tight}.
\end{remark}

\section{Bounded Relative Information} \label{sec:bounded}

In this section we show that it is possible to find bounds among $f$-divergences
without requiring a strong condition of functional domination (see Section~\ref{sec: functional domination})
as long as the relative information is upper and/or lower bounded almost surely.

\subsection{Definition of $\beta_1$ and $ \beta_2$.}
The following notation is used throughout the rest of the paper. Given a pair of probability measures $(P,Q)$
on the same measurable space, denote $\beta_1, \beta_2 \in [0,1]$ by
\begin{align}
\label{eq: beta1}
\beta_1
&= \exp\bigl(- D_\infty (P\|Q)\bigr), \\[0.1cm]
\label{eq: beta2}
\beta_2
&= \exp\bigl(-D_\infty (Q\|P)\bigr)
\end{align}
with the convention that if $D_\infty (P\|Q) = \infty$, then $\beta_1 = 0$,
and if $D_\infty (Q\|P) = \infty$, then $\beta_2 = 0$. Note that if $\beta_1 >0$, then $P \ll Q$,
while $\beta_2 >0$ implies $Q \ll P$. Furthermore, if $P \ll \gg Q$, then with $Y \sim Q$,
\begin{align}
\label{eq: beta1-alt}
\beta_1 &= \essinf \frac{\text{d}Q}{\text{d}P} \, (Y) = \left( \esssup \frac{\text{d}P}{\text{d}Q} \, (Y) \right)^{-1},\\
\label{eq: beta2-alt}
\beta_2 &= \essinf \frac{\text{d}P}{\text{d}Q} \, (Y) = \left( \esssup \frac{\text{d}Q}{\text{d}P} \, (Y) \right)^{-1}.
\end{align}
The following examples illustrate important cases in which $\beta_1$ and
$\beta_2$ are positive.

\begin{example} (\textit{Gaussian distributions.}) \label{example: Gaussians}
Let $P$ and $Q$ be Gaussian probability measures with equal means,
and variances $\sigma_0^2$ and $\sigma_1^2$ respectively. Then,
\begin{align}
\beta_1 &= \frac{\sigma_0}{\sigma_1} 1\{ {\sigma_0} \leq {\sigma_1} \}, \\
\beta_2 &= \frac{\sigma_1}{\sigma_0} 1\{ {\sigma_1} \leq {\sigma_0} \}.
\end{align}
\end{example}

\begin{example} (\textit{Shifted Laplace distributions.}) \label{example: two Laplacians}
Let $P$ and $Q$ be the probability measures whose probability density functions
are, respectively, given by $f_{\lambda} ( \cdot - a_0 )$
and $f_{\lambda}( \cdot - a_1 )$ with
\begin{align} \label{eq: two Laplacians}
f_{\lambda} (x) =  \tfrac{\lambda}{2} \, \exp(-\lambda |x|), \quad x \in \Reals
\end{align}
 where $\lambda > 0$. In this case,
\eqref{eq: two Laplacians} gives
\begin{align}
\frac{\text{d}P}{\text{d}Q} \, (x)
 = \exp\bigl( \lambda (|x-a_1| - |x-a_0|) \bigr), \quad x \in \Reals
\end{align}
which yields
\begin{align}  \label{eq: beta - 2 Laplacians}
\beta_1 = \beta_2 = \exp\bigl(-\lambda \, |a_1-a_0| \bigr) \in (0,1].
\end{align}
\end{example}

\begin{example} (\textit{Cram\'{e}r distributions.}) \label{example: two Cramer distributions}
Suppose that $P$ and $Q$ have Cram\'{e}r probability density functions $f_{\theta_1, m_1}$ and
$f_{\theta_0, m_0}$, respectively, with
\begin{align}  \label{eq: Cramer dist.}
f_{\theta, m}(x) = \frac{\theta}{2(1+\theta |x-m|)^2}, \quad x \in \Reals
\end{align}
where $\theta > 0$ and $m \in \Reals$. In this case, we have $\beta_1, \beta_2 \in (0,1)$
since the ratio of the probability density functions, $\frac{f_{\theta_1, m_1}}{f_{\theta_0, m_0}}$,
tends to $\frac{\theta_0}{\theta_1} < \infty$ in the limit $x \to \pm \infty$.
In the special case where $m_0 = m_1 = m \in \Reals$, the ratio of these
probability density functions is $\frac{\theta_1}{\theta_0}$ at $x=m$;
due also to the symmetry of the probability density functions around $m$,
it can be verified that in this special case
\begin{align}  \label{eq: beta - 2 Cramer dist.}
\beta_1 = \beta_2 = \min \left\{\frac{\theta_0}{\theta_1}, \frac{\theta_1}{\theta_0}\right\}.
\end{align}
\end{example}

\begin{example} (\textit{Cauchy distributions.}) \label{example: two Cauchy distributions}
Suppose that $P$ and $Q$ have Cauchy probability density functions $g_{\gamma_1, m_1}$ and
$g_{\gamma_0, m_0}$, respectively, $\gamma_0 \neq \gamma_1$ and
\begin{align}  \label{eq: Cauchy dist.}
g_{\gamma, m}(x) = \frac1{\pi \gamma} \left[1 + \left(\frac{x-m}{\gamma}\right)^2\right]^{-1},
\quad x \in \Reals
\end{align}
where $\gamma > 0$. In this case, we also have $\beta_1, \beta_2 \in (0,1)$
since the ratio of the probability density functions tends to
$\frac{\gamma_0}{\gamma_1} < \infty$ in the limit $x \to \pm \infty$.
In the special case where $m_0 = m_1$,
\begin{align}  \label{eq: beta - 2 Cauchy dist.}
\beta_1 = \beta_2 = \min \left\{ \frac{\gamma_1}{\gamma_0},
\, \frac{\gamma_0}{\gamma_1} \right\}.
\end{align}
\end{example}

\subsection{Basic Tool} \label{subsec:bounds among fD}
Since $\beta_1 =1 \Leftrightarrow \beta_2 =1 \Leftrightarrow P = Q$, it
is advisable to avoid trivialities by excluding that case.
\begin{theorem} \label{thm: fD1}
Let $f$ and $g$ satisfy the assumptions in Theorem~\ref{theorem: tight bound},
and assume that $(\beta_1, \beta_2) \in [0,1)^2$. Then,
\begin{align} \label{eq:fD bound2}
D_f(P \| Q)
&\leq \kappa^\star \; D_g(P \| Q)
\end{align}
where
\begin{align}
\kappa^\star = \sup_{\beta \in (\beta_2, 1) \cup (1,\beta_1^{-1})} \kappa(\beta)
\end{align}
and $\kappa(\cdot)$ is defined in \eqref{kappadef-1}.
\end{theorem}

\begin{proof}
Defining $g(0)$ and $g^\star (0)$ as in \eqref{eq: f at 0}-\eqref{eq: fstar at 0}, respectively,
note that
\begin{align}\label{mur1}
f(0) \, Q (p=0) \leq g(0) \, \kappa^\star \, Q ( p = 0 )
\end{align}
because if $\beta_2 = 0$ then $f(0) \leq \kappa^\star\, g(0)  $, and
if  $\beta_2 >0$ then $Q (p=0) = 0$. Similarly,
\begin{align}\label{mur2}
f^\star (0) \, P (q=0) \leq g^\star (0) \, \kappa^\star \, P ( q = 0 )
\end{align}
because if $\beta_1 = 0$ then $f^\star(0) \leq \kappa^\star \, g^\star(0)  $
and if  $\beta_1 >0$ then $P (q=0) = 0$.
Moreover, since $f(1)=g(1)=0$, we can substitute $pq >0$ by $\{p q >0\} \cap \{ p \neq q\}$
in the right side of \eqref{dpqverygralalterf} for $D_f(P\|Q)$ and likewise for $D_g(P\|Q)$.
In view of the definition of $\kappa^\star$, \eqref{mur1} and \eqref{mur2},
\begin{align}
D_f (P\,\|\,Q) &\leq
\kappa^\star \int_{pq > 0 , p \neq q } q \, g \left( \frac{p}{q} \right) \,\mathrm{d} \mu \nonumber\\
&+ \kappa^\star \, g(0)  \, Q(p = 0) + \kappa^\star \, g^\star (0)\, P(q = 0) \\
&= \kappa^\star\, D_g(P \| Q).
\end{align}
\end{proof}
Note that if $\beta_1 = \beta_2 = 0$, then Theorem~\ref{thm: fD1} does not improve
upon Theorem~\ref{theorem: tight bound}\ref{theorem: tight bound: parta}).

\begin{remark} \label{remark: play}
In the application of Theorem~\ref{thm: fD1}, it is often convenient to make use of the
freedom afforded by Remark~\ref{remark: equivalence-fD} and choose the corresponding offsets
such that:
\begin{itemize}
\item the positivity property of $g$ required by Theorem~\ref{thm: fD1} is satisfied;
\item the lowest $\kappa^\star$ is obtained.
\end{itemize}
\end{remark}

\begin{remark} \label{remark: tight constants}
Similarly to the proof of Theorem~\ref{theorem: tight bound}\ref{theorem: tight bound: partb}),
under the conditions therein, one can verify that the constants in Theorem~\ref{thm: fD1}
are the best possible among all probability measures $P,Q$ with given $(\beta_1, \beta_2) \in [0,1)^2$.
\end{remark}

\begin{remark}\label{remark:reverse}
Note that if we swap the assumptions on $f$ and $g$ in Theorem~\ref{thm: fD1}, the same result translates into
\begin{align}\label{eq:fD bound1}
\inf_{\beta \in (\beta_2, 1) \cup (1,\beta_1^{-1})} \kappa(\beta) \cdot D_g(P \| Q) \leq D_f(P \| Q).
\end{align}
Furthermore, provided both $f$ and $g$ are positive (except at $t=1$) and $\kappa$ is monotonically increasing,
Theorem~\ref{thm: fD1} and \eqref{eq:fD bound1} result in
\begin{align}
\kappa(\beta_2) \, D_g(P \| Q) & \leq D_f(P \| Q) \label{eq:fD bound1_pc} \\
&\leq \kappa(\beta_1^{-1}) \, D_g(P \| Q). \label{eq:fD bound2_pc}
\end{align}
In this case, if $\beta_1 > 0$, sometimes it is convenient to replace $\beta_1 > 0$ with
$\beta_1^\prime \in (0, \beta_1)$
at the expense of loosening the bound. A similar observation applies to $\beta_2$.
\end{remark}

\begin{example} \label{example: non-monotonic}
If $f(t) = (t-1)^2$ and $g(t)= |t-1|$, we get
\begin{align} \label{eq1: chi square - TV}
\chi^2 (P\|Q) \leq \max\{ \beta_1^{-1}-1, 1 - \beta_2\}  \; |P-Q|.
\end{align}
\end{example}

\subsection{Bounds on $\frac{D(P\|Q)}{D(Q\|P)}$}
\label{subsec: Bounds on RE/dual}

The remaining part of this section is devoted to various
applications of Theorem~\ref{thm: fD1}. From this point,
we make use of the definition of
$r \colon (0, \infty) \to [0, \infty)$ in \eqref{eq: r}.

An illustrative application of Theorem~\ref{thm: fD1} gives
upper and lower bounds on the ratio of relative entropies.

\begin{theorem}  \label{thm: bounds RE and dual}
Let $P \ll \gg Q$, $P \neq Q$, and $(\beta_1, \beta_2) \in (0,1)^2$.
Let $\kappa \colon (0,1)\cup(1,\infty) \to (0, \infty)$ be defined as
\begin{align}
\label{eq: kappa RE and dual}
\kappa(t) = \frac{t \log t + (1-t) \, \log e}{(t-1) \log e - \log t}.
\end{align}
Then,
\begin{align}
\label{eq: bounds RE and dual}
\kappa(\beta_2)  \leq \frac{D(P \| Q)}{D(Q \| P)} \leq \kappa(\beta_1^{-1}).
\end{align}
\end{theorem}
\begin{proof}
For $t>0$, let
\begin{align}
\label{1g}
g(t) &= -\log t + (t-1) \log e
\end{align}
then  $ D_g(P \| Q) = D(Q \| P)$,  $D_r(P\|Q) = D(P\|Q)$ and the conditions of Theorem~\ref{thm: fD1}
are satisfied. The desired result follows from
Theorem~\ref{thm: fD1} and the monotonicity of $\kappa ( \cdot )$
shown in Appendix~\ref{appendix: properties of kappa}.
\end{proof}

\subsection{Reverse Samson's Inequality}
\label{subsec: Reverse Samson's Inequality}
The next result gives a counterpart to Samson's inequality \eqref{eq: Samson}.

\begin{theorem} \label{thm: Samson - refined}
Let  $(\beta_1, \beta_2) \in (0,1)^2$.
Then,
\begin{align}
\label{eq: Samson - refined}
\inf \frac{d_2^2(P,Q) + d_2^2(Q,P)}{ D(P\|Q)} = \min \bigl\{\kappa(\beta_1^{-1}), \, \kappa(\beta_2) \bigr\}
\end{align}
where the infimum is over all $P \ll Q$ with given $(\beta_1, \beta_2)$, and
where $\kappa\colon (0,1)\cup(1,\infty) \to \bigl(0, \tfrac{2}{\log e} \bigr)$
is given in \eqref{eq: kappa - Samson}.
\end{theorem}

\begin{proof}
Applying Remark~\ref{remark:reverse} to the convex and positive (except at $t=1$) function
$f(t)$ given in \eqref{eq: f - Samson}, and  $g(t) =r(t)$, the lower bound on
$\frac{d_2^2(P,Q) + d_2^2(Q,P)}{ D(P\|Q)}$ in the right side of \eqref{eq: Samson - refined}
follows from the fact that \eqref{eq: kappa - Samson} is monotonically increasing on $(0,1)$,
and monotonically decreasing on $(1, \infty)$. To verify that this is the best possible lower
bound, we recall Remark~\ref{remark: tight constants} since in this case $f'(1) = g'(1)=0$.
\end{proof}

\subsection{Bounds on $\frac{D(P\|Q)}{\mathscr{H}_{\alpha}(P\|Q)}$}
\label{subsec:Bounds RE HD}

The following result bounds the ratio of relative entropy to Hellinger
divergence of an arbitrary positive order $\alpha \neq 1$.
Theorem~\ref{thm: improved HausslerO} extends and strengthens a result
for $\alpha \in (0,1)$ by Haussler and Opper \cite[Lemma~4 and (6)]{HausslerO97}
(see also \cite{ChenGZ_arXiv14}), which in turn generalizes the special case for
$\alpha = \tfrac12$  obtained simultaneously and independently by Birg\'{e}
and Massart \cite[(7.6)]{BirgeM98}.

\begin{theorem}  \label{thm: improved HausslerO}
Let $P \ll Q$, $P \neq Q$, $\alpha \in (0,1) \cup (1, \infty)$ and $(\beta_1, \beta_2) \in [0,1)^2$.
Define the continuous function on $[0, \infty]$:
\begin{align}  \label{eq: kappa RE/HD}
\kappa_{\alpha}(t) =
\left\{
\begin{array}{ll}
\log e & t = 0;\\[0.1cm]
\frac{(1-\alpha) \, r(t)}{1-t^{\alpha}+\alpha t - \alpha}& t \in (0,1) \cup (1, \infty);\\[0.1cm]
\alpha^{-1} \log e & t=1;\\
\infty &t = \infty ~\mbox{and}~ \alpha \in (0,1);\\
0 &t = \infty ~\mbox{and}~ \alpha \in (1, \infty).
\end{array}
\right.
\end{align}
Then, for $\alpha \in (0,1)$,
\begin{align}  \label{eq: s1HausslerO}
\kappa_{\alpha}(\beta_2) \leq \frac{D(P\|Q)}{\mathscr{H}_{\alpha}(P\|Q)}
\leq \kappa_{\alpha}(\beta_1^{-1})
\end{align}
and, for $\alpha \in (1, \infty)$,
\begin{align}  \label{eq: s2HausslerO}
\kappa_{\alpha}(\beta_1^{-1}) \leq \frac{D(P\|Q)}{\mathscr{H}_{\alpha}(P\|Q)}
\leq \kappa_{\alpha}(\beta_2).
\end{align}
\end{theorem}

\begin{proof}

$D_r(P\|Q) = D(P \| Q)$, and $ D_{g_\alpha}(P\|Q) = \mathscr{H}_{\alpha}(P\|Q)$ with
\begin{align} \label{eq: s3HausslerO}
g_{\alpha}(t) = \frac{1-t^{\alpha}+\alpha t - \alpha}{1-\alpha}, \quad t \in (0, \infty)
\end{align}
in view of Remark~\ref{remark: equivalence-fD}, \eqref{eq: r-divergence}
and \eqref{eq: Hel-divergence}.
Since $g'_{\alpha}(t) = \frac{\alpha (1-t^{\alpha-1})}{1-\alpha}$
 $g_{\alpha}$ is monotonically
decreasing on $(0,1]$ and monotonically increasing on $[1, \infty)$;
hence, $g_{\alpha}(1)=0$ implies that $g_{\alpha}$ is positive except
at $t=1$. The convexity conditions required by Theorem~\ref{thm: fD1}
are also easy to check for both $r(\cdot)$ and $g_{\alpha} (\cdot)$.
The function in \eqref{eq: kappa RE/HD}
is the continuous extension of $\frac{r}{g_{\alpha}}$, which as
shown in Appendix~\ref{appendix: monotonicity of kappa_alpha},
is monotonically increasing on $[0, \infty]$ if $\alpha \in (0,1)$,
and it is monotonically decreasing on $[0, \infty]$ if $\alpha \in (1,\infty)$.
Therefore, Theorem~\ref{thm: fD1} results in \eqref{eq: s1HausslerO} and
\eqref{eq: s2HausslerO}  for $\alpha \in (0,1)$ and $\alpha \in (1, \infty)$, respectively.
\end{proof}

\begin{remark}
Theorem~\ref{thm: improved HausslerO} is  of particular interest for $\alpha = \frac12$.
In this case since, from \eqref{eq: kappa RE/HD},
$\kappa_{\frac12} \colon [0, \infty] \to [0, \infty]$ is monotonically increasing
with $\kappa_{\frac12}(0) = \log e$, the left inequality in \eqref{eq: s2HausslerO} yields
\eqref{eq: HoeffdingW58}.

For large arguments, $\kappa_{\frac12}(\cdot)$ grows logarithmically. For example, if
$\beta_1 > 9.56 \cdot 10^{-9}$, it follows from \eqref{eq: s1HausslerO} that
\begin{align}
D (P \| Q) \leq \left(14 + \frac{2 \log_e 2}{(1-e^{-1})^2}\right) \, \mathscr{H}_{\frac12} (P \| Q) \; \; \mbox{nats}
\end{align}
which is strictly smaller than the upper bound on the relative entropy in \cite[Theorem~5]{wongshen95},
given not in terms of $\beta_1$ but in terms of another more cumbersome quantity that controls the mass
that $\frac{\text{d}P}{\text{d}Q}$ may have at large values.
\end{remark}

As mentioned in Section~\ref{section:fD}, $\chi^2(P\|Q)$ is equal to the Hellinger divergence
of order~2. Specializing Theorem~\ref{thm: improved HausslerO} to the case $\alpha = 2$ results in
\begin{align} \label{eq: RE and chi-square}
\kappa_2(\beta_1^{-1}) \leq \frac{D(P \| Q)}{\chi^2(P \| Q)} \leq \kappa_2(\beta_2),
\end{align}
which improves the upper and lower bounds in \cite[Proposition~2]{Dragomir00a}:
\begin{align} \label{eq: dragomir}
\tfrac12 \, \beta_1 \, \log e \leq \frac{D(P \| Q)}{\chi^2(P \| Q)} \leq \tfrac12 \, \beta_2^{-1} \, \log e.
\end{align}
For example, if $\beta_1 = \beta_2 = \tfrac1{100}$, \eqref{eq: RE and chi-square} gives
a possible range $[0.037, 0.9631]$ nats for the ratio of relative entropy to $\chi^2$-divergence,
while \eqref{eq: dragomir} gives a range of $[0.005, 50]$ nats.
Note that if $\beta_2=0$, then the upper bound in \eqref{eq: RE and chi-square} is
$\kappa_2(0) = \log e$ whereas it is $\infty$ according to \cite[Proposition~2]{Dragomir00a}.
In view of Remark~\ref{remark: tight constants}, the bounds in \eqref{eq: RE and chi-square}
are the best possible among all probability measures $P, Q$ with given $(\beta_1, \beta_2) \in [0,1)^2$.

\subsection{Bounds on $\frac{\mathrm{JS}(P\|Q)}{\Delta(P\|Q)}$,
$\frac{\mathscr{H}_{\alpha}(P\|Q)}{\mathrm{JS}(P\|Q)}$,
$\frac{\Delta(P\|Q)}{|P-Q|}$}
\label{subsec:C/Delta}
Let $P \ll Q$ and $P \neq Q$. \cite[Theorem~2]{Topsoe_IT00}
and \cite[(11)]{Topsoe_IT00} show, respectively,
\begin{align} \label{eq: Topsoe_IT00}
\tfrac12 \log e & \leq \frac{\mathrm{JS}(P\|Q)}{\Delta(P\|Q)} \leq \log 2, \\
\label{eq: LL2}
\tfrac12 \, |P-Q|^2 & \leq \Delta(P\|Q) \leq |P-Q|.
\end{align}

The following result suggests a refinement in the upper bounds of
\eqref{eq: Topsoe_IT00} and \eqref{eq: LL2}.
\begin{theorem} \label{thm: superTopsoeLL}
Let $P \ll Q$, $P \neq Q$, and let $\beta_1, \beta_2 \in [0,1)$. Then,
\begin{align}  \label{eq1: superTopsoe}
\tfrac12 \log e \leq \frac{\mathrm{JS}(P\|Q)}{\Delta(P\|Q)} & \leq
\max \bigl\{\kappa_1(\beta_1^{-1}), \, \kappa_1(\beta_2) \bigr\}, \\[0.1cm]
\label{eq1: superLL2}
\frac{\Delta(P\|Q)}{|P-Q|} & \leq
\max \bigl\{\kappa_2(\beta_1^{-1}), \, \kappa_2(\beta_2) \bigr\}
\end{align}
with $\kappa_1 \colon [0, \infty] \to [0, \log 2]$ and
$\kappa_2 \colon [0, \infty] \to [0, 1]$ defined by
\begin{align}  \label{eq2: superTopsoe}
\kappa_1(t) =
\left\{
\begin{array}{ll}
\log 2 & t = 0;\\[0.1cm]
\frac{t+1}{(t-1)^2} \, \left[t \log t
- (t+1) \log\left(\frac{t+1}{2}\right) \right]
& t \in (0,1) \cup (1, \infty);\\[0.1cm]
\tfrac12 \log e & t=1;\\
\log 2 &t = \infty;
\end{array}
\right.
\end{align}
and
\begin{align}
\label{eq2: superLL2}
\hspace*{-5.4cm}
\kappa_2(t) =
\left\{
\begin{array}{ll}
\frac{|t-1|}{t+1}
& \hspace*{0.5cm} t \in [0,\infty);\\[0.1cm]
1 & \hspace*{0.5cm} t = \infty.
\end{array}
\right.
\end{align}
\end{theorem}

\begin{proof}
Let $f_{\mathrm{TV}}$, $f_\Delta, f_{\mathrm{JS}}$ denote the
functions $f \colon (0, \infty) \to \Reals$ in \eqref{eq: TV as fD},
\eqref{eq:tridiv} and \eqref{eq:js2}, respectively;
these functions yield $|P-Q|$, $\Delta(P\|Q)$ and
$\mathrm{JS}(P\|Q)$ as $f$-divergences. The functions $\kappa_1$
and $\kappa_2$, as introduced in \eqref{eq2: superTopsoe} and \eqref{eq2: superLL2},
respectively, are the continuous extensions to $[0, \infty]$ of
\begin{align}  \label{eq3: superTopsoe}
& \kappa_1(t) = \frac{f_{\mathrm{JS}}(t) +
(t-1) \; \log 2}{f_{\Delta}(t)}, \\
\label{eq3: superLL2}
& \kappa_2(t) = \frac{f_\Delta(t)}{f_{\mathrm{TV}}(t)}.
\end{align}
It can be verified by \eqref{eq2: superTopsoe} and
\eqref{eq2: superLL2} that $\kappa_i(t) = \kappa_i\left(\frac1{t}\right)$
for $t \in [0, \infty]$ and $i \in \{1, 2\}$; furthermore, $\kappa_1$
and $\kappa_2$ are both monotonically
decreasing on $[0,1]$, and monotonically increasing on
$[1, \infty]$. Consequently, if $\beta \in [\beta_2, \beta_1^{-1}]$
and $i \in \{1, 2\}$,
\begin{align}  \label{eq3b: superTopsoe-LL}
& \kappa_i(1) \leq \kappa_i(\beta) \leq \max\bigl\{\kappa_i(\beta_1^{-1}),
\kappa_i(\beta_2) \bigr\}.
\end{align}
In view of Theorem~\ref{thm: fD1} and Remark~\ref{remark: equivalence-fD},
\eqref{eq1: superTopsoe} follows from
\eqref{eq3: superTopsoe} and \eqref{eq3b: superTopsoe-LL};
\eqref{eq1: superLL2} follows from \eqref{eq3: superLL2} and
\eqref{eq3b: superTopsoe-LL}.
\end{proof}

If $\beta_1=0$, referring to an unbounded relative information
$\imath_{P\|Q}$, the right inequalities in \eqref{eq: Topsoe_IT00}, \eqref{eq: LL2}
and \eqref{eq1: superTopsoe}, \eqref{eq1: superLL2} respectively coincide (due to
\eqref{eq2: superTopsoe}, \eqref{eq2: superLL2}, and since
$\kappa_2(0)=1$); otherwise, the right inequalities in \eqref{eq1: superTopsoe},
\eqref{eq1: superLL2} provide, respectively, sharper upper bounds than those in
\eqref{eq: Topsoe_IT00}, \eqref{eq: LL2}.

\begin{example}  \label{example: superTopsoe}
Suppose that $P \ll Q$, $P \neq Q$ and $\tfrac1{2} \leq \frac{\text{d}P}{\text{d}Q} \, (a) \leq 2$
for all $a \in \set{A}$ (e.g., $P = \bigl(\tfrac12, \tfrac12 \bigr)$ and $Q = \bigl(\tfrac14, \tfrac34\bigr)$);
substituting $\beta_1 = \beta_2 = \frac1{2}$ in
\eqref{eq1: superTopsoe}--\eqref{eq2: superLL2} gives
\begin{align} \label{eq4: superTopsoe}
\tfrac12 \log e \leq  \frac{\mathrm{JS}(P\|Q)}{\Delta(P\|Q)} & \leq 0.510 \, \log e \\[0.1cm]
\label{eq4: superLL2}
\frac{\Delta(P\|Q)}{|P-Q|} & \leq \tfrac13
\end{align}
improving the upper bounds on $\frac{\mathrm{JS}(P\|Q)}{\Delta(P\|Q)}$
and $\frac{\Delta(P\|Q)}{|P-Q|}$ which are $\log 2 \approx 0.693 \, \log e$ and~1,
respectively, according to the right inequalities in \eqref{eq: Topsoe_IT00},
\eqref{eq: LL2}.
\end{example}

For finite alphabets, \cite[Theorem~7]{LL2015} shows
\begin{align} \label{eq: LL1}
& \log 2 \leq \frac{\mathrm{JS}(P\|Q)}{\mathscr{H}_{\frac12}(P\|Q)} \leq \log e.
\end{align}
The following theorem extends the validity of \eqref{eq: LL1} for Hellinger
divergences of an arbitrary order $\alpha \in (0, \infty)$ and for a general alphabet,
while also providing a refinement of both upper and lower bounds in \eqref{eq: LL1}
when the relative information $\imath_{P\|Q}$ is bounded.
\begin{theorem} \label{thm: supperLL}
Let $P \ll Q$, $P \neq Q$, $\alpha \in (0,1) \cup (1,\infty)$, and let
$\beta_1, \beta_2 \in [0,1]$ be given in \eqref{eq: beta1}-\eqref{eq: beta2}. Let
$\kappa_{\alpha}\colon [0,\infty] \to [0, \infty)$ be given by
\begin{align}  \label{eq1: kappaJSHel}
\kappa_{\alpha}(t) =
\left\{
\begin{array}{ll}
\log 2 & t = 0;\\[0.1cm]
\frac{(\alpha-1) \Bigl[t \log t - (t+1)
\log\left(\frac{t+1}{2}\right) \Bigr]}{t^{\alpha}+\alpha-1-\alpha t }
& t \in (0,1) \cup (1, \infty);\\[0.1cm]
\frac{\log e}{2\alpha} & t=1;\\
\left( \frac1{\alpha}-1 \right)^+ \, \log 2 &t = \infty.
\end{array}
\right.
\end{align}
Then, if $\alpha \in (0,1)$,
\begin{align} \label{eq1a: superLL1}
\min \bigl\{\kappa_{\alpha}(\beta_1^{-1}), \kappa_{\alpha}(\beta_2) \bigr\} \leq
\frac{\mathrm{JS}(P\|Q)}{\mathscr{H}_{\alpha}(P\|Q)} \leq
\max_{\beta \in [\beta_2, \beta_1^{-1}]} \kappa_{\alpha}(\beta)
\end{align}
and, if $\alpha \in (1, \infty)$,
\begin{align} \label{eq1b: superLL1}
\kappa_{\alpha}(\beta_1^{-1}) \leq \frac{\mathrm{JS}(P\|Q)}{\mathscr{H}_{\alpha}(P\|Q)} \leq \kappa_{\alpha}(\beta_2).
\end{align}
\end{theorem}

\begin{proof}
Let $f_\alpha$ and $f_{\mathrm{JS}}$ denote, respectively, the functions
$f \colon (0, \infty) \to \Reals$ in \eqref{eq: H as fD} and \eqref{eq:js2}
which yield the Hellinger divergence, $\mathscr{H}_{\alpha}$,
and the Jensen-Shannon divergence, $\mathrm{JS}(P\|Q)$, as $f$-divergences.
From \eqref{eq1: kappaJSHel},
$\kappa_{\alpha}$ is the continuous extension to $[0, \infty]$ of
\begin{align} \label{eq2: kappaJSHel}
\kappa_\alpha(t) = \frac{f_{\mathrm{JS}}(t) + (t-1) \, \log 2}{f_{\alpha}(t)
+ \bigl(\frac{\alpha}{1-\alpha}\bigr) (t-1)}.
\end{align}
As shown in the proof of Theorem~\ref{thm: improved HausslerO}, for every
$\alpha \in (0,1) \cup (1, \infty)$ and $t \in (0,1) \cup (1, \infty)$, we have
$g_{\alpha}(t) \triangleq f_{\alpha}(t) + \left(\frac{\alpha}{1-\alpha}\right) (t-1) > 0$.
It can be verified that $\kappa_{\alpha} \colon [0, \infty] \to \Reals$ has the
following monotonicity properties:
\begin{itemize}
\item if $\alpha \in (0,1)$, there exists $t_{\alpha} > 0$ such that $\kappa_\alpha$
is monotonically increasing on $[0,t_{\alpha}]$,
and it is monotonically decreasing on $[t_{\alpha}, \infty]$;
\item if $\alpha \in (1, \infty)$, $\kappa_{\alpha}$ is monotonically decreasing
on $[0, \infty]$.
\end{itemize}
Based on these properties of $\kappa_\alpha$, for $\beta \in [\beta_2, \beta_1^{-1}]$:
\begin{itemize}
\item if $\alpha \in (0,1)$
\begin{align} \label{eq3: kappaJSHel}
\min \bigl\{\kappa_{\alpha}(\beta_1^{-1}), \kappa_{\alpha}(\beta_2) \bigr\} \leq
\kappa_{\alpha}(\beta) \leq \max_{t \in [\beta_2, \beta_1^{-1}]} \kappa_{\alpha}(t);
\end{align}
\item if $\alpha \in (1,\infty)$,
\begin{align} \label{eq4: kappaJSHel}
\kappa_{\alpha}(\beta_1^{-1}) \leq \kappa_{\alpha}(\beta) \leq \kappa_{\alpha}(\beta_2).
\end{align}
\end{itemize}
In view of Theorem~\ref{thm: fD1}, Remark~\ref{remark: equivalence-fD} and
\eqref{eq2: kappaJSHel}, the bounds in \eqref{eq1a: superLL1} and \eqref{eq1b: superLL1}
follow respectively from \eqref{eq3: kappaJSHel} and \eqref{eq4: kappaJSHel}.
\end{proof}

\begin{example}
Specializing Theorem~\ref{thm: supperLL} to $\alpha=\tfrac12$, since
$\kappa_{\frac12}(t) = \kappa_{\frac12}\bigl(\frac1{t}\bigr)$ for $t \in (0, \infty)$
and $\kappa_{\frac12}$ achieves its global maximum at $t=1$ with $\kappa_{\frac12}(1)=\log e$,
\eqref{eq1a: superLL1} implies that
\begin{align} \label{eq1: superLL12}
\log 2 \leq \kappa_{\frac12}\Bigl(\min \{\beta_1^{-1}, \beta_2 \} \Bigr) \leq
\frac{\mathrm{JS}(P\|Q)}{\mathscr{H}_{\frac12}(P\|Q)} \leq \log e.
\end{align}
Under the assumptions in Example~\ref{example: superTopsoe}, it follows
from \eqref{eq1: superLL12} that
\begin{align}  \label{eq2: superLL12}
0.990 \, \log e \leq \frac{\mathrm{JS}(P\|Q)}{\mathscr{H}_{\frac12}(P\|Q)} \leq \log e
\end{align}
which improves the lower bound $\log 2 \approx 0.693 \, \log e$ in the left side of \eqref{eq: LL1}.

Specializing Theorem~\ref{thm: supperLL} to $\alpha=2$ implies that (cf. \eqref{eq1: kappaJSHel}
and \eqref{eq1b: superLL1})
\begin{align} \label{eq1: superLLchi2}
\kappa_2(\beta_1^{-1}) \leq \frac{\mathrm{JS}(P\|Q)}{\chi^2(P\|Q)} \leq \kappa_2(\beta_2).
\end{align}
Under the assumptions in Example~\ref{example: superTopsoe}, it follows
from \eqref{eq1: superLLchi2} that
\begin{align} \label{eq2: superLLchi2}
0.170 \, \log e \leq \frac{\mathrm{JS}(P\|Q)}{\chi^2(P\|Q)} \leq 0.340 \, \log e
\end{align}
while without any boundedness assumption, \eqref{eq1b: superLL1} yields the weaker upper
and lower bounds
\begin{align} \label{eq3: superLLchi2}
0 \leq \frac{\mathrm{JS}(P\|Q)}{\chi^2(P\|Q)} \leq \log 2 \approx 0.693 \, \log e.
\end{align}
\end{example}

\subsection{Local Behavior of $f$-Divergences}
\label{sec:On the Local Behavior of f-divergences}
Another application of Theorem~\ref{thm: fD1} shows that the local behavior of $f$-divergences
differs by only a constant, provided that the first distribution approaches the reference measure
in a certain strong sense.

\begin{theorem} \label{thm: local behavior fD}
Suppose that $\{P_n\}$, a sequence of probability measures defined on a measurable space
$(\set{A}, \mathscr{F})$, converges to $Q$ (another probability measure on the same space)
in the sense that, for $Y \sim Q$,
\begin{align}
\label{eq: 1st condition}
\lim_{n \to \infty} \esssup \frac{\text{d}P_n}{\text{d}Q} \, (Y) = 1
\end{align}
where it is assumed that $P_n \ll Q$ for all sufficiently large $n$.
If $f$ and $g$ are convex on $(0, \infty)$ and they are positive except at $t=1$
(where they are 0), then
\begin{align}  \label{eq: limit fD, Pn-->Q}
\lim_{n \to \infty} D_f(P_n \| Q) = \lim_{n \to \infty} D_g(P_n \| Q) = 0,
\end{align}
and
\begin{align}
\min\{ \kappa ( 1^-) , \kappa (1^+) \} \leq \lim_{n \to \infty}
\frac{D_f(P_n \| Q)}{D_g(P_n \| Q)} \leq \max\{ \kappa ( 1^-) , \kappa (1^+) \}
\label{eq: limit of ratio of f-divergences}
\end{align}
where we have indicated the left and right limits of the function $\kappa (\cdot)$,
defined in \eqref{kappadef-1}, at $1$ by $\kappa ( 1^-)$ and $\kappa (1^+)$, respectively.
\end{theorem}

\begin{proof}
Since $f(1)=0$,
\begin{align} \label{ubDf}
0 \leq D_f(P_n \| Q) & = \int  f\left(\frac{\text{d}P_n}{\text{d}Q}\right) \, \text{d}Q  \\
\label{eq2: ubDf}
& \leq \sup_{\beta \in [\beta_{2,n}, \, \beta_{1,n}^{-1}]} f(\beta)
\end{align}
where we have abbreviated
\begin{align}
\label{eq: beta1n}
\beta_{1,n}^{-1} &\triangleq \esssup \frac{\text{d}P_n}{\text{d}Q} \, (Y), \\
\label{eq: beta2n}
\beta_{2,n} &\triangleq \essinf \frac{\text{d}P_n}{\text{d}Q} \, (Y).
\end{align}
The condition in \eqref{eq: 1st condition} yields
\begin{align}
\label{eq2: 1st condition}
& \lim_{n \to \infty} \beta_{1,n} = 1, \\
\label{eq: 2nd condition}
& \lim_{n \to \infty} \beta_{2,n} = 1
\end{align}
where \eqref{eq2: 1st condition} is a restatement of \eqref{eq: 1st condition}
(see the notation in \eqref{eq: beta1n}), and
Appendix~\ref{appendix: proof of 2nd condition} justifies \eqref{eq: 2nd condition}.
Hence, \eqref{eq: limit fD, Pn-->Q} follows from
\eqref{eq2: ubDf}, \eqref{eq2: 1st condition}, \eqref{eq: 2nd condition},
the continuity of $f$ at~1 (due to its convexity).

Abbreviating $I_n = [\beta_{2,n}, 1) \cup (1,  \beta_{1,n}^{-1}] $,
\eqref{eq:fD bound2} and \eqref{eq:fD bound1} result in
\begin{align}
\inf_{\beta \in I_n} \kappa(\beta)  D_g(P_n \| Q)
\leq D_f(P_n \| Q)
\leq \sup_{\beta \in I_n} \kappa(\beta)  D_g(P_n \| Q).
\end{align}
The right and left continuity of $\kappa (\cdot)$ at 1 together with
\eqref{eq2: 1st condition} and \eqref{eq: 2nd condition} imply that
\begin{align}
& \inf_{\beta \in I_n} \kappa(\beta) \to \min\{ \kappa ( 1^-) , \kappa (1^+) \}, \\
& \sup_{\beta \in I_n} \kappa(\beta) \to \max\{ \kappa ( 1^-) , \kappa (1^+) \}
\end{align}
by letting $n \to \infty$.
\end{proof}
\begin{corollary}\label{cor:ratios}
Let $\{P_n \ll Q \}$ converge to $Q$ in the sense of  \eqref{eq: 1st condition}.
Then,  $D(P_n \| Q)$ and $D(Q \| P_n)$ vanish
as $n \to \infty$ with
\begin{align} \label{eq: limit of RE/dual}
\lim_{n \to \infty} \frac{D(P_n \| Q)}{D(Q \| P_n)} &= 1.
\end{align}
\end{corollary}

\begin{corollary}\label{cor:ratios:chi}
Let $\{P_n \ll Q \}$ converge to $Q$ in the sense of  \eqref{eq: 1st condition}.
Then, $\chi^2(P_n \| Q)$ and $D(P_n \| Q)$ vanish as $n \to \infty$ with
\begin{align}\label{eq:cor:ratios:chi}
\lim_{n \to \infty} \frac{D(P_n \| Q)}{\chi^2(P_n \| Q)} &= \tfrac12 \log e.
\end{align}
\end{corollary}
Note that \eqref{eq:cor:ratios:chi} is known in the  finite alphabet case \cite[Theorem~4.1]{CsiszarS_FnT}).
\par
In Example \ref{example:counter}, the ratio in \eqref{eq: limit of ratio of f-divergences} is equal to $\tfrac12$, while the
lower and upper bounds are $\frac13$ and $1$, respectively.

Continuing with Examples~\ref{example: two Laplacians}, \ref{example: two Cramer distributions}
and~\ref{example: two Cauchy distributions}, it is easy to check that
\eqref{eq: 1st condition} is satisfied in the following cases.

\begin{example}
\label{example: convergent Laplacians}
A sequence of Laplacian probability density functions with common variance and converging means:
\begin{align}
p_n(x) &= \frac{\lambda}{2} \cdot \exp \bigl(-\lambda |x - a_n| \bigr) \\
\lim_{n \rightarrow \infty} a_n &= a .
\end{align}
\end{example}

\begin{example}
A sequence of converging Cram\'{e}r probability density functions:
\label{example: convergent Cramer pdfs}
\begin{align}
\label{eq: convergent Cramer pdfs}
& p_n(x) = \frac{\theta_n}{2\bigl(1+\theta_n |x-m_n|\bigr)^2}, \quad x \in \Reals \\
& \lim_{n \to \infty} m_n = m \in \Reals \\
& \lim_{n \to \infty} \theta_n = \theta > 0 .
\end{align}
\end{example}

\begin{example}
\label{example: convergent Cauchy pdfs}
A sequence of converging Cauchy probability density functions:
\begin{align}
\label{eq: convergent Cauchy pdfs}
& p_n(x) = \frac1{\pi \gamma_n} \left[1 + \left(\frac{x-m_n}{\gamma_n}\right)^2\right]^{-1}, \quad x \in \Reals \\
& \lim_{n \to \infty} m_n = m \in \Reals \\
& \lim_{n \to \infty} \gamma_n = \gamma > 0.
\end{align}
\end{example}

\subsection{Strengthened Jensen's inequality} \label{subsec: superjensen}

Bounding away from zero a certain density between two probability measures
enables  the following
strengthened version of Jensen's inequality, which generalizes a result
in \cite[Theorem~1]{Dragomir06}.

\begin{lemma}\label{lemma: superjensen}
Let $f \colon \Reals \to \Reals$ be a convex function, $\PU \ll \PZ$ be probability measures defined on a measurable space $(\set{A}, \mathscr{F})$, and
fix an arbitrary random transformation $P_{Z|X}\colon \set{A} \to \Reals$.
Denote\footnote{We follow the notation in \cite{Verdu_book} where $\PZ \to P_{Z|X} \to P_{Z_\mathtt{0}}$
means that the marginal probability measures of the joint distribution $\PZ  P_{Z|X}$ are
$\PZ$ and $ P_{Z_\mathtt{0}}$.}  $\PZ \to P_{Z|X} \to P_{Z_\mathtt{0}}$,
and $\PU \to P_{Z|X} \to P_{Z_\mathtt{1}}$.
Then,
\begin{align}
\beta \, \bigl(\mathbb{E} \left[ f ( \mathbb{E} [ Z_\mathtt{0} | X_\mathtt{0} ] ) \right]
- f ( \mathbb{E} [ Z_\mathtt{0} ]  ) \bigr)
\label{eq: lemma-superjensen}
\leq \mathbb{E} [ f ( \mathbb{E} [ Z_\mathtt{1} | X_\mathtt{1} ] ) ]
- f ( \mathbb{E} [ Z_\mathtt{1} ]  )
\end{align}
where $X_0 \sim \PZ$, $X_1 \sim \PU$, and
\begin{align}  \label{eq: beta}
\beta \triangleq \essinf \frac{\mathrm{d}\PU}{\mathrm{d}\PZ} \, (X_0).
\end{align}
\end{lemma}
\begin{proof}
If $\beta=0$, the claimed result is Jensen's inequality, while if $\beta=1$,  $\PZ = \PU$ and the result is trivial.
Hence, we assume $\beta \in (0,1)$.
Note that $\PU = \beta \PZ + (1 - \beta) P_\mathtt{2}$ where $ P_\mathtt{2}$ is the probability measure
whose density with respect to $\PZ$ is given by
\begin{align}
\frac{\mathrm{d} P_\mathtt{2}}{\mathrm{d}\PZ} =
\frac1{1-\beta} \left(\frac{\mathrm{d}\PU}{\mathrm{d}\PZ} - \beta \right) \geq 0.
\end{align}
Letting $P_\mathtt{2} \to P_{Z|X} \to P_ {Z_\mathtt{2}} $, Jensen's inequality implies
\begin{align}
f (\mathbb{E} [ Z_\mathtt{1} ]) &\leq \beta\,  f ( \mathbb{E} [ Z_\mathtt{0} ] )
+ (1 - \beta ) \, f ( \mathbb{E} [ Z_\mathtt{2} ] ). \label{calfio1}
\end{align}
Furthermore, we can apply  Jensen's inequality again to obtain
\begin{align}
f ( \mathbb{E} [ Z_\mathtt{2} ] ) &=  f ( \mathbb{E}\left[ \mathbb{E} [ Z_\mathtt{2} | X_\mathtt{2} ] \right] )  \\
&\leq \mathbb{E}\left[  f ( \mathbb{E} [ Z_\mathtt{2} | X_\mathtt{2} ] )  \right]  \\
&= \frac{\mathbb{E}\left[  f ( \mathbb{E} [ Z_\mathtt{1} | X_\mathtt{1} ] )  \right] - \beta \;
\mathbb{E}\left[  f ( \mathbb{E} [ Z_\mathtt{0} | X_\mathtt{0} ] )  \right]}{1-\beta}.
\label{calfio2}
\end{align}
Substituting this bound on $f ( \mathbb{E} [ Z_\mathtt{2} ] )$ in \eqref{calfio1} we obtain the desired result.
\end{proof}

\begin{remark}
Letting $Z = X$, and choosing $\PZ$ so that $\beta = 0$ (e.g., $\PU$ is a restriction of $\PZ$
to an event of $\PZ$-probability less than~1), \eqref{eq: lemma-superjensen} becomes Jensen's
inequality $f(\mathbb{E}[X_\mathtt{1}]) \leq \mathbb{E}[f(X_\mathtt{1})]$.
\end{remark}

Lemma~\ref{lemma: superjensen} finds the following application to the derivation of $f$-divergence inequalities.

\begin{theorem} \label{thm: GI fD}
Let $f \colon (0, \infty) \to \Reals$ be a convex function with $f(1)=0$.
Fix $P \ll Q$ on the same space with $(\beta_1, \beta_2) \in [0,1)^2$ and let $X \sim P$. Then,
\begin{align}
\label{juventus}
\beta_2 \, D_f(P\|Q) &\leq \mathbb{E} \left[ f \left( \exp ( \imath_{P\|Q} (X) )
\right) \right]  - f\bigl( 1 + \chi^2(P\|Q) \bigr) \\
\label{eq: GI fD}
&\leq \beta_1^{-1} \, D_f(P \| Q).
\end{align}
\end{theorem}

\begin{proof}
We invoke Lemma~\ref{lemma: superjensen}
with $P_{Z|X}$ that is given by the deterministic transformation
$\exp \left(\imath_{P\|Q} (\cdot) \right) \colon \set{A} \to \Reals$.
Then, $\mathbb{E} [ Z_\mathtt{0} | X_\mathtt{0} ] = \exp \left(\imath_{P\|Q} ( X_\mathtt{0} )\right) $.
If, moreover, we let $X_\mathtt{0} \sim Q = \PZ$, we obtain
\begin{align}
\mathbb{E} [ Z_\mathtt{0} ]  &= 1, \\
\mathbb{E} \left[ f ( \mathbb{E} [ Z_\mathtt{0} | X_\mathtt{0} ] ) \right] &= D_f ( P \| Q )
\end{align}
and if we let $X_\mathtt{1} \sim P = \PU$, we have (see \eqref{eq: chi-square 4})
\begin{align}
\mathbb{E} [ Z_\mathtt{1} ]  &= 1 + \chi^2 ( P \| Q ), \\
\mathbb{E} \left[ f ( \mathbb{E} [ Z_\mathtt{1} | X_\mathtt{1} ] ) \right] &=
\mathbb{E} \left[ f \left( \exp ( \imath_{P\|Q} (X) )\right) \right].
\end{align}
Therefore, \eqref{juventus} follows from Lemma~\ref{lemma: superjensen}.
Recalling \eqref{eq: beta1-alt},
inequality  \eqref{eq: GI fD} follows from
Lemma~\ref{lemma: superjensen} as well switching the roles $\PZ$ and $\PU$, namely, now we take  $P =\PZ $
and $Q = \PU$.
\end{proof}
Specializing Theorem~\ref{thm: GI fD} to the convex function on $(0, \infty)$ where
$f(t) = -\log t$ sharpens inequality~\eqref{grout425} under the assumption of bounded
relative information.

\begin{theorem} \label{thm:d-chi}
Fix $P \ll \gg Q$ such that $(\beta_1, \beta_2) \in (0,1)^2$. Then,
\begin{align}
\label{lbchi}
\beta_2 \, D(Q\|P) &\leq \log \bigl(1 + \chi^2(P\|Q) \bigr) - D(P\|Q) \\
\label{ubchi}
& \leq \beta_1^{-1} \, D(Q\|P).
\end{align}
\end{theorem}

\eject
\section{Total Variation Distance, Relative Information Spectrum and Relative Entropy}
\label{sec: TV-RI}
\subsection{Exact Expressions}
\label{subsec: TV-RI}

The following result provides several useful expressions of
the total variation distance in terms of the relative information.

\begin{theorem}  \label{thm: tv}
Let $P \ll Q$, and let $X \sim P$ and $Y \sim Q$ be defined on a measurable
space $(\set{A}, \mathscr{F})$. Then,\footnote{
$ (z)^+ \triangleq z \, 1\{z > 0\} = \max\{z, 0\}$, and
$ (z)^- \triangleq -z \, 1\{z < 0\} = \max\{-z, 0\}.$}
\begin{align} \label{eq: TV1}
|P-Q| & = \mathbb{E} \bigl[ \bigl| 1 - \exp(\imath_{P\|Q}(Y)) \bigr| \bigr] \\
\label{eq: TV2}
& = 2 \, \mathbb{E} \bigl[ \bigl( 1 - \exp(\imath_{P\|Q}(Y)) \bigr)^+ \bigr] \\
\label{eq: TV3}
& = 2 \, \mathbb{E} \bigl[ \bigl( 1 - \exp(\imath_{P\|Q}(Y)) \bigr)^- \bigr] \\
\label{eq: TV4}
& = 2 \, \mathbb{E} \bigl[ \bigl( 1 - \exp(-\imath_{P\|Q}(X)) \bigr)^+ \bigr] \\[0.1cm]
\label{eq: TV4a}
& = 2 \left( \mathbb{P}\bigl[\imath_{P\|Q}(X) > 0\bigr] - \mathbb{P}\bigl[\imath_{P\|Q}(Y) > 0\bigr] \right) \\
\label{eq: TV4b}
& = 2 \left( \mathbb{P}\bigl[\imath_{P\|Q}(Y) \leq 0\bigr] - \mathbb{P}\bigl[\imath_{P\|Q}(X) \leq 0\bigr] \right) \\
\label{eq: TV5}
& = 2 \, \int_0^1 \mathbb{P}\bigl[\imath_{P \| Q}(Y) < \log \beta\bigr] \, \text{d}\beta \\
\label{eq: TV6}
& = 2 \, \int_0^1 \mathbb{P}\Bigl[\imath_{P \| Q}(X) > \log \frac1{\beta} \Bigr] \, \text{d}\beta \\
\label{eq: TV100}
& = 2 \, \int_1^{\beta_1^{-1}} \beta^{-2} \left[ 1 -\mathds{F}_{P\|Q}(\log \beta) \right] \, \text{d}\beta.
\end{align}
Furthermore, if $P \ll \gg Q$, then
\begin{align}
\label{eq: TV7}
|P-Q| & = 2 \, \mathbb{E} \bigl[ \bigl( 1 - \exp(-\imath_{P\|Q}(X)) \bigr)^- \bigr] \\
\label{eq: TV8}
& = \mathbb{E} \bigl[ \bigl| 1 - \exp(-\imath_{P\|Q}(X)) \bigr| \bigr].
\end{align}
\end{theorem}

\begin{proof}
See Appendix~\ref{appendix: tv}.
\end{proof}

\begin{remark}
In view of \eqref{eq: TV4a}, if $P \ll Q$, the supremum in \eqref{eq3: TV distance}
is a maximum which is achieved by the event
\begin{align}
\label{eq: max-achieving TV}
\set{F}^* = \bigl\{ a \in \set{A} \colon \imath_{P \| Q}(a) > 0 \bigr\} \in \mathscr{F}.
\end{align}
\end{remark}

\par
Similarly to Theorem~\ref{thm: tv}, the following theorem provides several expressions of
the relative entropy in terms of the relative information spectrum.

\begin{theorem} \label{theorem: relative entropy - exact expressions}
If $D(P\|Q) < \infty$, then
\begin{align}
D (P \| Q )
\label{eq1: RE}
&= \int_0^\infty \left(1 - \mathds{F}_{P\|Q} (\alpha) \right) \,
\mathrm{d} \alpha - \int_{-\infty}^0  \mathds{F}_{P\|Q} (\alpha)  \, \mathrm{d} \alpha \\[0.1cm]
\label{eq2: RE}
&= \int_1^\infty \frac{1 - \mathds{F}_{P\|Q}(\log \beta)}{\beta} \, \mathrm{d}\beta
- \int_0^1 \frac{\mathds{F}_{P\|Q}(\log \beta)}{\beta} \, \mathrm{d}\beta \\[0.1cm]
\label{eq3: RE}
&= \int_0^\infty \mathbb{P}\bigl[\imath_{P\|Q}(Y) > \alpha \bigr] \, \alpha e^{\alpha} \, \mathrm{d}\alpha
- \int_{-\infty}^0 \mathbb{P}\bigl[\imath_{P\|Q}(Y) < \alpha \bigr] \, \alpha e^{\alpha} \, \mathrm{d}\alpha
\end{align}
where $Y \sim Q$, and for convenience \eqref{eq2: RE} and \eqref{eq3: RE} assume that the relative
information and the resulting relative entropy are in nats.
\end{theorem}

\begin{proof}
The expectation of a real-valued random variable $V$ is equal to
\begin{align} \label{eq: expectation}
\mathbb{E}[V] = \int_0^\infty \mathbb{P}[V > t] \, \mathrm{d}t
- \int_{-\infty}^0 \mathbb{P}[V < t] \, \mathrm{d}t
\end{align}
where we are free to substitute $>$ by $\geq$, and $<$ by $\leq$. If we let $V = \imath_{P\|Q}(X)$
with $X \sim P$, then \eqref{eq: expectation} yields \eqref{eq1: RE} provided that
$D(P\|Q) = \mathbb{E}[V] < \infty$.

Eq.~\eqref{eq2: RE} follows from \eqref{eq1: RE} by the substitution $\alpha = \log \beta$
when the relative entropy is expressed in nats.

To prove \eqref{eq3: RE}, let $Z = \imath_{P \| Q}(Y)$ with $Y \sim Q$, and let $V = r(Z)$ where
$r \colon (0, \infty) \to [0, \infty)$ is given in \eqref{eq: r} with natural logarithm.
The function $r$ is strictly monotonically increasing on $[1, \infty)$, on which interval
we define its inverse by $s_1 \colon [0, \infty) \to [1, \infty)$; it is also strictly monotonically
decreasing on $(0,1]$, on which interval we define its inverse by $s_2 \colon [0,1] \to (0,1]$. Then,
only the first integral on the right side of \eqref{eq: expectation} can be non-zero, and we decompose
it as
\begin{align}
D(P\|Q)
& = \int_0^{\infty} \hspace*{-0.05cm} \mathbb{P}\bigl[Z \geq 1, r(Z) > t\bigr] \mathrm{d}t
+ \int_0^{\infty} \hspace*{-0.05cm} \mathbb{P}\bigl[Z < 1, r(Z) > t\bigr] \mathrm{d}t \nonumber\\
& = \int_0^{\infty} \mathbb{P}\bigl[Z > s_1(t)\bigr] \, \mathrm{d}t +
\int_0^1 \mathbb{P}\bigl[Z < s_2(t) \bigr] \, \mathrm{d}t \\
\label{eq: change of var}
& = \int_1^{\infty} \mathbb{P}[Z > v] \,
\log_e v \, \mathrm{d}v - \int_0^1 \mathbb{P}[Z < v] \, \log_e v \, \mathrm{d}v
\end{align}
where \eqref{eq: change of var} follows from the change of variable of integration $t = r(v)$. Upon taking
$\log_e$ on both sides of the inequalities inside the probabilities in \eqref{eq: change of var}, and a further
change of the variable of integration $v = e^{\alpha}$, \eqref{eq: change of var} is seen to be
equal to \eqref{eq3: RE}.
\end{proof}

\subsection{Upper Bounds on $|P-Q|$} \label{subsec: ub-tv}
In this section, we provide three upper bounds on $|P-Q|$ which
complement \eqref{eq: Pinsker}.

\begin{theorem} \label{thm: ubtv}
If $P \ll Q$ and $X \sim P$, then
\begin{align} \label{eq: ubtv}
|P-Q| \, \log e \leq D(P \| Q) + \mathbb{E} \bigl[|\imath_{P\|Q}(X)| \bigr].
\end{align}
\end{theorem}
\begin{proof}
For every $z \in [-\infty, \infty]$,
\begin{align} \label{eq: SV-ITA 14}
\bigl( 1 - \exp(-z) \bigr) \, 1\{z>0\} \leq \left(\frac{z}{\log e}\right) 1\{z>0\}.
\end{align}
Substituting $z = \imath_{P\|Q}(X)$, taking expectation of both sides of
\eqref{eq: SV-ITA 14}, and using \eqref{eq: TV4} give
\begin{align}
|P-Q| \, \log e & \leq 2 \, \mathbb{E} \bigl[\imath_{P\|Q}(X)
\; 1\{\imath_{P\|Q}(X)>0\} \bigr] \\
& = \mathbb{E} \bigl[ \imath_{P\|Q}(X) + | \imath_{P\|Q}(X) | \bigr] \\
& = D(P \| Q) + \mathbb{E} \bigl[|\imath_{P\|Q}(X)| \bigr].
\end{align}
\end{proof}

\begin{remark}
Theorem~\ref{thm: ubtv} is tighter than Pinsker's bound in
\cite[(2.3.14)]{Pinsker60}:
\begin{align} \label{eq: Pinsker60}
|P-Q| \, \log e \leq 2\, \mathbb{E} \bigl[|\imath_{P\|Q}(X)| \bigr].
\end{align}
\end{remark}

The second upper bound on $|P-Q|$ is a consequence of Theorem~\ref{thm: tv}.
\begin{theorem} \label{thm: ubtv-ita14th5}
Let $P \ll Q$ with $(\beta_1, \beta_2) \in [0,1)^2$.
Then, for every $\beta_0 \in [\beta_1, 1]$,
\begin{align}
\tfrac12 \, |P-Q| \leq & (1-\beta_0) \, \prob[\imath_{P\|Q}(X) > 0] \nonumber \\
& + (\beta_0 - \beta_1) \, \prob\left[\imath_{P\|Q}(X) > \log \tfrac1{\beta_0}\right]
 \label{eq: ubtv1-ita14th5}
\end{align}
where $X \sim P$,
and, for every $\beta_0 \in [\beta_2, 1]$,
\begin{align}
\tfrac12 \, |P-Q| \leq & (1-\beta_0) \, \prob[\imath_{P\|Q}(Y) < 0] \nonumber \\
& + (\beta_0 - \beta_2) \, \prob[\imath_{P\|Q}(Y) < \log \beta_0]
\label{eq: ubtv2-ita14th5}
\end{align}
where $Y \sim Q$.
Furthermore, both upper bounds on $|P-Q|$ in \eqref{eq: ubtv1-ita14th5} and
\eqref{eq: ubtv2-ita14th5} are tight in the sense that they are achievable
by suitable pairs of probability measures defined on a binary alphabet.
\end{theorem}
\begin{proof}
Since the integrand in the right side of \eqref{eq: TV6} is monotonically
increasing in $\beta$, we may upper bound it by
$\prob\left[\imath_{P\|Q}(X) > \log \frac1{\beta_0}\right]$ when
$\beta \in [\beta_1, \beta_0]$, and by $\prob[\imath_{P\|Q}(X) > 0]$ when
$\beta \in (\beta_0, 1]$.
The same reasoning applied to \eqref{eq: TV5} yields \eqref{eq: ubtv2-ita14th5}.

To check the tightness of \eqref{eq: ubtv1-ita14th5} for any $\beta_1 \in (0, 1]$, choose
an arbitrary $\eta \in (0,1)$ and a pair of probability measures
$P$ and $Q$ defined on the binary alphabet $\set{A} = \{0,1\}$ with
\begin{align}
P(0)&= \frac{1-\eta}{1-\eta \beta_1}, \\[0.1cm]
Q(0)&= \beta_1 P(0).
\end{align}
Then, we have
$\imath_{P\|Q}(0) = \log \tfrac1{\beta_1}$,
$\imath_{P\|Q}(1) = \log \eta < 0 $, and both sides of \eqref{eq: ubtv1-ita14th5}
are readily seen to be equal when $\beta_0 = \beta_1$.
The tightness of \eqref{eq: ubtv2-ita14th5} can be shown in a similar way.
\end{proof}

The third upper bound on $|P-Q|$ is a classical inequality \cite[(99)]{Kailath67},
usually given in the context of bounding the error probability of Bayesian binary
hypothesis testing in terms of the Bhattacharyya distance.

\begin{theorem}
\label{thm: Guntuboyina11}
\begin{align}
\label{eq: Guntuboyina11}
\tfrac14 |P-Q|^2 \leq 1-\exp\bigl(-D_{\frac12}(P\|Q)\bigr).
\end{align}
\end{theorem}

\begin{remark}
The bound in \eqref{eq: Guntuboyina11} is tight if $P$, $Q$
are defined on $\{0, 1\}$ with $P(0)=Q(0)$ or $P(0)=Q(1)$.
\end{remark}

In view of the monotonicity of $D_{\alpha}(P\|Q)$ in $\alpha$,
Theorem~\ref{thm: Guntuboyina11} yields \eqref{eq: BretagnolleH79},
which is equivalent to the Bretagnole-Huber inequality
\cite[(2.2)]{BretagnolleH79} (see also \cite[pp.~30--31]{Vapnik98}).
Note that \eqref{eq: BretagnolleH79} is tighter than \eqref{eq: Pinsker}
only when $|P-Q| > 1.7853$.

\subsection{Lower Bounds on $|P-Q|$}  \label{subsec: lb-tv}
In this section, we give several lower bounds on $|P-Q|$ in terms of the
relative information spectrum. Furthermore, in Section~\ref{sec:reverseP},
we give lower bounds on $|P-Q|$ in terms of the relative entropy (as well
as other features of $P$ and $Q$).
\par
If, for at least one value of $\beta \in (0,1)$, either
$\mathbb{P} \left[ \imath_{P\|Q}(X) > \log \frac1{\beta} \right]$
or $\mathbb{P} \left[ \imath_{P\|Q}(Y) < \log \beta\right]$
are known then we get the following lower bounds on $|P-Q|$ as a
consequence of Theorem~\ref{thm: tv}:

\begin{theorem} \label{thm: 2 LBs TV}
If $P \ll Q$ then, for every $\beta_0 \in (0,1)$,
\begin{align}
\label{eq: 1lbtv}
& |P-Q| \geq 2 (1-\beta_0) \, \mathbb{P} \bigl[ \imath_{P \| Q}(Y) <
\log \beta_0 \bigr], \\[0.1cm]
\label{eq: 2lbtv}
& |P-Q| \geq 2 (1-\beta_0) \, \mathbb{P} \left[ \imath_{P \| Q}(X) >
\log \frac1{\beta_0} \right]
\end{align}
with $X \sim P$ and $Y \sim Q$.
\end{theorem}
\begin{proof}
The lower bounds in \eqref{eq: 1lbtv} and \eqref{eq: 2lbtv}
follow from \eqref{eq: TV5} and \eqref{eq: TV6} respectively.
For example, from \eqref{eq: TV5}, it follows that for an
arbitrary $\beta_0 \in (0,1)$
\begin{align}
\tfrac12 |P-Q| & = \int_0^1 \mathbb{P}\bigl[\imath_{P \| Q}(Y) <
\log \beta\bigr] \, \mathrm{d}\beta \\
\label{eq: novotel}
& \geq \int_{\beta_0}^1 \mathbb{P}\bigl[\imath_{P \| Q}(Y) <
\log \beta\bigr] \, \mathrm{d}\beta \\
\label{eq: monotonicity}
& \geq (1-\beta_0) \, \mathbb{P}\bigl[\imath_{P \| Q}(Y) <
\log \beta_0\bigr]
\end{align}
where \eqref{eq: monotonicity} holds since the integrand in
\eqref{eq: novotel} is monotonically increasing in $\beta \in (0,1]$.
\end{proof}

\par
Next we exemplify the utility of Theorem~\ref{thm: 2 LBs TV} by giving an
alternative proof to the tight lower bound on the relative information spectrum,
given in \cite[Proposition~2]{LiuCV1_IT15} as a function of the total variation
distance.

\begin{proposition} \label{prop: lbris}
Let $P \ll Q$, then for every $\beta>0$
\begin{align} \label{eq: lbris}
\hspace*{-0.15cm} \mathds{F}_{P\|Q}(\log \beta)
& \geq \left\{ \begin{array}{ll}
\hspace*{-0.05cm} 0,
& \hspace*{-0.15cm} \beta \in \bigl(0, \, \frac{2}{2-|P-Q|} \bigr], \\[0.2cm]
\hspace*{-0.05cm} 1-\frac{\beta \, |P-Q|}{2(\beta-1)},
& \hspace*{-0.15cm} \beta \in \bigl(\frac{2}{2-|P-Q|}, \, \infty \bigr).
\end{array}
\right.
\end{align}
Furthermore, for every $\beta > 0$ and $\delta \in [0,1)$, the lower bound in
\eqref{eq: lbris} is attainable by a pair $(P,Q)$ with $|P-Q| = 2\delta$.
\end{proposition}
\begin{proof}
Since $P \ll Q$, they cannot be mutually singular and therefore $|P-Q|<2$.
From \eqref{eq:RIS} and \eqref{eq: 2lbtv} (see Theorem~\ref{thm: 2 LBs TV}),
it follows that for every $\beta_0 \in (0,1)$
\begin{align}
\tfrac12 |P-Q| \geq  (1-\beta_0) \,
\Bigl[1-\mathds{F}_{P\|Q}\Bigl(\log \frac1{\beta_0}\Bigr) \Bigr].
\end{align}
Consequently, the substitution $\beta = \frac1{\beta_0} > 1$ yields
\begin{align}
\mathds{F}_{P\|Q}(\log \beta) \geq 1 - \tfrac{\beta \, |P-Q|}{2(\beta-1)}
\end{align}
which provides a non-negative lower bound on the relative information
spectrum provided that $\beta \geq \frac{2}{2-|P-Q|}$. Having shown
\eqref{eq: lbris}, we proceed to argue that it is tight.
Fix $\delta \in [0,1)$ and let  $|P-Q| = 2\delta$, which yields
$\frac{2}{2-|P-Q|} = \frac1{1-\delta}$ in the right side of \eqref{eq: lbris}.
\begin{itemize}
\item If $\beta < \frac1{1-\delta}$,  let the pair $(P,Q)$ be defined
on the binary alphabet $\{0,1\}$ with $P(1)=1$ and $Q(1)=1-\delta$
(thereby ensuring $2\delta = |P-Q| $). Then, from \eqref{eq:RIS},
\begin{align}
\mathds{F}_{P\|Q}(\log \beta) = P(0) = 0.
\end{align}
\item If $\beta \geq \frac1{1-\delta}$, let $\tau > \beta$ and consider
the probability measures $P = P_{\tau}$ and $Q = Q_{\tau}$ defined on
the binary alphabet $\{0,1\}$ with $P_{\tau}(1) = \frac{\tau \delta}{\tau-1}$
and $Q_{\tau}(1) = \frac{\delta}{\tau-1}$ (note that indeed
$2 \delta = |P_{\tau}-Q_{\tau}|$). Since $1 < \beta < \tau $ then
\begin{align}
\mathds{F}_{P\|Q}(\log \beta) = P_{\tau}(0) = 1-\tfrac{\tau \delta}{\tau-1}
\end{align}
which tends to $1-\frac{\beta \delta}{\beta-1}$ in the right side of
\eqref{eq: lbris} by letting $\tau \downarrow \beta$.
\end{itemize}
\end{proof}

Attained under certain conditions, the following counterpart
to Theorem~\ref{thm: ubtv-ita14th5} gives a lower bound on the
total variation distance based on the distribution of the relative
information. It strengthens the bound in \cite[Theorem~8]{Verdu_ITA14},
which in turn tightens the lower bounds in \cite[(2.3.18)]{Pinsker60}
and \cite[Lemma~7]{SteinbergS}.

\begin{theorem} \label{thm: LB3 TV}
If $P \ll \gg Q$ then, for any $\eta_1, \eta_2 > 0$,
\begin{align}
|P-Q| \geq \bigl(1-\exp(-\eta_1)\bigr) \;
\prob\bigl[\imath_{P\|Q}(X) \geq \eta_1 \bigr]
+ \bigl(\exp(\eta_2)-1 \bigr) \;
\prob\bigl[\imath_{P\|Q}(X) \leq -\eta_2 \bigr]
\label{eq: LB3 TV}
\end{align}
with $X \sim P$. Equality holds in \eqref{eq: LB3 TV} if $P$ and $Q$
are probability measures defined on $ \{0,1\}$ and, for an arbitrary
$\eta_1, \eta_2>0$,
\begin{align} \label{eq: tlbtv}
P(0) &= \frac{1-\exp(-\eta_2)}{1-\exp(-\eta_1-\eta_2)} \, ,\\
Q(0) &= \exp(-\eta_1) \, P(0).
\end{align}
\end{theorem}

\begin{proof}
From \eqref{eq: TV8}, it follows that for arbitrary $\eta_1, \eta_2 > 0$,
\begin{align}
 |P-Q|
& \geq \mathbb{E}\bigl[ \bigl| 1 - \exp\bigl(-\imath_{P\|Q}(X)\bigr) \bigr|
\; 1\bigl\{\imath_{P\|Q}(X) \geq \eta_1 \bigr\} \bigr]  \nonumber \\[0.1cm]
& \hspace*{0.4cm} + \mathbb{E}\bigl[ \bigl| 1 - \exp\bigl(-\imath_{P\|Q}(X)\bigr)
\bigr| \; 1\bigl\{\imath_{P\|Q}(X) \leq -\eta_2 \bigr\} \bigr]
\end{align}
which is readily loosened to obtain \eqref{eq: LB3 TV}. Equality holds in
\eqref{eq: LB3 TV} for $P$ and $Q$ in the theorem statement
since $\imath_{P\|Q}(X)$ only takes the values
$\log \, \frac{P(0)}{Q(0)} = \eta_1$ and
$\log \, \frac{P(1)}{Q(1)} = -\eta_2$.
\end{proof}

The following lower bound on the total variation distance is the counterpart
to Theorem~\ref{thm: ubtv}.

\begin{theorem} \label{thm: ita14-thm6}
If $P \ll \gg Q$, and $X \sim P$ then
\begin{align} \label{eq: ita14-thm6}
|P-Q| \, \log e \geq \mathbb{E} \bigl[ |\imath_{P\|Q}(X)| \bigr] - D(P \| Q).
\end{align}
\end{theorem}
\begin{proof}
We reason in parallel to the proof of Theorem~\ref{thm: ubtv}. For all $z \in [-\infty, \infty]$,
\begin{align}
\label{eq: to-prove-ita14-thm6}
\bigl[ 1 - \exp(-z) \bigr]^{-} \geq \frac{(z)^-}{\log e}.
\end{align}
Substituting $z = \imath_{P\|Q}(X)$, taking expectation of both sides of
\eqref{eq: to-prove-ita14-thm6}, and using  \eqref{eq: TV7} we obtain
\begin{align}
|P-Q| \, \log e & \geq 2 \, \mathbb{E} \bigl[\bigl( \, \imath_{P\|Q}(X) \, \bigr)^- \bigr] \\
& = \mathbb{E} \bigl[ | \imath_{P\|Q}(X) | - \imath_{P\|Q}(X) \bigr] \\
& = \mathbb{E} \bigl[ | \imath_{P\|Q}(X) | \bigr] - D(P \| Q).
\end{align}
\end{proof}

\begin{remark} \label{remark: Barron's inequality}
The combination of Pinsker's inequality \eqref{eq: Pinsker} and \eqref{eq: ita14-thm6}
yields the following inequality due to Barron (see \cite[p.~339]{Barron86}) which is
useful in establishing convergence results for relative entropy (e.g. \cite{Barron91})
\begin{align} \label{eq: Barron86}
\mathbb{E} \bigl[ | \imath_{P\|Q}(X) | \bigr] \leq D(P \| Q) + \sqrt{2 \, D(P\|Q) \, \log e}
\end{align}
with $X \sim P$.
\end{remark}

\subsection{Relative Entropy and Bhattacharyya Distance} \label{subsec: bounds-RE}
The following result refines \eqref{grout425 - introduction}
by using an approach which relies on moment inequalities \cite{Simic07}--\cite{Simic15}. The
coverage in this section is self-contained.
\begin{theorem}  \label{theorem: refined upper/lower bounds on RE}
If $P \ll \gg Q$, then
\begin{align}
\label{eq: improved chi^2-bound on RE}
D(P \| Q) \leq & \log\bigl(1+\chi^2(P\|Q)\bigr)
- \frac{\frac32 \bigl(\chi^2(P \| Q) \bigr)^2
\, \log e}{\bigl(1 + \chi^2(Q\|P)\bigr) \,
\bigl( 1 + \chi^2(P \| Q) \bigr)^2 - 1}.
\end{align}
Furthermore, if $\{P_n\}$ converges to $Q$ in the sense of \eqref{eq: 1st condition},
then the ratio of $D(P_n \| Q)$ and its upper bound in \eqref{eq: improved chi^2-bound on RE} tends
to~1 as $n \to \infty$.
\end{theorem}
\begin{proof}
The derivation of \eqref{eq: improved chi^2-bound on RE} relies on \cite[Theorem~2.1]{Simic07} which
states that if $W$ is a non-negative random variable, then
\begin{align} \label{eq: log-convex lambda}
\lambda_{\alpha} \triangleq \left\{ \begin{array}{ll}
\frac{\bigl(\mathbb{E}[W^\alpha] - \mathbb{E}^\alpha[W] \bigr) \, \log e}{\alpha(\alpha-1)}, &
\; \alpha \neq 0, 1 \\[0.1cm]
\log \bigl( \mathbb{E}[W] \bigr) - \mathbb{E}[\log W], & \; \alpha=0 \\[0.1cm]
\mathbb{E}[W \log W] - \mathbb{E}[W] \, \log \bigl(\mathbb{E}[W]\bigr), &
\; \alpha=1
\end{array}
\right.
\end{align}
is log-convex in $\alpha \in \Reals$.

To prove \eqref{eq: improved chi^2-bound on RE}, let $W = \frac{\mathrm{d}P}{\mathrm{d}Q} \, (X)$
with $X \sim P$, then \eqref{eq: log-convex lambda} yields
\begin{align}
\label{eq1: lambda0}
& \lambda_0 = \log \bigl(1 + \chi^2(P \| Q) \bigr) - D(P \| Q), \\[0.1cm]
& \lambda_{-\alpha} = \frac1{\alpha (\alpha+1)} \Bigl[1+(\alpha-1) \mathscr{H}_{\alpha}(Q \| P)
- \bigl(1 + \chi^2(P \| Q) \bigr)^{-\alpha} \Bigr] \, \log e  \label{eq1: lambda -alpha}
\end{align}
for all $\alpha > 0$, and specializing \eqref{eq1: lambda -alpha} yields
\begin{align}
\label{eq1: lambda-1}
& \lambda_{-1} = \frac{\chi^2(P \| Q) \, \log e}{2 \bigl(1+\chi^2(P \| Q) \bigr)}, \\
\label{eq1: lambda2}
& \lambda_{-2} = \frac16 \left[ 1 + \chi^2(Q \| P) - \frac1{\bigl(1+\chi^2(P\|Q)\bigr)^2} \right] \, \log e.
\end{align}
In view of the log-convexity of $\lambda_{\alpha}$ in $\alpha \in \Reals$, then
\begin{align} \label{eq1: log-convexity ineq.}
\lambda_0 \, \lambda_{-2} \geq \lambda_{-1}^2
\end{align}
which, by assembling \eqref{eq1: lambda0}--\eqref{eq1: log-convexity ineq.},
yields \eqref{eq: improved chi^2-bound on RE}.

Suppose that $\{P_n\}$ converges to $Q$ in the sense of \eqref{eq: 1st condition}.
Then, it follows from Theorem~\ref{thm: local behavior fD} and Corollary~\ref{cor:ratios:chi}
that
\begin{align}
\label{eq1: local}
& \lim_{n \to \infty} D(P_n \| Q) = 0, \\
\label{eq2: local}
& \lim_{n \to \infty} \chi^2(P_n \| Q) = 0, \\
\label{eq3: local}
& \lim_{n \to \infty} \frac{D(P_n \| Q)}{\chi^2(P_n \| Q)} = \tfrac12 \, \log e, \\
\label{eq4: local}
& \lim_{n \to \infty} \frac{\chi^2(Q \| P_n)}{\chi^2(P_n \| Q)} = 1.
\end{align}
Let $U_n$ denote the upper bound on $D(P_n \| Q)$ in \eqref{eq: improved chi^2-bound on RE}.
Assembling \eqref{eq1: local}--\eqref{eq4: local}, it can be verified that
\begin{align}
\lim_{n \to \infty} \frac{U_n}{D(P_n \| Q)} = 1.
\end{align}
\end{proof}

\begin{remark}
In view of \eqref{eq1: local}--\eqref{eq4: local}, while the ratio of the right side of
\eqref{eq: improved chi^2-bound on RE} with $P=P_n$ and $D(P_n \| Q)$ tends to~1, the
ratio of the looser bound in \eqref{grout425 - introduction} and $D(P_n \| Q)$ tends to~2.
\end{remark}

\begin{remark}
If $\{P_n\}$ and $Q$ are defined on a finite set $\set{A}$, then the condition
in \eqref{eq: 1st condition} is equivalent to $|P_n - Q| \to 0$ with $Q(a)>0$
for all $a \in \set{A}$.
\end{remark}

\begin{remark}
An alternative refinement of \eqref{grout425 - introduction} has been recently obtained in \cite{Simic15}
as a function of $\chi^2(P\|Q)$ and the Bhattacharyya distance $B(P \| Q)$ (see Definition~\ref{definition: B distance}):
\begin{align}
\label{eq: Simic's bound}
D(P \| Q)  \leq \log\bigl( 1 + \chi^2(P \| Q) \bigr)
- \frac{32}{9} \, \frac{\left[\exp\bigl(-B(P\|Q)\bigr)
\sqrt{1+\chi^2(P \| Q)} - 1 \right]^2 \, \log e}{\chi^2(P\|Q)}.
\end{align}
Eq.~\eqref{eq: Simic's bound} can be generalized by relying on the log-convexity
of $\lambda_{\alpha}$ in $\alpha \in \Reals$, which yields
\begin{align} \label{eq2: log-convexity}
\lambda_0^{1-\alpha} \, \lambda_{-1}^{\alpha} \geq \lambda_{-\alpha}
\end{align}
for all $\alpha \in (0,1)$; consequently, assembling \eqref{eq1: lambda0}, \eqref{eq1: lambda -alpha},
\eqref{eq1: lambda-1} and \eqref{eq2: log-convexity} yields
\begin{align}
\label{eq: generalization - Simic}
D(P \| Q)  \leq & \log\bigl(1+\chi^2(P\|Q)\bigr) \\
& - \left(\frac{2^\alpha}{\alpha(\alpha+1)}\right)^{\frac1{1-\alpha}} \,
\bigl(\chi^2(P\|Q)\bigr)^{-\frac{\alpha}{1-\alpha}}
\; \Bigl[\bigl(1-(1-\alpha) \mathscr{H}_{\alpha}(Q\|P) \bigr) \,
\bigl(1+\chi^2(P \| Q) \bigr)^{\alpha}-1 \Bigr]^{\frac1{1-\alpha}} \, \log e \nonumber
\end{align}
for all $\alpha \in (0,1)$.
Note that in the special case $\alpha= \tfrac12$, \eqref{eq: generalization - Simic} becomes
\eqref{eq: Simic's bound}, as can be readily verified in view of \eqref{eq2: B distance} and
the symmetry property $\mathscr{H}_{\frac12}(P\|Q)=\mathscr{H}_{\frac12}(Q\|P)$.
\end{remark}

\begin{remark}
The following lower bound on the relative entropy has been derived in \cite{Simic15},
based on the approach of moment inequalities:\footnote{For the derivation of
\eqref{eq: LB, Simic15} for a general alphabet, similarly to \cite{Simic15}, set
$W = \sqrt{\frac{\mathrm{d}Q}{\mathrm{d}P} \, (X)}$ in \eqref{eq: log-convex lambda}
with $X \sim P$, and use the inequality $\lambda_0 \, \lambda_4 \geq \lambda_2^2$
which follows from the log-convexity of $\lambda_\alpha$ in $\alpha$.}
\begin{align}
D(P \| Q) \geq 2 B(P \| Q)
+ \frac{6 \left[1 - \exp\bigl(-2B(P\|Q)\bigr)\right]^2 \, \log e}{1 - \exp\bigl(-4B(P\|Q)\bigr) + \chi^2(Q\|P)}.
\label{eq: LB, Simic15}
\end{align}
Note that from \eqref{eq1: B distance}
\begin{align} \label{eq: upper bound on B}
B(P\|Q) \geq \tfrac12 \log\biggl(\frac1{1-\tfrac14 \, |P-Q|^2}\biggr)
\end{align}
and since the right side of \eqref{eq: LB, Simic15} is monotonically increasing in $B(P\|Q)$,
the replacement of $B(P\|Q)$ in the right side of \eqref{eq: LB, Simic15} with its
lower bound in \eqref{eq: upper bound on B} yields
\begin{align}
D(P \| Q) \geq \log\biggl(\frac1{1-\tfrac14 \, |P-Q|^2} \biggr)
+ \frac{\tfrac34 \, |P-Q|^2 \log e}{1-\tfrac18 \, |P-Q|^2 + \frac{2 \, \chi^2(Q \| P)}{|P-Q|^2}}.
\label{eq: refined Bretagnole-Huber ineq.}
\end{align}
Although \eqref{eq: refined Bretagnole-Huber ineq.} improves
the bound in \eqref{eq: BretagnolleH79}, it is weaker than
\eqref{eq: LB, Simic15}, and it satisfies the tightness property in
Theorem~\ref{theorem: refined upper/lower bounds on RE} only in special
cases such as when $P$ and $Q$ are defined on $\set{A} = \{0,1\}$ with
$P(0)=Q(1) = \tfrac12 - \varepsilon$ and we let $\varepsilon \to 0$.
\end{remark}

Define the binary relative entropy function as the continuous extension to $[0,1]^2$ of
\begin{align} \label{eq: binary RE}
d(x\|y) = x \log \left(\frac{x}{y}\right) + (1-x) \log \left(\frac{1-x}{1-y}\right).
\end{align}
The following result improves the upper bound in \eqref{eq: RE and chi-square}.

\begin{theorem} \label{theorem: bounds RE/chi^2 for given beta1,2}
Let $P \ll \gg Q$ with  $(\beta_1, \beta_2) \in (0,1)^2$. Then,
\begin{enumerate}[a)]
\item
\begin{align} \label{eq: UB on chi^2 for given beta1,2}
\chi^2(P\|Q) \leq \bigl(\beta_1^{-1}-1\bigr) (1-\beta_2),
\end{align}
which is attainable for binary alphabets.
\item
\begin{align}
D(P \| Q)
\label{eq1: improved UB on RE}
& \leq \min \biggl\{ \log(1+c) - \frac{\frac32 \, c
\, \log e}{1 + \bigl(1+\beta_2^{-1} \bigr) \, (1+c)}, \\[0.1cm]
\label{eq3: improved UB on RE}
&  \hspace*{1.2cm} \frac{\bigl(\sqrt{c^2+8\beta_2 c \log_e(1+c)}
- c \bigr) \, \log e}{4 \beta_2}, \\[0.1cm]
\label{eq4: improved UB on RE}
& \hspace*{1.2cm} d\left( \frac{\beta_1^{-1}-1}{\beta_1^{-1}
\beta_2^{-1} - 1} \, \Bigl\| \, \frac{\beta_2^{-1}
(\beta_1^{-1}-1)}{\beta_1^{-1} \beta_2^{-1} - 1} \right) \biggr\}
\end{align}
where we have abbreviated $c=\chi^2(P\|Q)$ for typographical convenience.
\end{enumerate}
\end{theorem}

\begin{proof}
To prove \eqref{eq: UB on chi^2 for given beta1,2}, we first consider the case where
$P,Q$ are defined on $\set{A} = \{0,1\}$ and $\frac{P(0)}{Q(0)} = \beta_2$,
$\frac{P(1)}{Q(1)} = \beta_1^{-1}$. Straightforward calculation yields
\begin{align}
& P(0) = \frac{\beta_1^{-1}-1}{\beta_1^{-1} \, \beta_2^{-1} - 1},
\quad Q(0) = \beta_2^{-1} \, P(0)
\end{align}
and
\begin{align}
& \chi^2(P\|Q) = \bigl(\beta_1^{-1}-1\bigr) (1-\beta_2).
\end{align}
In the case of a general alphabet,
consider the elementary bound with $a <0<b$:
$ \mathbb{E} [ Z^2] \leq -ab$
which holds for any $ Z\in[a,b]$, $\mathbb{E}[Z] = 0$, and follows simply by taking expectations of
\begin{align}
Z^2 &= - a b + Z (a+b) - (Z-a)(b-Z) \\
\label{eq: Z^2}
&\leq - a b + Z (a+b).
\end{align}
Since $ \chi^2(P\|Q) = \mathbb{E}[Z^2] $,  \eqref{eq: UB on chi^2 for given beta1,2}
follows by letting $a = \beta_2 -1$, $b = \beta_1^{-1}-1$ and
\begin{align}
\label{eq1a: Z}
& Z = \frac{\mathrm{d}P}{\mathrm{d}Q} \, (Y) - 1, \quad Y \sim Q.
\end{align}

To prove \eqref{eq1: improved UB on RE}, note that it follows by combining
\eqref{eq: improved chi^2-bound on RE} with the left side of the inequality
\begin{align} \label{eq: ratio of chi^2 divergences}
\beta_2 \, \chi^2(Q \| P) \leq \frac{\chi^2(P \| Q)}{1 + \chi^2(P \| Q)}
\leq \beta_1^{-1} \, \chi^2(Q \| P)
\end{align}
where \eqref{eq: ratio of chi^2 divergences} follows from
Theorem~\ref{thm: GI fD} with $f(t) = \frac1{t} - 1$ for $t > 0$.

To prove \eqref{eq3: improved UB on RE}, note that assembling
\eqref{eq1: Simic08} and \eqref{lbchi} yields
\begin{align} \label{eq: quadratic inequality}
D(P \| Q) \leq \log \bigl(1+\chi^2(P\|Q)\bigr)
- \frac{2 \beta_2 D^2(P\|Q)}{\chi^2(P\|Q) \, \log e}
\end{align}
and solving this quadratic inequality in $D(P\|Q)$,
for fixed $\chi^2(P \| Q)$, yields the bound in
\eqref{eq3: improved UB on RE}.

Bound \eqref{eq4: improved UB on RE} holds since the
maximal $D(P\|Q)$, for fixed $\chi^2(P\|Q)$, is monotonically increasing
in $\chi^2(P\|Q)$. In view of \eqref{eq: UB on chi^2 for given beta1,2},
$D(P\|Q)$ cannot be larger than its maximal value when
$\chi^2(P\|Q) = \bigl(\beta_1^{-1}-1\bigr) (1-\beta_2)$.
In the latter case, the condition of equality in \eqref{eq: Z^2}
(recall that $\mathbb{E}[Z]=0$) is
\begin{align}
\label{eq: Prob1 Z}
& \mathbb{P}[Z=a] = \frac{b}{b-a} = 1 - \mathbb{P}[Z=b]
\end{align}
which implies that the maximal relative entropy $D(P\|Q)$ over all $P \ll \gg Q$
with given $(\beta_1, \beta_2) \in (0,1)^2$ is equal to
\begin{align}
\label{eq11:max RE}
& \mathbb{E}\bigl[(1+Z) \, \log(1+Z)\bigr] \\[0.1cm]
\label{eq12:max RE}
&= \frac{b(1+a) \, \log(1+a) - a(1+b) \log(1+b)}{b-a} \\[0.1cm]
\label{eq13:max RE}
&= d\left( \frac{\beta_1^{-1}-1}{\beta_1^{-1}
\beta_2^{-1} - 1} \, \Bigl\| \, \frac{\beta_2^{-1}
(\beta_1^{-1}-1)}{\beta_1^{-1} \beta_2^{-1} - 1} \right)
\end{align}
where \eqref{eq11:max RE}--\eqref{eq13:max RE} follow from
\eqref{eq: binary RE} and
\eqref{eq: Prob1 Z} with $a = \beta_2 -1$ and $b = \beta_1^{-1}-1$.
\end{proof}

\begin{remark}
The proof of \eqref{eq1: improved UB on RE}
relies on the left side of \eqref{eq: ratio of chi^2 divergences}; this
strengthens the bound which follows from Theorem~\ref{thm: fD1}, given by
$\chi^2(Q \| P) \leq \beta_2^{-1} \, \chi^2(P \| Q)$. The bound
\eqref{eq3: improved UB on RE} is typically of similar tightness as
the bound in \eqref{eq1: improved UB on RE},
although none of them outperforms the other for all $(\beta_1, \beta_2)
\in (0,1)^2$ and $\chi^2(P \| Q) \in \bigl[0, \, \bigl(\beta_1^{-1}-1\bigr) (1-\beta_2) \bigr]$
(see \eqref{eq: UB on chi^2 for given beta1,2}).
\end{remark}

\begin{remark} \label{remark: locus of RE-chi^2}
The left inequality in \eqref{eq: RE and chi-square} and
Theorem~\ref{theorem: bounds RE/chi^2 for given beta1,2} provide an
analytical outer bound on the locus of the points $(\chi^2(P\|Q), D(P\|Q))$
where $P \ll \gg Q$ and
$\beta_2 \leq \frac{\mathrm{d}P}{\mathrm{d}Q} \leq \beta_1^{-1}$ for
given $(\beta_1, \beta_2) \in (0,1)^2$.
\end{remark}

\begin{example}
In continuation to Remark~\ref{remark: locus of RE-chi^2}, for given $(\beta_1, \beta_2) \in (0,1)^2$,
Figure~\ref{figure: locus of RE-chi^2} compares the locus of the points $(\chi^2(P\|Q), D(P\|Q))$ when
$P, Q$ are restricted to binary alphabets, and $\frac{P}{Q}$ is bounded between $\beta_2$ and $\beta_1^{-1}$,
with an outer bound constructed with the left inequality in
\eqref{eq: RE and chi-square} and Theorem~\ref{theorem: bounds RE/chi^2 for given beta1,2} (recall that
the outer bound is valid for an arbitrary alphabet).
\begin{figure}[h]
\includegraphics[width=7.7cm]{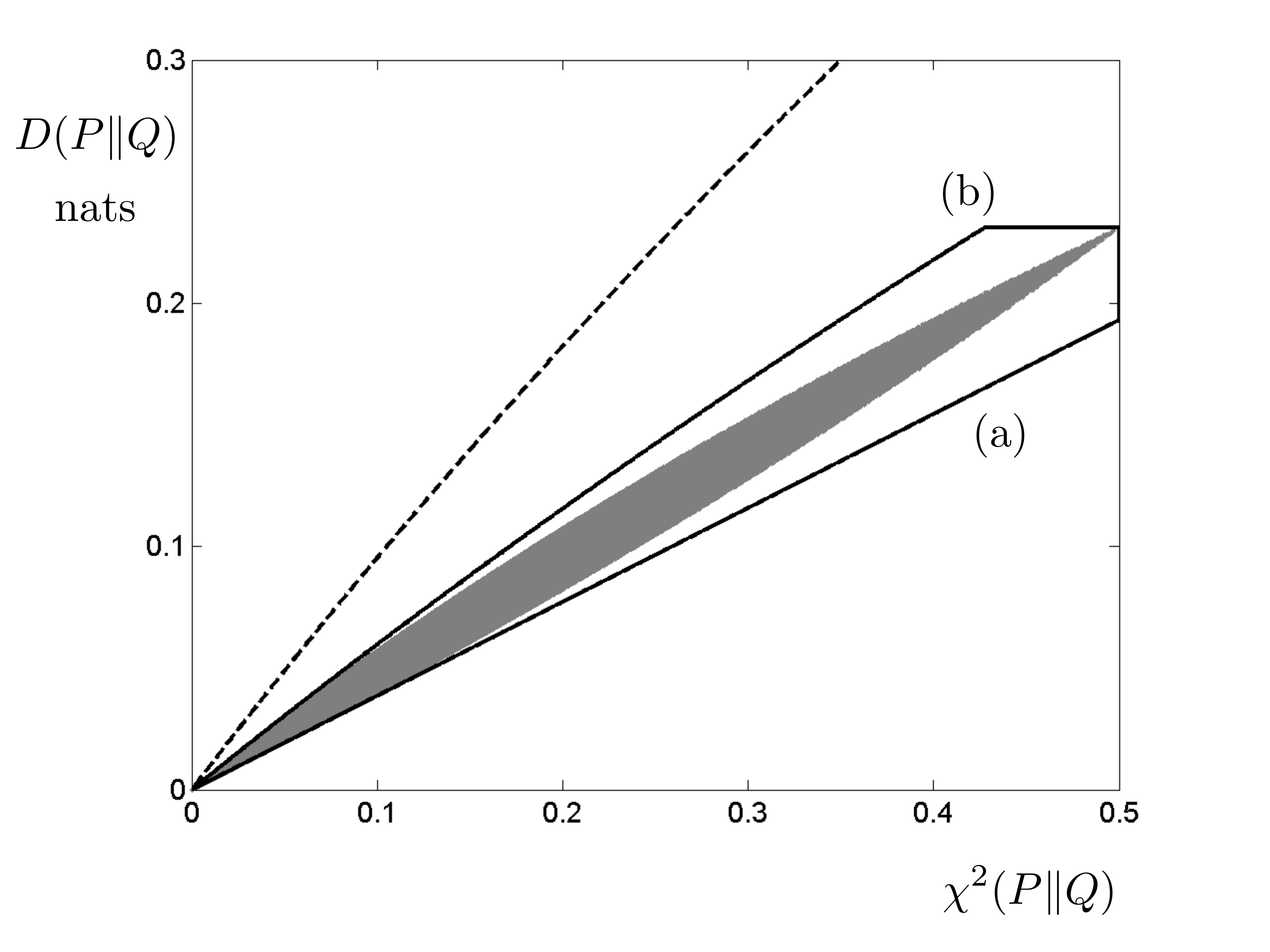}
\hspace*{-0.1cm}
\includegraphics[width=7.7cm]{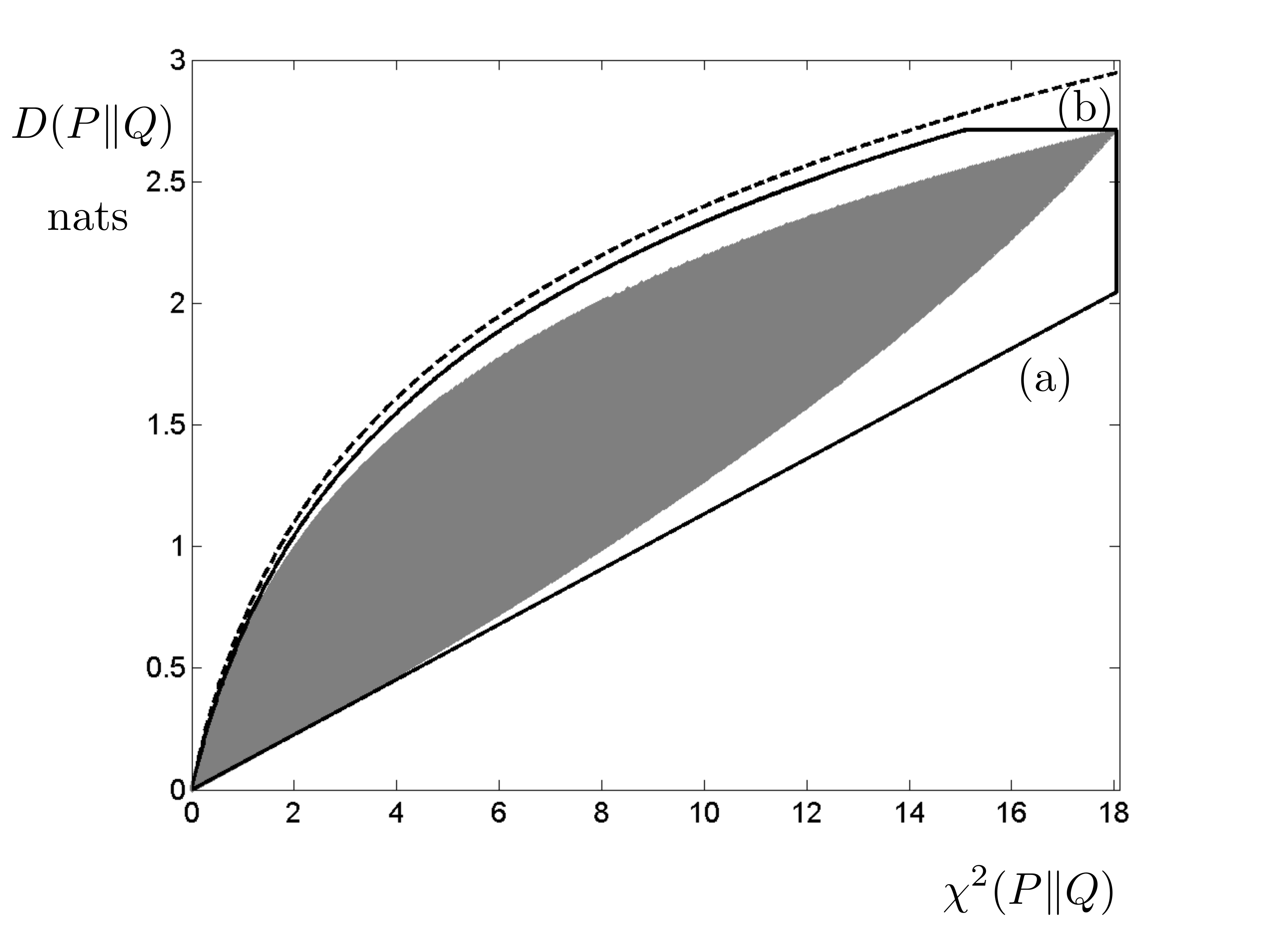}
\vspace*{-0.5cm}
\caption{\label{figure: locus of RE-chi^2}
Comparison of the locus of the points $(\chi^2(P\|Q), D(P\|Q))$ when
$P, Q$ are defined on $\set{A} = \{0,1\}$ with the bounds in the left side
of \eqref{eq: RE and chi-square} (a) and
Theorem~\ref{theorem: bounds RE/chi^2 for given beta1,2} (b);
$(\beta_1, \beta_2) = \bigl(\tfrac12, \tfrac12\bigr)$ and $(\beta_1, \beta_2)
= \bigl(\tfrac1{20}, \tfrac1{20}\bigr)$ in the upper and lower plots, respectively.
The dashed curve in each plot corresponds to the looser bound in \eqref{grout425 - introduction}.}
\end{figure}
\end{example}
\par
The following result relies on the earlier analysis to provide bounds on the
Bhattacharyya distance, expressed in terms of $\chi^2$ divergences
and relative entropy.
\medskip
\begin{theorem} \label{thm: bounds on B distance}
If $P \ll \gg Q$, then the following bounds on the Bhattacharyya distance hold:
\begin{align}
& \tfrac12 \log \bigl(1+\chi^2(P\|Q)\bigr)
- \log \left(1+\tfrac34 \sqrt{\tfrac{\chi^2(P\|Q)}{2\log e}
\Bigl[\log \bigl(1+\chi^2(P\|Q)\bigr)-D(P\|Q) \Bigr]} \right) \nonumber \\[0.1cm]
\label{eq: B2}
& \leq B(P \| Q) \\[0.1cm]
& \leq \tfrac12 \log \bigl(1+\chi^2(P\|Q)\bigr)
-\log \left(1+\frac{\bigl(\tfrac34 \, \chi^2(P\|Q)\bigr)^{\frac32}}
{\sqrt{\bigl(1+\chi^2(P\|Q)\bigr)^2 \bigl(1+\chi^2(Q\|P)\bigr)-1}} \right).
\label{eq: B3}
\end{align}
Furthermore, if $\{P_n\}$ converges to $Q$ in the sense of \eqref{eq: 1st condition},
then the ratio of the bounds on $B(P_n \| Q)$ in \eqref{eq: B2} and \eqref{eq: B3}
tends to~1 as $n \to \infty$.
\end{theorem}

\begin{proof}
In view of the log-convexity of $\lambda_{\alpha}$ in $\alpha \in \Reals$,
\begin{align} \label{eq: lambda inequalities}
\lambda_0 \, \lambda_{-1} \geq \lambda_{-\frac12}^2, \quad
\lambda_{-\frac12}^2 \, \lambda_{-2} \geq \lambda_{-1}^3
\end{align}
for any choice of the random variable $W$ in \eqref{eq: log-convex lambda}.
Consequently, assembling \eqref{eq2: B distance}, \eqref{eq1: lambda0},
\eqref{eq1: lambda -alpha} and \eqref{eq: lambda inequalities} yield
the bounds on $B(P\|Q)$ in \eqref{eq: B2} and \eqref{eq: B3}.

Suppose that $\{P_n\}$ converges to $Q$ in the sense of \eqref{eq: 1st condition}.
Let $L_n$ and $U_n$ denote, respectively, the lower and upper bounds on $B(P_n \| Q)$
in \eqref{eq: B2} and \eqref{eq: B3}. Assembling \eqref{eq1: local}--\eqref{eq4: local},
it can be easily verified that
\begin{align} \label{eq: B bounds are tight}
\lim_{n \to \infty} \frac{L_n}{\chi^2(P_n \| Q)} = \lim_{n \to \infty} \frac{U_n}{\chi^2(P_n \| Q)} = \tfrac18 \, \log e
\end{align}
which yields that $\lim_{n \to \infty} \frac{U_n}{L_n} = 1.$
\end{proof}

\begin{remark}
Note that \eqref{eq: B3} refines the bound
\begin{align}
B(P\|Q) \leq \tfrac12 \log\bigl(1+\chi^2(P\|Q)\bigr)
\end{align}
which is equivalent to $\lambda_{-\frac12} \geq 0$ (in view of Jensen's inequality, \eqref{eq2: B distance}
and \eqref{eq1: lambda -alpha}).
\end{remark}

\begin{remark}
Let $\{P_n\}$ converge to $Q$ in the sense of \eqref{eq: 1st condition}.
In view of \eqref{eq: B bounds are tight}, it follows that
\begin{align} \label{eq: B-chi^2}
\lim_{n \to \infty} \frac{B(P_n \| Q)}{\chi^2(P_n\|Q)} = \tfrac18 \, \log e,
\end{align}
from which we can surmise that both upper bounds in \eqref{eq: improved chi^2-bound on RE} and \eqref{eq: Simic's bound}
are tight under the condition in \eqref{eq: 1st condition} (see Theorem~\ref{theorem: refined upper/lower bounds on RE}),
although \eqref{eq: improved chi^2-bound on RE} only depends on $\chi^2$-divergences.
In view of \eqref{eq: B-chi^2}, the lower bound in \eqref{eq: LB, Simic15} is also tight under the
condition in \eqref{eq: 1st condition}, in the sense that the ratio of $D(P_n \| Q)$ and its lower bound
in \eqref{eq: LB, Simic15} tends to~1 as $n \to \infty$; this sufficient condition for the tightness of
\eqref{eq: LB, Simic15} strengthens the result in \cite[Section~4]{Simic15}.
\end{remark}

\section{Reverse Pinsker Inequalities}
\label{sec:reverseP}
It is not possible to lower bound $|P - Q|$ solely in terms of $D(P\|Q)$ since for any
arbitrarily small $\epsilon >0$ and arbitrarily large $\lambda >0$, we can construct
examples with $|P - Q| < \epsilon$ and $\lambda < D(P\|Q) < \infty$. Therefore, each
of the bounds in this section involves not only $D(P\|Q)$ but another feature of the
pair $(P,Q)$.

\subsection{Bounded Relative Information}
\label{subsec: bounded RI}
As in Section~\ref{sec:bounded}, the following result involves the bounds on the
relative information.
\begin{theorem} \label{thm: improved SV-ITA14}
If $\beta_1 \in (0, 1)$ and $\beta_2 \in [0, 1)$, then,
\begin{align} \label{eq: improved SV-ITA14}
D(P\|Q) \leq \tfrac12 \left( \varphi(\beta_1^{-1}) - \varphi(\beta_2) \right) \, |P-Q|
\end{align}
where $\varphi \colon [0, \infty) \to [0, \infty)$ is given by
\begin{align}
\varphi (t) =
\left\{
\begin{array}{ll}
0 & t=0\\
\frac{t \log t}{t-1} & t \in (0,1) \cup (1,\infty) \\
\log e & t=1.
\end{array}
\right.
\end{align}
\end{theorem}

\begin{proof}
Let $X \sim P$, $Y \sim Q$, and $Z$ be defined in \eqref{eq: Z}.
The function $\varphi \colon [0, \infty) \to [0, \infty)$ is continuous,
monotonically increasing and non-negative; the monotonicity property holds
since $(t-1)^2 \varphi'(t) = (t-1) \log e - \log t \geq 0$ for all $t > 0$,
and its non-negativity follows from the fact that $\varphi$ is monotonically
increasing on $[0, \infty)$ and $\varphi(0)=0$. Accordingly,
\begin{align}
\label{eq: range of values for phi of Z}
\varphi(\beta_2) \leq \varphi(Z) \leq \varphi(\beta_1^{-1})
\end{align}
since \eqref{eq: Z} and \eqref{eq: beta1-alt}--\eqref{eq: beta2-alt} imply
that $Z \in [\beta_2, \beta_1^{-1}]$ with probability one. The relative entropy satisfies
\begin{align}
D(P \| Q) &= \mathbb{E} \bigl[ Z \log Z \bigr]  \\
& = \mathbb{E} \bigl[ \varphi(Z) \, (Z-1) \bigr]  \\
& = \mathbb{E} \bigl[ \varphi(Z) \, (Z-1) \, 1\{Z>1\} \bigr]
+ \mathbb{E} \bigl[ \varphi(Z) \, (Z-1) \, 1\{Z<1\} \bigr].
\label{eq: relative entropy}
\end{align}
We bound each of the summands in the right side of \eqref{eq: relative entropy} separately.
Invoking \eqref{eq: range of values for phi of Z}, we have
\begin{align}
\mathbb{E} \bigl[ \varphi(Z) \, (Z-1) \, 1\{Z>1\} \bigr]
& \leq \varphi(\beta_1^{-1}) \, \mathbb{E} \bigl[(Z-1) \, 1\{Z>1\} \bigr]  \\[0.1cm]
\label{eq: see footnote 11}
& = \varphi(\beta_1^{-1}) \, \mathbb{E} \bigl[(1-Z)^-\bigr]  \\[0.1cm]
\label{eq: bound on 1st summand}
& = \tfrac12 \, \varphi(\beta_1^{-1}) \, |P-Q|
\end{align}
where \eqref{eq: see footnote 11} holds since $ (x)^- \triangleq -x \, 1\{x < 0\}$,
and \eqref{eq: bound on 1st summand} follows from \eqref{eq: TV3} with $Z$ in \eqref{eq: Z}.
Similarly, \eqref{eq: range of values for phi of Z} yields
\begin{align}
\mathbb{E} \bigl[ \varphi(Z) \, (Z-1) \, 1\{Z<1\} \bigr]
& \leq \varphi(\beta_2) \, \mathbb{E} \bigl[(Z-1) \, 1\{Z<1\} \bigr]  \\[0.1cm]
& = -\varphi(\beta_2) \, \mathbb{E} \bigl[(1-Z)^+\bigr]  \\[0.1cm]
\label{eq: bound on 2nd summand}
& = -\tfrac12 \, \varphi(\beta_2) \, |P-Q|
\end{align}
where \eqref{eq: bound on 2nd summand} follows from \eqref{eq: TV2}.
Assembling \eqref{eq: relative entropy}, \eqref{eq: bound on 1st summand}
and \eqref{eq: bound on 2nd summand}, we obtain \eqref{eq: improved SV-ITA14}.
\end{proof}

\begin{remark}
By dropping the negative term in \eqref{eq: improved SV-ITA14},
we can get the weaker version in \cite[Theorem~7]{Verdu_ITA14}:
\begin{align}
\label{eq: SV-ITA14}
D(P \| Q) \leq \left( \frac{\log \frac1{\beta_1}}{2(1-\beta_1)} \right) |P-Q|.
\end{align}
The coefficient of $|P-Q|$ in the right side of \eqref{eq: SV-ITA14} is monotonically
decreasing in $\beta_1$ and it tends to $\tfrac12 \log e$ by letting $\beta_1 \to 1$.
The improvement over \eqref{eq: SV-ITA14} afforded by \eqref{eq: improved SV-ITA14}  is exemplified
in Appendix~\ref{appendix:improvement}. The bound in \eqref{eq: SV-ITA14} has been recently used in
the context of the optimal quantization of probability measures \cite[Proposition~4]{BochererG_arXiv}.
\end{remark}

\begin{remark} \label{remark: ub RE-TV-chi}
The proof of Theorem~\ref{thm: improved SV-ITA14} hinges on the fact that the function
$\varphi$ is monotonically increasing. It can be verified
that $\varphi$ is also concave and differentiable. Taking into account these additional
properties of $\varphi$, the bound in Theorem~\ref{thm: improved SV-ITA14} can be tightened
as (see Appendix~\ref{appendix: ub RE-TV-chi}):
\begin{align}
\label{eq: ub RE RPI}
D(P\|Q) \leq & \tfrac12 \left(  \varphi(\beta_1^{-1})-\varphi(\beta_2) - \varphi'(\beta_1^{-1})
\; \beta_1^{-1} \right) \; |P-Q|
+ \varphi'(\beta_1^{-1}) \; \mathbb{E} \bigl[Z(Z-1) \, 1\{Z>1\}\bigr]
\end{align}
which is expressed in terms of the distribution of the relative information. The second summand
in the right side of \eqref{eq: ub RE RPI}
satisfies
\begin{align}
\chi^2(P \| Q) + \tfrac{\beta_2}{2} \, |P-Q| &\leq \mathbb{E} \bigl[ Z(Z-1) \, 1\{Z>1\} \bigr]  \\[0.1cm]
& \leq \chi^2(P \| Q) + \tfrac12 |P-Q|.
\label{eq: issv15}
\end{align}
From \eqref{eq: Z}, \eqref{eq: TV3} and
$Z \in [\beta_2, \beta_1^{-1}]$, the gap between the upper and lower bounds
in \eqref{eq: issv15} satisfies
\begin{align}
\tfrac12 (1-\beta_2) \, |P-Q|
& = (1-\beta_2) \, \mathbb{E}\bigl[(1-Z)^+\bigr] \\
& \leq (1-\beta_2)^2
\end{align}
which is upper bounded by~1, and it is close to zero if $\beta_2 \approx 1$. The combination of
\eqref{eq: ub RE RPI} and \eqref{eq: issv15} leads to
\begin{align}
D(P\|Q) \leq  \tfrac12 \, \Bigl(  \varphi(\beta_1^{-1})-\varphi(\beta_2) + \varphi'(\beta_1^{-1})
\; \bigl(1-\beta_1^{-1}\bigr)\Bigr) \; |P-Q| + \varphi'(\beta_1^{-1}) \cdot \chi^2(P \| Q).
\label{eq: ub RE-TV-chi}
\end{align}
\end{remark}

\begin{remark}
A special case of \eqref{eq: SV-ITA14}, where $Q$ is the uniform distribution over
a set of a finite size, was recently rediscovered in \cite[Corollary~13]{Moser_ISIT2015} based
on results from \cite{HoY_IT2010}.
\end{remark}

\begin{remark}
For $ \varepsilon \in [0,2]$ and a fixed probability measure $Q$, define
\begin{align}
D^*(\varepsilon, Q) = \inf_{P \colon |P-Q| \geq \varepsilon} D(P \| Q).
\end{align}
From Sanov's theorem (see \cite[Theorem~11.4.1]{Cover_Thomas}), $D^*(\varepsilon, Q)$
is equal to the asymptotic exponential decay of the probability that the total variation
distance between the empirical distribution of a sequence of i.i.d. random variables
and the true distribution $Q$ is more than a specified value $\varepsilon$. Bounds on
$D^*(\varepsilon, Q)$ have been shown in \cite[Theorem~1]{BerendHK_IT14}, which, locally,
behave  quadratically in $\varepsilon$ . Although this result was classified in
\cite{BerendHK_IT14} as a reverse Pinsker inequality, note that it differs
from the scope of this section which provides, under suitable conditions, lower bounds on
the total variation distance as a function of the relative entropy.
\end{remark}

\subsection{Lipschitz Constraints}
\label{subsec:Lipschitz}

\begin{definition}  \label{def: Lipschitz}
A function $f \colon \set{B} \to \Reals$, where $\set{B} \subseteq \Reals$, is $L$-Lipschitz if
for all $x, y \in \set{B}$
\begin{align}  \label{eq: Lipschitz}
| f(x) - f(y) | \leq L \; |x-y|.
\end{align}
\end{definition}

The following bound generalizes \cite[Theorem~6]{Dragomir00b} to the non-discrete setting.
\begin{theorem}  \label{thm: generalize Dragomir's inequality}
Let $P \ll Q$ with $\beta_1 \in (0,1)$ and $\beta_2 \in [0,1)$, and
$f \colon [0, \infty) \to \Reals$ be continuous and convex with $f(1)=0$,
and $L$-Lipschitz on $[\beta_2, \beta_1^{-1}]$. Then,
\begin{align}  \label{eq: Dragomir's inequality}
D_f(P \| Q) \leq L \, |P-Q|.
\end{align}
\end{theorem}
\begin{proof}
If  $Y \sim Q$, and $Z$ is given by \eqref{eq: Z} then $f(1)=0$ yields
\begin{align}
D_f(P \| Q)
&= \mathbb{E} \bigl[ f(Z) \bigr] \\
&\leq \mathbb{E} \bigl[ \bigl| f(Z) - f(1) \bigr| \bigr]\\
&\leq L \; \mathbb{E} \bigl[ \bigl|Z - 1\bigr| \bigr]\\
&= L \; |P-Q| \label{siro}
\end{align}
where \eqref{siro} holds due to \eqref{eq: TV1}.
\end{proof}

Note that if $f$ has a bounded derivative on $[\beta_2, \beta_1^{-1}]$, we can choose
\begin{align} \label{eq: esssup finite}
L = \sup_{t \in [\beta_2, \beta_1^{-1}]} |f'(t)| < \infty.
\end{align}

\begin{remark}
In the case $f(t) = t \log t$, $f(0) =0$, \eqref{eq: esssup finite} particularizes to
\begin{align}
L = \max\bigl\{ \bigl| \log(e \beta_2) \bigr|, \, \log(e \beta_1^{-1}) \bigr\}
\end{align}
resulting in a reverse Pinsker inequality which is weaker than that
in \eqref{eq: improved SV-ITA14} by at least a factor of 2.
\end{remark}

%%%%%%%%%%%%%%%%%%%%%%%%%%%%%%%%%%%%%%%%%%%%
\subsection{Finite Alphabet} \label{subsec: revPinsker-finite}
Throughout this subsection, we assume that $P$ and $Q$ are probability measures defined
on a common finite set $\set{A}$, and $Q$ is strictly positive on $\set{A}$, which has
more than one element.
\par
The bound in \eqref{eq: SV-ITA14} strengthens the finite-alphabet bound in
\cite[Lemma~3.10]{Even-Dar07}:
\begin{align}
\label{eq: lemma by Even-Dar et al.}
D(P \| Q) \leq \log \left(\frac1{Q_{\min}}\right) \cdot |P-Q|
\end{align}
with
\begin{align}\label{qmin}
Q_{\min} = \min_{a \in \set{A}} Q(a) \leq \tfrac12.
\end{align}
To verify this, notice that  $\beta_1 \geq Q_{\min}$.
Let $v \colon (0, 1) \to (0, \infty)$ be defined by
$v(t) = \frac1{1-t} \cdot \log \frac1{t}$; since $v$
is a monotonically decreasing and non-negative function,
we can weaken \eqref{eq: SV-ITA14} to write
\begin{align}
D(P \| Q)
&\leq \left( \frac{\log \frac1{Q_{\min}}}{2(1-Q_{\min})} \right) |P-Q| \\
&\leq \log \left( \frac1{Q_{\min}} \right) \cdot |P-Q|
\label{mira}
\end{align}
where \eqref{mira} follows from \eqref{qmin}.

The main result in this subsection is the following bound.
\begin{theorem} \label{thm: UB-RE-FS}
\begin{align}
\label{eq: UB-RE-FS2}
D(P \| Q) \leq \log \left(1 + \frac{|P-Q|^2}{2 Q_{\min}} \right).
\end{align}
Furthermore, if $Q \ll P$ and $\beta_2$ is defined as in \eqref{eq: beta2},
then the following tightened bound holds:
\begin{align}
\label{eq: UB-RE-FS1}
D(P \| Q) &\leq \log \left(1 + \frac{|P-Q|^2}{2 Q_{\min}} \right)
- \frac{\beta_2 \log e}{2} \cdot |P-Q|^2.
\end{align}
\end{theorem}

\begin{proof}
Combining \eqref{grout425 - introduction}
and the following finite-alphabet upper bound on $\chi^2(P \| Q)$ yields  \eqref{eq: UB-RE-FS2}:
\begin{align}
Q_{\min} \; \chi^2(P\|Q)
\label{barsa3}
& = \sum_{a \in \set{A}} \frac{(P(a)-Q(a))^2}{Q(a)/Q_{\min}}  \\[0.1cm]
\label{elle}
& \leq \sum_{a \in \set{A}} \bigl(P(a)-Q(a)\bigr)^2 \\[0.1cm]
& \leq \max_{x \in \set{A}} |P(x)-Q(x)| \, \sum_{a \in \set{A}}
\bigl|P(a)-Q(a)\bigr| \nonumber \\[0.1cm]
\label{eq: 1st-ubchi}
& = |P-Q| \; \max_{a \in \set{A}} |P(a)-Q(a)| \\
\label{juventus1}
&\leq \tfrac12 \, |P-Q|^2.
\end{align}

If $P \ll \gg Q$, then \eqref{eq: UB-RE-FS1} follows by combining \eqref{juventus1} and
\begin{align}
\label{eq: lb2chi}
 \chi^2(P\|Q)
& \geq \exp \Bigl( D(P\|Q) + \beta_2 \, D(Q \|P) \Bigr) - 1  \\[0.1cm]
\label{eq2: lb2chi}
& \geq \exp \left( D(P \|Q) + \tfrac12 \, |P-Q|^2 \, \beta_2 \, \log e \right) - 1
\end{align}
where \eqref{eq: lb2chi} follows by rearranging \eqref{lbchi}, and \eqref{eq2: lb2chi}
follows from \eqref{eq: Pinsker}.
\end{proof}

\begin{remark}
It is easy to check that Theorem~\ref{thm: UB-RE-FS} strengthens the bound by Csisz\'ar and Talata~\eqref{eq: CsTa}
by at least a factor of~2 since upper bounding  the logarithm in \eqref{eq: UB-RE-FS2} gives
\begin{align} \label{betterthanCT}
D(P \| Q) \leq \frac{\log e}{2 Q_{\min}} \cdot {|P-Q|^2}.
\end{align}
\end{remark}

\begin{remark}
In the finite-alphabet case, similarly to \eqref{eq: UB-RE-FS2}, one can obtain another upper bound on $D(P\|Q)$
as a function of the $\ell_2$ norm
$\| P - Q \|_2$:
\begin{align}  \label{eq: UB-RE-FS-looser}
D(P \| Q) \leq \frac{1}{Q_{\min}} \cdot \|P-Q\|_2^2 \, \log e
\end{align}
which appears in the proof of Property~4 of \cite[Lemma~7]{TomamichelT_IT13},
and also used in \cite[(174)]{KostinaV15}.
Furthermore, similarly to \eqref{eq: UB-RE-FS1}, the following tightened bound holds if $P \ll \gg Q$:
\begin{align}  \label{eq: UB-RE-FS}
D(P \| Q) \leq \log \left(1 + \frac{\|P-Q\|_2^2}{Q_{\min}} \right) - \frac{\beta_2 \log e}{2} \cdot \|P-Q\|_2^2
\end{align}
which follows by combining \eqref{elle}, \eqref{eq2: lb2chi}, and the inequality
$\|P-Q\|_2 \leq |P-Q|$.
\end{remark}

\begin{remark}
Combining \eqref{eq: Pinsker} and \eqref{juventus1} yields that if $P \neq Q$ are defined on a common finite set, then
\begin{align} \label{eq: lb on RE/chi^2 - finite alphabet}
\frac{D(P \| Q)}{\chi^2(P \| Q)} \geq Q_{\min} \, \log e
\end{align}
which at least doubles the lower bound in \cite[Lemma~6]{Makur15}. This, in turn,
improves the tightened upper bound on the strong data processing inequality constant
in \cite[Theorem~10]{Makur15} by a factor of~2.
\end{remark}

\begin{remark}
Reverse Pinsker inequalities have been also derived in quantum information theory
(\cite{AE1, AE2}), providing upper bounds on the relative entropy of two quantum
states as a function of the trace norm distance when the minimal eigenvalues of the
states are positive (c.f. \cite[Theorem~6]{AE1} and \cite[Theorem~1]{AE2}).
When the variational distance is much smaller than the minimal eigenvalue (see
\cite[Eq.~(57)]{AE1}), the latter bounds have a quadratic scaling in this distance,
similarly to \eqref{eq: UB-RE-FS2}; they are also inversely proportional to the
minimal eigenvalue, similarly to the dependence of \eqref{eq: UB-RE-FS2} in $Q_{\min}$.
\end{remark}

\begin{remark}
Let $P$ and $Q$ be probability distributions defined on an arbitrary alphabet
$\set{A}$. Combining Theorems~\ref{thm: bounds RE and dual}
and~\ref{thm: improved SV-ITA14} leads to a derivation of an upper bound
on the difference $D(P\|Q) - D(Q\|P)$ as a function of $|P-Q|$ as long as
$P \ll \gg Q$ and the relative information $\imath_{P\|Q}$ is bounded away from $- \infty$ and $+\infty$.
Furthermore, another upper bound on the difference of the relative entropies
can be readily obtained by combining Theorems~\ref{thm: bounds RE and dual}
and~\ref{thm: UB-RE-FS} when $P$ and $Q$ are probability measures defined on
a finite alphabet. In the latter case, combining Theorem~\ref{thm: bounds RE and dual}
and \eqref{eq: UB-RE-FS} also yields another upper bound on $D(P\|Q) - D(Q\|P)$
which scales quadratically with $\|P-Q\|_2$. All these bounds form a counterpart
to \cite[Theorem 1]{Audeaert13} and Theorem~\ref{thm: bounds RE and dual},
providing measures of the asymmetry of the relative entropy when the relative
information is bounded.
\end{remark}

\subsection{Distance From the Equiprobable Distribution}
\label{subsec: Distance From the Equiprobable Distribution}
If $P$ is a distribution on a finite set $\set{A}$, $H(P)$ gauges the ``distance" from
$\mathsf{U}$, the equiprobable distribution defined on $\set{A}$, since
$H(P) = \log | \set{A} | -  D(P \| \mathsf{U})$.
Thus, it is of interest to explore the relationship between $H(P) $ and $|P-\mathsf{U}|$.
Next, we determine the exact locus of the points $(H(P), \, |P-\mathsf{U}|)$ among all
probability measures $P$ defined on $\set{A}$, and this region is compared to
upper and lower bounds on $|P-\mathsf{U}|$ as a function of $H(P)$.
As usual,  $h(x)$ denotes the continuous extension of $ -x \log x - (1-x) \log(1-x)$ to $x \in [0,1]$
 and $d(x \| y)$ denotes the binary relative entropy in
\eqref{eq: binary RE}.
\begin{theorem} \label{thm: exact locus}
Let $\mathsf{U}$ be the equiprobable distribution on a
$\{ 1 , \ldots , |\set{A}|\}$, with $1 < |\set{A} | < \infty$.
\begin{enumerate}[a)]
\item  \label{thm: exact locus: parta}
For $\Delta \in (0, 2(1-|\set{A}|^{-1})]$,\footnote{There is no $P$ with $|P -\mathsf{U}| > 2(1-|\set{A}|^{-1})$.}
\begin{align}
\max_{P \colon |P -\mathsf{U}| = \Delta} H(P)
\label{eq: maximal entropy for given TV}
= \log |\set{A}| - \min_{m} \,
d \left( \tfrac{m}{|\set{A}|} + \tfrac12 \Delta \, \big\| \, \tfrac{m}{|\set{A}|} \right)
\end{align}
where the minimum in the right side of \eqref{eq: maximal entropy for given TV}
is over
\begin{align}\label{conditionmdelta}
m \in \{1, \ldots, | \set{A}| - \bigl\lceil \tfrac12 \Delta |\set{A}| \bigr\rceil \}.
\end{align}
Denoting such an integer by $m_\Delta$, the maximum in the left side of
\eqref{eq: maximal entropy for given TV} is attained by
\begin{align} \label{eq: P for maximizing the TV}
P_{\Delta}(\ell) =
\left\{
\begin{array}{ll}
\hspace*{-0.1cm} |\set{A}|^{-1} + \frac{\Delta}{2m_\Delta}
& \hspace*{-0.1cm} \ell \in \{1, \ldots, m_\Delta\}, \\[0.2cm]
\hspace*{-0.1cm} |\set{A}|^{-1} - \frac{\Delta}{2(|\set{A}|-m_\Delta)},
& \hspace*{-0.1cm} \ell \in \{m_\Delta+1, \ldots, |\set{A}|\}.
\end{array}
\right.
\end{align}
\item \label{thm: exact locus: partb}
Let
\begin{align} \label{eq: h_k}
h_k =
\left\{
\begin{array}{ll}
\hspace*{-0.1cm} 0, & \hspace*{-0.2cm} k=0 \\[0.1cm]
\hspace*{-0.1cm} h \bigl(|\set{A}|^{-1} k \bigr)
+ |\set{A}|^{-1} \, k \log k,
& \hspace*{-0.2cm} k \in \{1, \ldots, |\set{A}|-2\} \\[0.1cm]
\hspace*{-0.1cm} \log |\set{A}|,
& \hspace*{-0.2cm} k = |\set{A}|-1.
\end{array}
\right.
\end{align}
If $H \in [h_{k-1}, h_{k})$ for $k \in \{1, \ldots, |\set{A}|-1\}$, then
\begin{align}
\min_{P \colon H(P) = H}  |P-\mathsf{U}| = 2\bigl(1-(k+\theta) \, |\set{A}|^{-1}\bigr)
\end{align}
which is achieved by
\begin{align} \label{eq: P for minimizing the TV}
P^{(k)}_{\theta}(\ell) =
\left\{
\begin{array}{ll}
\hspace*{-0.1cm} 1-(k-1+\theta) \, |\set{A}|^{-1}, & \hspace*{-0.2cm} \ell = 1 \\[0.1cm]
\hspace*{-0.1cm} |\set{A}|^{-1}, & \hspace*{-0.2cm} \ell \in \{2, \ldots, k\}, \\[0.1cm]
\hspace*{-0.1cm} \theta \, |\set{A}|^{-1}, & \hspace*{-0.2cm} \ell = k+1 \\[0.1cm]
\hspace*{-0.1cm} 0, & \hspace*{-0.2cm} \ell \in \{k+2, \ldots, |\set{A}|\}
\end{array}
\right.
\end{align}
where $\theta \in [0, 1)$ is chosen so that $H(P^{(k)}_\theta ) = H$.
\end{enumerate}
\end{theorem}

\begin{proof}
See Appendix~\ref{appendix: exact locus}.
\end{proof}

\begin{remark}
For probability measures defined on a 2-element set $\set{A}$, the maximal and minimal values of $|P-\mathsf{U}|$ in
Theorem~\ref{thm: exact locus} coincide. This can be verified since, if $P(1)=p$ for $p \in [0,1]$, then
$|P-\mathsf{U}| = |1-2p|$ and $H(P) = h(p)$. Hence, if $|\set{A}|=2$ and $H(P)=H \in [0,\log 2]$, then
\begin{align}
|P-\mathsf{U}| = 1 - 2 h^{-1}(H)
\end{align}
where $h^{-1} \colon [0,\log 2] \to [0, \tfrac12]$ denotes the inverse of the binary entropy function.
\end{remark}

\par
Results on the more general problem of finding bounds on $|H( P ) - H(Q)|$ based on $|P-Q|$
can be found in \cite[Theorem~17.3.3]{Cover_Thomas}, \cite{HoY_IT2010}, \cite{Sason_IT2013}
and \cite{Zhang_IT2007}. Most well-known among them is
\begin{align} \label{hphqhv}
| H (P) - H (Q) | \leq | P - Q | \log \left(\frac{| \set{A}|}{|P-Q|}\right)
\end{align}
which holds if $P,Q$ are probability measures defined on a finite set $\set{A}$ with
$\bigl|P(a)-Q(a)\bigr| \leq \frac12$ for all $a \in \set{A}$ (see
\cite{Verdu_book}, and \cite[Lemma~2.7]{Csiszar_Korner} with a stronger sufficient condition).
Particularizing \eqref{hphqhv} to the case where $Q=\mathsf{U}$, and
$\left|P(a)-\tfrac1{|\set{A}|}\right| \leq \tfrac12$ for all $a \in \set{A}$ yields
\begin{align}
\label{hphqhv-U}
H(P) \geq \log|\set{A}| - |P-\mathsf{U}| \, \log\left(\frac{|\set{A}|}{|P-\mathsf{U}|}\right),
\end{align}
a bound which finds use in information-theoretic security \cite{CsiszarN04}.

Particularizing \eqref{eq: Pinsker}, \eqref{eq: BretagnolleH79}, and \eqref{eq: UB-RE-FS2} we obtain
\begin{align}
\label{eq: 1st ubtv-uniform}
H(P) &\leq \log |\set{A}| - \tfrac12 \, |P-\mathsf{U}|^2 \, \log e, \\
\label{eq: 2nd ubtv-uniform}
H(P) &\leq \log |\set{A}| + \log \left( 1 - \tfrac14 \, |P-\mathsf{U}|^2 \right), \\
\label{eq2: lbtv-uniform}
H(P) &\geq \log |\set{A}| - \log\left(1+ \tfrac{|\set{A}|}{2} \; |P-\mathsf{U}|^2\right).
\end{align}
If either $|\set{A}|=2$ or $8 \leq |\set{A}| \leq 102$, it can be checked that the lower bound on
$H(P)$ in \eqref{hphqhv-U} is worse than \eqref{eq2: lbtv-uniform}, irrespectively of $|P-\mathsf{U}|$
(note that $0 \leq |P-\mathsf{U}| \leq 2(1-|\set{A}|^{-1} )$).

The exact locus of $(H(P), |P-\mathsf{U}|)$ among all the probability measures $P$ defined on
a finite set $\set{A}$ (see Theorem~\ref{thm: exact locus}), and the
bounds in \eqref{eq: 1st ubtv-uniform}--\eqref{eq2: lbtv-uniform} are illustrated
in Figure~\ref{figure:equiprobable} for $|\set{A}| = 4$  and $|\set{A}| = 256$.
For $|\set{A}| = 4$, the lower bound in \eqref{eq2: lbtv-uniform} is tighter than
\eqref{hphqhv-U}. For $|\set{A}| = 256$, we only show \eqref{eq2: lbtv-uniform}  in
Figure~\ref{figure:equiprobable} as in this case  \eqref{hphqhv-U} offers a very minor improvement in a small range.
As the cardinality of the set $\set{A}$ increases, the
gap between the exact locus (shaded region) and the upper bound obtained from
\eqref{eq: 1st ubtv-uniform} and \eqref{eq: 2nd ubtv-uniform} (Curves~(a) and (b), respectively)
decreases, whereas the gap between the exact locus and the lower bound in \eqref{eq2: lbtv-uniform}
(Curve~(c)) increases.

\begin{figure}[h]
\begin{center}
\includegraphics[width=7.8cm]{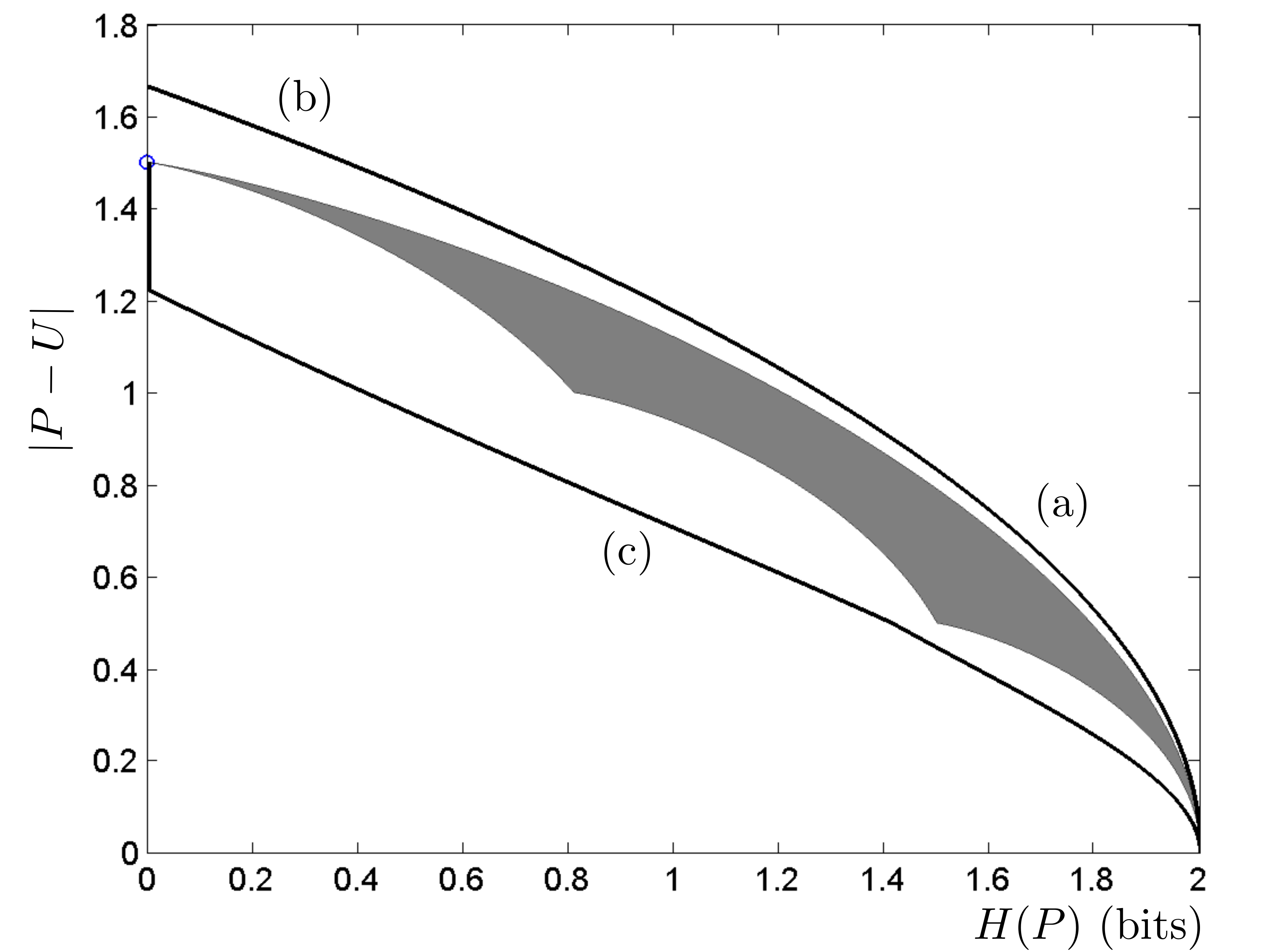} \\[0.2cm]
\includegraphics[width=7.8cm]{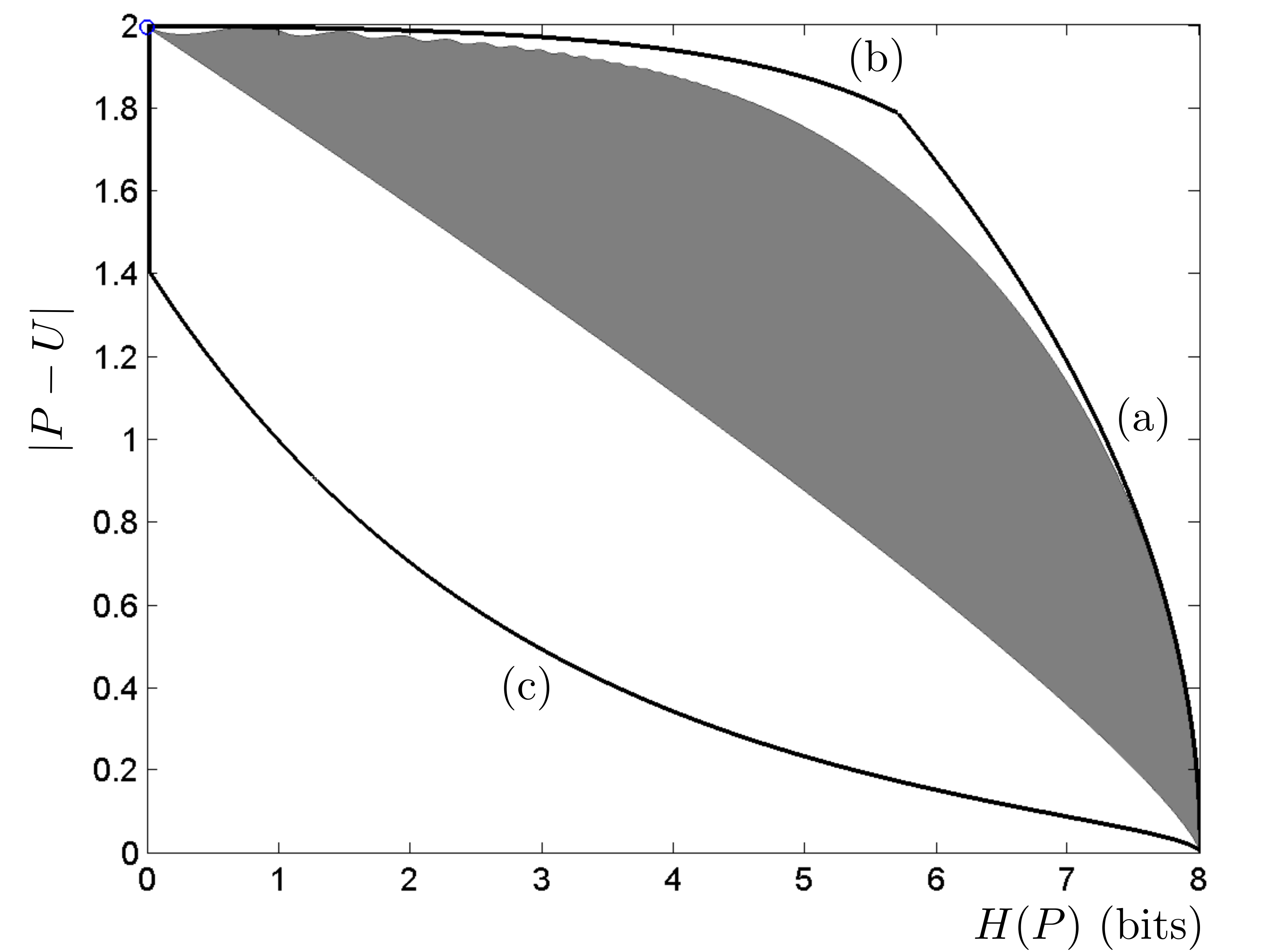}
\end{center}
\caption{\label{figure:equiprobable}
The exact locus of $(H(P), |P-\mathsf{U}|)$ among all the probability
measures $P$ defined on a finite set $\set{A}$,
and bounds on $|P-\mathsf{U}|$ as a function of $H(P)$ for $|\set{A}| = 4$ (left plot),
and $|\set{A}| = 256$ (right plot). The point $(H(P), |P-\mathsf{U}|)=(0, \, 2(1-|\set{A}|^{-1}) )$
is depicted on the $y$-axis. In the two plots, Curves~(a), (b) and (c)
refer, respectively, to \eqref{eq: 1st ubtv-uniform}, \eqref{eq: 2nd ubtv-uniform} and
\eqref{eq2: lbtv-uniform}; the exact locus (shaded region) refers to Theorem~\ref{thm: exact locus}.}
\end{figure}

%%%%%%%%%%%%%%%%%%%%%%%%%%%%%%%%%%%%%%

\subsection{The Exponential Decay of the Probability of Non-Strongly Typical Sequences}
\label{subsec: non-typical sequences}
The objective is to bound the function
\begin{align}\label{def:ldeltaq}
L_{\delta} ( Q ) = \min_{P \not \in \mathcal{T}_\delta (Q )} D(P\|Q)
\end{align}
where the subset of probability measures on $ (\set{A}, \mathscr{F})$
which are $\delta$-close to $Q$ is given by
\begin{align}
\mathcal{T}_\delta (Q) = \Bigl\{P
\colon \forall \, a \in \set{A}, \; \; |P(a)-Q(a)| \leq \delta \, Q(a) \Bigr\}.
\end{align}
Note that $(a_1, \ldots , a_n )$ is strongly $\delta$-typical according to $Q$ if its
empirical distribution belongs to $\mathcal{T}_\delta (Q )$. According to Sanov's theorem
(e.g. \cite[Theorem~11.4.1]{Cover_Thomas}), if the random variables are independent
and distributed according to $Q$, then the probability that
$(Y_1, \ldots , Y_n)$, is not $\delta$-typical
vanishes exponentially with exponent $L_{\delta} ( Q )$.
\par
To state the next result, we invoke the following notions from \cite{OrdentlichW_IT2005}.
Given a probability measure $Q$, its \textit{balance coefficient} is given by
\begin{align} \label{balanceOW}
\beta_Q = \inf_{\set{F} \in \mathscr{F} \colon Q(\set{F}) \geq \frac12} Q(\set{F}).
\end{align}
The function $\phi\colon(0, \tfrac12] \to [\tfrac12 \log e, \infty)$ is a monotonically
decreasing and convex function, which is given by
\begin{align}
\label{eq: phi refined pinsker}
\phi(p) = \left\{
\begin{array}{ll}
\frac1{4(1-2p)} \, \log \left( \frac{1-p}{p} \right), & p \in
\bigl(0, \tfrac12 \bigr), \\[0.1cm]
\tfrac12 \log e, & p=\tfrac12 .
\end{array}
\right.
\end{align}

\begin{theorem}
If $Q_{\min}>0$, then
\begin{align}
\label{potalo}
\phi(1 - \beta_Q) \, Q_{\min}^2 \, \delta^2
&\leq L_{\delta} ( Q ) \\
&\leq \log \left(1 + 2 Q_{\min} \, \delta^2 \right)
\label{potaup}
\end{align}
where \eqref{potaup} holds if $\delta < \frac{1-Q_{\max}}{Q_{\min}}$.
\end{theorem}

\begin{proof}
The following refinement of Pinsker's inequality \eqref{eq: Pinsker}
was derived in \cite[Section~4]{OrdentlichW_IT2005}:
\begin{align}  \label{eq: OrdentlichW}
\phi(1-\beta_Q) \; |P-Q|^2 \leq D(P \| Q).
\end{align}
Note that if $Q_{\min} > 0$ then $\beta_Q \leq 1 - Q_{\min} < 1$, and
$\phi(1-\beta_Q)$ is well defined and finite.
If $P \not\in \mathcal{T}_\delta (Q )$, the simple bound
\begin{align}
|P - Q| > \delta Q_{\min}
\end{align}
together with \eqref{def:ldeltaq} and \eqref{eq: OrdentlichW} yields \eqref{potalo}.
\par
The upper bound \eqref{potaup} follows from \eqref{eq: UB-RE-FS2}
and the fact that if $\delta < \frac{1-Q_{\max}}{Q_{\min}}$, then
\begin{align} \label{review2}
\inf_{P \not\in \mathcal{T}_\delta (Q )} | P - Q | = 2 \delta Q_{\min}.
\end{align}
To verify \eqref{review2}, note that for every $P \not\in \mathcal{T}_\delta (Q )$,
there exists $a \in \set{A}$ such that $|P(a)-Q(a)| > \delta \, Q(a)$,
which implies that $| P - Q | > 2 \delta \, Q(a) \geq 2 \delta Q_{\min}$,
thereby establishing $\geq$ in \eqref{review2}. To show equality, let
$a_0 \in \set{A}$ be such that $Q(a_0) = Q_{\min}$, and let $a_1 \neq a_0$;
since by assumption $\delta < \frac{1-Q_{\max}}{Q_{\min}}$, we have
$Q(a_1) + \delta \, Q_{\min} \leq Q_{\max} + \delta \, Q_{\min} < 1$.
Let
\begin{align} \label{eq: P construction}
P(a) =
\left\{
\begin{array}{ll}
(1-\delta - \varepsilon) \, Q_{\min} & a = a_0 \\
Q(a_1) + (\delta + \varepsilon) \, Q_{\min} & a = a_1 \\
Q(a)
& \mbox{otherwise}
\end{array}
\right.
\end{align}
for a sufficiently small $\varepsilon > 0$ so that \eqref{eq: P construction}
is a probability measure. Then, $P \not\in \mathcal{T}_\delta (Q )$ and
$|P - Q| = 2 (\delta + \varepsilon) \, Q_{\min}$, which verifies the equality in \eqref{review2}
by letting $\varepsilon \downarrow 0$.
\end{proof}

\begin{remark}
If $\delta < \frac{1-Q_{\max}}{Q_{\min}}$, the ratio between the upper and lower bounds
in \eqref{potaup}, satisfies
\begin{align}  \label{eq: exponents' ratio}
 \frac1{Q_{\min}} \cdot \frac{\log e}{2 \, \phi(1-\beta_Q)} \cdot
\frac{\log \left(1 + 2 Q_{\min} \, \delta^2 \right)}{\tfrac12
\, Q_{\min} \, \delta^2 \, \log e} \leq \frac{4}{Q_{\min}}
\end{align}
where \eqref{eq: exponents' ratio} follows from the fact that its second
and third factors are less than or equal to~1 and~4, respectively. Note
that both bounds in \eqref{potaup} scale like $\delta^2$ for $\delta \approx 0$.
\end{remark}

%%%%%%%%%%%%%%%%%%%%%%%%%%%%%%%%%%%%%%%%%%%
\section{The $E_\gamma$ Divergence}
\label{sec:EG}

\subsection{Basic Properties} \label{subsec: basic properties}

Generalizing the total variation distance, the $E_\gamma$ divergence in
\eqref{eq:Eg f-div} is an $f$-divergence whose utility in information theory
has been exemplified in \cite{CZK98}, \cite{LiuCV1_IT15}, \cite{LiuCV1_ISIT15},
\cite{LiuCV2_ISIT15}, \cite{PPV10},\cite{PV-Allerton10}, \cite{PW15}.

In this subsection, we provide some basic properties of the $E_\gamma$ divergence,
 which are essential to Sections~\ref{saens}--\ref{subsec: re-ris}.
The reader is referred to \cite[Sections~2.B, 2.C]{LiuCV1_IT15} for some additional
basic properties of the $E_{\gamma}$ divergence. We assume throughout this section
that $\gamma \geq 1$.

Let $P \ll Q$. The $E_\gamma$ divergence in \eqref{eq:Eg f-div}
can be expressed in the form
\begin{align}
\label{eq:EG}
E_{\gamma}(P \| Q)
&= \prob \bigl[ \imath_{P\|Q}(X) > \log \gamma \bigr]
- \gamma \, \prob \bigl[ \imath_{P\|Q}(Y) > \log \gamma \bigr] \\[0.1cm]
&= \max_{\set{F} \in \mathscr{F}} \bigl( P(\set{F}) - \gamma \, Q(\set{F}) \bigr)
\label{eq2: EG}
\end{align}
where $X \sim P$ and $Y \sim Q$, and \eqref{eq2: EG} follows from the Neyman-Pearson lemma.

Although the $E_\gamma$ divergence generalizes the total variation distance, $E_\gamma(P \| Q) =0$
for $\gamma > 1$ does not imply $P = Q$ since in that case \eqref{eq: f for EG} is not strictly
convex at $t=1$ (see Proposition~\ref{prop: fGibbs}). This is illustrated in the following example.
\begin{example}\label{example:zeroEg}
Let $\gamma > 1$, and let $P$ and $Q$ be probability measures defined on $\set{A} = \{0, 1\}$:
\begin{align}
& P(0) = \frac{1+\gamma}{2\gamma}, \quad Q(0) = \frac1{\gamma}.
\end{align}
Since $\imath_{P\|Q}(x) = \log \gamma \; 1\{x=0\} - \log 2 < \log \gamma$
for all $x \in \set{A}$, \eqref{eq:EG} implies that $E_{\gamma}(P\|Q) = 0$.
\end{example}

\par
The monotonicity of the $E_\gamma$ divergence in $\gamma \in [1, \infty)$ holds since
$f_{\gamma_1}(t) \leq f_{\gamma_2}(t)$ for all $t>0$ with $f_\gamma(t) = (t-\gamma)^+$
and $\gamma_1 \geq \gamma_2 \geq 1$. Therefore,
\begin{align} \label{eq: monotonicity E_gamma}
\frac{E_{\gamma_1}(P \| Q)}{ E_{\gamma_2}(P \| Q)} \leq 1.
\end{align}
Although Theorem~\ref{theorem: tight bound}\ref{theorem: tight bound: partb})
does not apply in order to prove that $1$ is the best constant in
\eqref{eq: monotonicity E_gamma}, we can verify it by defining $P$ and $Q$
on $\set{A} = \{0,1\}$ with $P(0) = \tfrac12$ and $Q(0) = \varepsilon > 0$.
This yields that if $\gamma_1 \geq \gamma_2 \geq 1$, then for all
$\varepsilon \in (0, \frac1{2\gamma_1})$,
\begin{align}
\label{leoneck}
\frac{E_{\gamma_1}(P \| Q)}{E_{\gamma_2}(P \| Q)} =
\frac{1- 2 \varepsilon \gamma_1}{1- 2 \varepsilon \gamma_2},
\end{align}
yielding the optimality of the constant in the right side of
\eqref{eq: monotonicity E_gamma} by letting $\varepsilon \downarrow 0$
in \eqref{leoneck}.

\par
From \eqref{eq2: EG}, the following inequality holds:
If $P \ll R \ll Q$, and $\gamma_1, \gamma_2 \geq 1$ then
\begin{align} \label{eq1: triangleEG}
E_{\gamma_1 \gamma_2}(P\|Q) \leq E_{\gamma_1}(P\|R) + \gamma_1 \, E_{\gamma_2}(R\|Q).
\end{align}
Letting $\gamma_1=1$ in \eqref{eq1: triangleEG} (see \eqref{eq:EG-TV}) and $\gamma_2 = \gamma$ yield
\begin{align} \label{eq2: triangleEG}
E_{\gamma}(P\|Q) - E_{\gamma}(R\|Q) \leq \tfrac12 |P-R|.
\end{align}

Generalizing the fact that $E_1(P \| Q) = \tfrac12 | P - Q |$, the following identity is
a special case of \cite[Corollary~2.3]{Guntuboyina11}:
\begin{align} \label{eq:EG-TV3}
\min_{R: P, Q \ll R} \Bigl\{ E_\gamma(P \| R) + E_\gamma(Q \| R) \Bigr\}
= \left(1 -\gamma  + \tfrac12 {|P-Q|}\right)^+
\end{align}
while \cite[(21)]{LiuCV1_IT15} states that
\begin{align} \label{eq: Jingbo}
\min_{R \ll P, Q} \Bigl\{ E_{\gamma}(R\|P) + E_{\gamma}(R\|Q) \Bigr\}
\geq \left( 1-\gamma + \tfrac{\gamma}{2} \, |P-Q |  \right)^+,
\end{align}
which implies, by taking $R = P$,
\begin{align}
\label{eq10b: EG}
& \left(1-\gamma +\tfrac{\gamma}{2} \, |P-Q| \right)^+ \leq E_\gamma(P \| Q).
\end{align}

We end this subsection with the following result.
\begin{theorem}
If $P \ll Q$ and $Y \sim Q$, then
\begin{align}
\label{eq1: abs}
 \mathbb{E} \left[ \bigl| \exp\bigl(\imath_{P\|Q}(Y)\bigr) - \gamma \bigr| \right] &=
2 E_\gamma(P \| Q) + \gamma - 1, \\[0.1cm]
\label{eq1: max}
 \mathbb{E} \left[ \max \bigl\{\gamma, \, \exp\bigl(\imath_{P\|Q}(Y)\bigr) \bigr\} \right]
&= \gamma + E_\gamma(P \| Q), \\[0.1cm]
\label{eq1: min}
 \mathbb{E} \left[ \min \bigl\{\gamma, \, \exp\bigl(\imath_{P\|Q}(Y)\bigr) \bigr\} \right]
&= 1 - E_\gamma(P \| Q).
\end{align}
\end{theorem}
\begin{proof}
The identity $|z| = 2 (z)^+ - z$, for all $z \in \Reals$, and \eqref{eq:Eg f-div} are used to prove \eqref{eq1: abs}:
\begin{align}
\mathbb{E} \left[ \bigl| \exp\bigl(\imath_{P\|Q}(Y)\bigr) - \gamma \bigr| \right]
& = 2 \mathbb{E} \left[ \bigl( \exp\bigl(\imath_{P\|Q}(Y)\bigr) - \gamma \bigr)^+ \right]
- \mathbb{E} \left[ \exp\bigl(\imath_{P\|Q}(Y)\bigr) - \gamma \right] \\[0.1cm]
& = 2 E_\gamma(P \| Q) + \gamma - 1.
\end{align}
Eqs.~\eqref{eq1: max} and \eqref{eq1: min} follow from \eqref{eq1: abs}, and the identities
\begin{align}
\label{eq2: max}
& \max\{x_1, x_2\} = \tfrac12 \bigl[x_1+x_2+|x_1-x_2|\bigr], \\
\label{eq2: min}
& \min\{x_1, x_2\} = \, \tfrac12 \bigl[x_1+x_2-|x_1-x_2|\bigr]
\end{align}
for all $x_1, x_2 \in \Reals$.
\end{proof}
\begin{remark}
In view of \eqref{eq:EG-TV}, it follows that \eqref{eq1: abs} and \eqref{eq1: max}
are specialized respectively to \eqref{eq: TV1} and \cite[(20)]{Harremoes_2008} by letting $\gamma=1$.
\end{remark}

\subsection{An Integral Representation of $f$-divergences} \label{saens}

In this subsection we show that
$$\bigl\{(E_\gamma ( P \| Q ), E_\gamma (Q \| P)), \gamma \geq 1\bigr\}$$
uniquely determines $D( P \| Q )$, $\mathscr{H}_{\alpha}(P \| Q)$,
as well as any other $f$-divergence with twice differentiable $f$.

\begin{proposition} \label{prop: integral}
Let $P \ll \gg Q$, and let $f \colon (0, \infty) \to \Reals$ be convex and twice differentiable
with $f(1)=0$. Then,
\begin{align}  \label{eq: integral}
D_f(P\|Q) = \int_1^{\infty} & \bigr( f''(\gamma) \, E_\gamma(P\|Q)
+ \gamma^{-3} f''(\gamma^{-1}) \, E_\gamma(Q\|P) \bigr) \, \text{d}\gamma.
\end{align}
\end{proposition}

\begin{proof}
From \cite[Theorem~11]{LieseV_IT2006}, if $f \colon (0, \infty) \to \Reals$ is a convex function with
$f(1)=0$, then\footnote{See also \cite[Theorem~1]{OV93} for an
earlier representation of $f$-divergence as an averaged DeGroot statistical information.}
\begin{align} \label{eq1: LieseV06}
D_f(P\|Q) = \int_0^1 \mathcal{I}_p(P\|Q) \, \text{d}\Gamma_f(p)
\end{align}
where $\Gamma_f$ is the $\sigma$-finite measure defined on Borel subsets of $(0,1)$ by
\begin{align} \label{eq2: LieseV06}
\Gamma_f\bigl((p_1, p_2]\bigr) = \int_{p_1}^{p_2} \frac1{p} \, \text{d}g_f(p)
\end{align}
for the non-decreasing function
\begin{align}  \label{eq3: LieseV06}
g_f(p) = -f'_{+}\left(\frac{1-p}{p}\right), \quad p \in (0,1)
\end{align}
where $f'_{+}$ denotes the right derivative of $f$.

The DeGroot statistical information in \eqref{eq:DG f-div} has the following operational role
\cite{DeGroot62}, which is used in this proof. Assume hypotheses $H_0$ and $H_1$ have a-priori
probabilities $p$ and $1-p$, respectively, and let $P$ and $Q$ be the conditional
probability measures of an observation $Y$ given $H_0$ or $H_1$. Then, $\mathcal{I}_p(P\|Q)$
is equal to the difference between the minimum error probabilities when the most likely
{\em a-priori} hypothesis is selected, and when the most likely {\em a posteriori} hypothesis
is selected. This measure therefore quantifies the value of the observations for the task of
discriminating between the hypotheses. From the operational role of this measure, it follows
that if $P \ll \gg Q$
\begin{align} \label{eq: symmetry DG}
\mathcal{I}_p(P\|Q) = \mathcal{I}_{1-p}(Q\|P).
\end{align}

The $E_\gamma$ divergence and DeGroot statistical information are related by
\begin{align}
\label{eq: DG-EG}
\mathcal{I}_p(P\|Q) = \left\{
\begin{array}{ll}
p E_{\frac{1-p}{p}}(P\|Q), & \quad \mbox{$p \in \bigl(0, \tfrac12]$} \\[0.2cm]
(1-p) E_{\frac{p}{1-p}}(Q\|P), & \quad \mbox{$p \in \bigl[\tfrac12, 1)$.}
\end{array}
\right.
\end{align}
The expression for $p \in \bigl(0, \tfrac12]$ follows from the fact that the
functions that yield $E_\gamma$ and $\mathcal{I}_p$ in \eqref{eq: f for EG}
and \eqref{eq: f for DG}, respectively, satisfy
\begin{align} \label{eq1: DG-EG}
\phi_p = p f_{\frac{1-p}{p}}.
\end{align}
The remainder of \eqref{eq: DG-EG} follows in view of \eqref{eq: symmetry DG}.

Specializing \eqref{eq1: LieseV06} to a twice differentiable $f$ gives
\begin{align}
\label{eq4: LieseV06}
D_f(P\|Q)
&= \int_0^1  \mathcal{I}_p(P\|Q) \cdot \frac1{p^3}
\, f''\left(\frac{1-p}{p}\right) \, \text{d}p \\[0.1cm]
\label{eq5: LieseV06}
&= \int_0^{\frac12}  \mathcal{I}_p(P\|Q) \cdot \frac1{p^3}
\, f''\left(\frac{1-p}{p}\right) \, \text{d}p
+ \int_{\frac12}^1 \mathcal{I}_p(P\|Q) \cdot \frac1{p^3}
\, f''\left(\frac{1-p}{p}\right) \, \text{d}p \\[0.1cm]
\label{eq6: LieseV06}
&= \int_0^{\frac12} E_{\frac{1-p}{p}}(P\|Q) \cdot \frac1{p^2}
\, f''\left(\frac{1-p}{p}\right) \, \text{d}p +
\int_{\frac12}^1 E_{\frac{p}{1-p}}(Q\|P) \cdot \frac{1-p}{p^3}
\, f''\left(\frac{1-p}{p}\right) \, \text{d}p \\[0.1cm]
\label{eq7: LieseV06}
&= \int_1^{\infty} E_{\gamma}(P\|Q) \, f''(\gamma) \, \text{d}\gamma
+ \int_0^1 \gamma E_{\gamma^{-1}}(Q\|P) \, f''(\gamma) \, \text{d}\gamma \\[0.1cm]
\label{eq8: LieseV06}
&= \int_1^{\infty} \Bigl[f''(\gamma) \, E_{\gamma}(P\|Q) + \gamma^{-3}
\, f''(\gamma^{-1}) \, E_{\gamma}(Q\|P) \Bigr] \, \text{d}\gamma
\end{align}
where \eqref{eq4: LieseV06} follows from \eqref{eq1: LieseV06}--\eqref{eq3: LieseV06};
\eqref{eq5: LieseV06} follows by splitting the interval of integration into two parts;
\eqref{eq6: LieseV06} follows from \eqref{eq: DG-EG};
\eqref{eq7: LieseV06} follows by the substitution $\gamma = \tfrac{1-p}{p}$, and
\eqref{eq8: LieseV06} follows by changing the variable of integration  $t = \frac1{\gamma}$
in the second integral in \eqref{eq7: LieseV06}.
\end{proof}

Particularizing Proposition~\ref{prop: integral} to the most salient $f$-divergences we obtain
(cf. \cite[(84)--(86)]{LieseV_IT2006} for alternative integral representations as a function of
DeGroot statistical information)
\begin{align}
D(P\|Q)
\label{eq: int re}
&= \log e \, \int_1^{\infty} \left( \gamma^{-1} {E_{\gamma}(P\|Q)} +
\gamma^{-2}{E_{\gamma}(Q\|P)} \right) \, \text{d}\gamma, \\[0.2cm]
\label{eq: int H}
\mathscr{H}_{\alpha}(P \| Q)
&= \alpha \int_1^{\infty}
\left( \gamma^{\alpha-2} \, E_{\gamma}(P\|Q) + \gamma^{-\alpha-1}
\, E_{\gamma}(Q\|P) \right) \, \text{d}\gamma,
\end{align}
and specializing \eqref{eq: int H} to $\alpha=2$ yields
\begin{align}
\label{eq: int chi}
\chi^2(P \| Q) &= 2 \int_1^{\infty} \left( \, E_{\gamma}(P\|Q) +
\gamma^{-3} E_{\gamma}(Q\|P) \right) \, \text{d}\gamma.
\end{align}
Accordingly, bounds on the $E_\gamma$ divergence, such as those presented in
Section~\ref{subsec: Pinsker for EG}, directly translate into bounds on other
important $f$-divergences.

\begin{remark}
Proposition~\ref{prop: integral} can be derived also from the integral representation of
$f$-divergences in \cite[Corollary~3.7]{CZK98}.
\end{remark}

\subsection{Extension of Pinsker's Inequality to $E_\gamma$ Divergence}
\label{subsec: Pinsker for EG}
This subsection upper bounds $E_\gamma$ divergence in terms
of the relative entropy.

\begin{theorem}
\begin{align} \label{eq: loose_egre}
E_\gamma(P\|Q) \, {\log \gamma} &\leq D(P\|Q) + 2e^{-1} \log e, \\
 \label{eq: Pinsker3}
2 \, E_\gamma^2(P\|Q) \, \log e &\leq D(P\|Q).
\end{align}
\end{theorem}

\begin{proof}
The bound in \eqref{eq: loose_egre} appears in \cite[Proposition~13]{LiuCV1_IT15}.
For $\gamma = 1$, \eqref{eq: Pinsker3} reduces to \eqref{eq: Pinsker}.
Since $E_\gamma$ is monotonically decreasing in $\gamma$, \eqref{eq: Pinsker3} also
holds for $\gamma > 1$.
\end{proof}
\par
For $\gamma = 1$, \eqref{eq: Pinsker3} becomes Pinsker’s inequality \eqref{eq: Pinsker},
for which there is no tighter constant. Moreover, in view of \eqref{eq:EG-TV}, for small
$E_1 (P\|Q)$, the minimum achievable $D(P\|Q)$ is indeed quadratic in $E_1 (P\|Q)$
\cite{FedotovHT_IT03}. This ceases to be the case for $\gamma > 1$, in which case it
is possible to upper bound $E_\gamma (P\|Q)$ as a constant times $D(P\|Q)$.

\begin{theorem}  \label{thm:EG vs. RE}
For every $\gamma > 1$,
\begin{align}  \label{eq:sup-EG and RE}
\sup \frac{E_{\gamma}(P\|Q)}{D(P\|Q)} = c_{\gamma}
\end{align}
where the supremum is over $P \ll Q,  P \neq Q$, and $c_{\gamma}$ is a universal
function (independent of $P$ and $Q$), given by
\begin{align}
\label{eq: c_gamma}
& c_{\gamma} = \frac{t_\gamma-\gamma}{r(t_\gamma)},
\\[0.1cm]
\label{eq: t_gamma}
& t_\gamma = - \gamma \, W_{-1}\left(-\tfrac1{\gamma} \, e^{-\frac1{\gamma}} \right)
\end{align}
with $r$ in \eqref{eq: c_gamma} is given in \eqref{eq: r}, and $W_{-1}$ in \eqref{eq: t_gamma}
denotes the secondary real branch of the Lambert $W$ function \cite{Corless96}.
\end{theorem}

\begin{proof}
The functions $f_\gamma(t) = (t-\gamma)^+$  and $r$ (see \eqref{eq: r}) satisfy
the sufficient conditions of Theorem~\ref{theorem: tight bound}. Their ratio is
\begin{align} \label{eq: kappa-EG,RE}
\kappa_\gamma(t) &=
\left\{
\begin{array}{ll}
\frac{t-\gamma}{r(t)} &  t \in [\gamma, \infty)\\
0 & t \in (0, \gamma].
\end{array}
\right.
\end{align}
For $t > \gamma$
\begin{align} \label{eq: diff kappa-EG,RE}
\kappa_\gamma'(t) = \frac{\gamma \log t + (1-t) \log e}{r^2(t)}.
\end{align}
Since $\gamma > 1$, it follows from \eqref{eq: diff kappa-EG,RE} that there exists
$t_\gamma \in (\gamma, \infty)$ such that $\kappa_\gamma$ is monotonically increasing
on $[\gamma, t_\gamma]$, and it is monotonically decreasing on $[t_\gamma, \infty)$.
The value $t_\gamma$ is the unique solution of the equation $\kappa_\gamma'(t)=0$ in
$(\gamma, \infty)$. From \eqref{eq: diff kappa-EG,RE}, $t_\gamma \in (\gamma, \infty)$
solves the equation
\begin{align} \label{equation1 for t_gamma}
\gamma \log t = (t-1) \, \log e
\end{align}
which, after exponentiating both sides of \eqref{equation1 for t_gamma} and making the
substitution $x = -\frac{t}{\gamma}$, gives
\begin{align}  \label{equation2 for t_gamma}
x e^x = -\tfrac1{\gamma} \, e^{-\frac1{\gamma}}.
\end{align}
The trivial solution of \eqref{equation2 for t_gamma} $x = -\frac1{\gamma}$  corresponds to
$t=1$, which is an improper solution of \eqref{equation1 for t_gamma} since $t < \gamma$.
The proper solution of \eqref{equation2 for t_gamma} is its second real solution given by
\begin{align} \label{eq: solution}
x = W_{-1}\left(-\tfrac1{\gamma} \, e^{-\frac1{\gamma}}\right);
\end{align}
consequently, $t = -\gamma x$ and \eqref{eq: solution} give \eqref{eq: t_gamma}.
In conclusion, for $t \geq 0$ and $\gamma > 1$,
\begin{align} \label{eq: bounds on kappa}
0 \leq \kappa_\gamma(t) \leq \kappa_\gamma(t_\gamma) = c_{\gamma}
\end{align}
where the equality in \eqref{eq: bounds on kappa} follows from \eqref{eq: c_gamma},
\eqref{eq: kappa-EG,RE}, and $t_\gamma > \gamma$.
Theorem~\ref{theorem: tight bound}\ref{theorem: tight bound: parta})  yields
\begin{align} \label{eq:EG and RE}
E_{\gamma}(P\|Q) \leq c_{\gamma} D(P\|Q).
\end{align}
To show \eqref{eq:sup-EG and RE}, or in other words that there is no better constant in \eqref{eq:EG and RE}
than $c_\gamma$, it is enough to restrict to binary alphabets:
Let $\set{A} = \{0, 1\}$, $\varepsilon \in (0,1)$, and $P_\varepsilon(0) = \varepsilon$,
$Q_\varepsilon(0) = \frac{\varepsilon}{t_\gamma}$.
Since $t_\gamma > 1$, we have
\begin{align} \label{less than gamma}
\frac{P_\varepsilon(1)}{Q_\varepsilon(1)} = \frac{1-\varepsilon}{1-\frac{\varepsilon}{t_\gamma}} < 1 < \gamma
\end{align}
and
\begin{align} \label{eq: kappa is zero}
\kappa_\gamma\left(\frac{P_\varepsilon(1)}{Q_\varepsilon(1)}\right) = 0
\end{align}
where \eqref{eq: kappa is zero} follows from \eqref{eq: kappa-EG,RE} and \eqref{less than gamma}.
We show in the following that $\frac{E_\gamma(P_\varepsilon\|Q_\varepsilon)}{D(P_\varepsilon\|Q_\varepsilon)}$
can come arbitrarily close (from below) to $c_\gamma$ by choosing a sufficiently small $\varepsilon > 0$.
To that end, for all $\varepsilon \in (0,1)$,
\begin{align}
\label{EG1} E_\gamma(P_\varepsilon\|Q_\varepsilon)
& = Q_\varepsilon(0) \; f_\gamma\left(\frac{P_\varepsilon(0)}{Q_\varepsilon(0)}\right)
+ Q_\varepsilon(1) \; f_\gamma\left(\frac{P_\varepsilon(1)}{Q_\varepsilon(1)}\right) \\[0.1cm]
\label{EG2}
& = Q_\varepsilon(0) \; r\left(\frac{P_\varepsilon(0)}{Q_\varepsilon(0)}\right) \;
\kappa_\gamma\left(\frac{P_\varepsilon(0)}{Q_\varepsilon(0)}\right)
+ Q_\varepsilon(1) \; r\left(\frac{P_\varepsilon(1)}{Q_\varepsilon(1)}\right) \;
\kappa_\gamma\left(\frac{P_\varepsilon(1)}{Q_\varepsilon(1)}\right) \\[0.1cm]
\label{EG3}
& = Q_\varepsilon(0) \; r\left(\frac{P_\varepsilon(0)}{Q_\varepsilon(0)}\right) \;
\kappa_\gamma\left(\frac{P_\varepsilon(0)}{Q_\varepsilon(0)}\right)\\[0.1cm]
\label{EG4}
& = c_\gamma \; Q_\varepsilon(0) \; r\left(\frac{P_\varepsilon(0)}{Q_\varepsilon(0)}\right) \\[0.1cm]
\label{EG5}
& = c_\gamma \left[ D(P_\varepsilon\|Q_\varepsilon)
- Q_\varepsilon(1) \; r\left(\frac{P_\varepsilon(1)}{Q_\varepsilon(1)}\right) \right]
\end{align}
where \eqref{EG1} holds due to \eqref{eq:Eg f-div};
\eqref{EG2} follows
from the definition of $\kappa_\gamma$ as the continuous extension of $\frac{f_\gamma}{r}$ with
$r$ in \eqref{eq: r};
\eqref{EG3} holds due to \eqref{eq: kappa is zero};
\eqref{EG4} follows from \eqref{eq: c_gamma}, \eqref{eq: kappa-EG,RE};
\eqref{EG5} follows from \eqref{eq: r-divergence} which implies that
$D_r(P_{\varepsilon} \| Q_{\varepsilon}) = D(P_{\varepsilon} \| Q_{\varepsilon})$.
From \eqref{EG5},
\begin{align} \label{eq1: EG/RE}
c_\gamma \left[1 - \frac1{D(P_\varepsilon\|Q_\varepsilon)}
\cdot r\left(\frac{P_\varepsilon(1)}{Q_\varepsilon(1)}\right) \right]
\leq \frac{E_\gamma(P_\varepsilon\|Q_\varepsilon)}{D(P_\varepsilon\|Q_\varepsilon)} < c_\gamma.
\end{align}
Appendix~\ref{appendix: limit} shows that
\begin{align} \label{eq: limit}
\lim_{\varepsilon \to 0} \left\{\frac1{D(P_\varepsilon\|Q_\varepsilon)} \cdot
r\left(\frac{P_\varepsilon(1)}{Q_\varepsilon(1)}\right) \right\} = 0,
\end{align}
which implies from \eqref{eq1: EG/RE} that
\begin{align} \label{eq2: limit of EG/RE}
\lim_{\varepsilon \to 0} \frac{E_\gamma(P_\varepsilon\|Q_\varepsilon)}{D(P_\varepsilon\|Q_\varepsilon)} = c_\gamma.
\end{align}
\end{proof}

\begin{remark} \label{remark: approx c_gamma}
The value of $c_\gamma$ given in \eqref{eq: c_gamma}
can be approximated by
\begin{align} \label{eq1: approx c_gamma}
& c_\gamma \approx \frac{\delta}{(\delta+\gamma) \log\left(\frac{\delta+\gamma}{e}\right)+\log e} \\
\label{eq2: approx c_gamma}
& \delta = \frac{\alpha \gamma \log \gamma}{\log e}, \quad \alpha = 1.1791
\end{align}
with a relative error of less than $1 \%$ for all $\gamma > 1$, and no more
than $10^{-3}$ for $\gamma \geq 2$.
\end{remark}

\begin{figure}[h]
\centerline{\includegraphics[width=9cm]{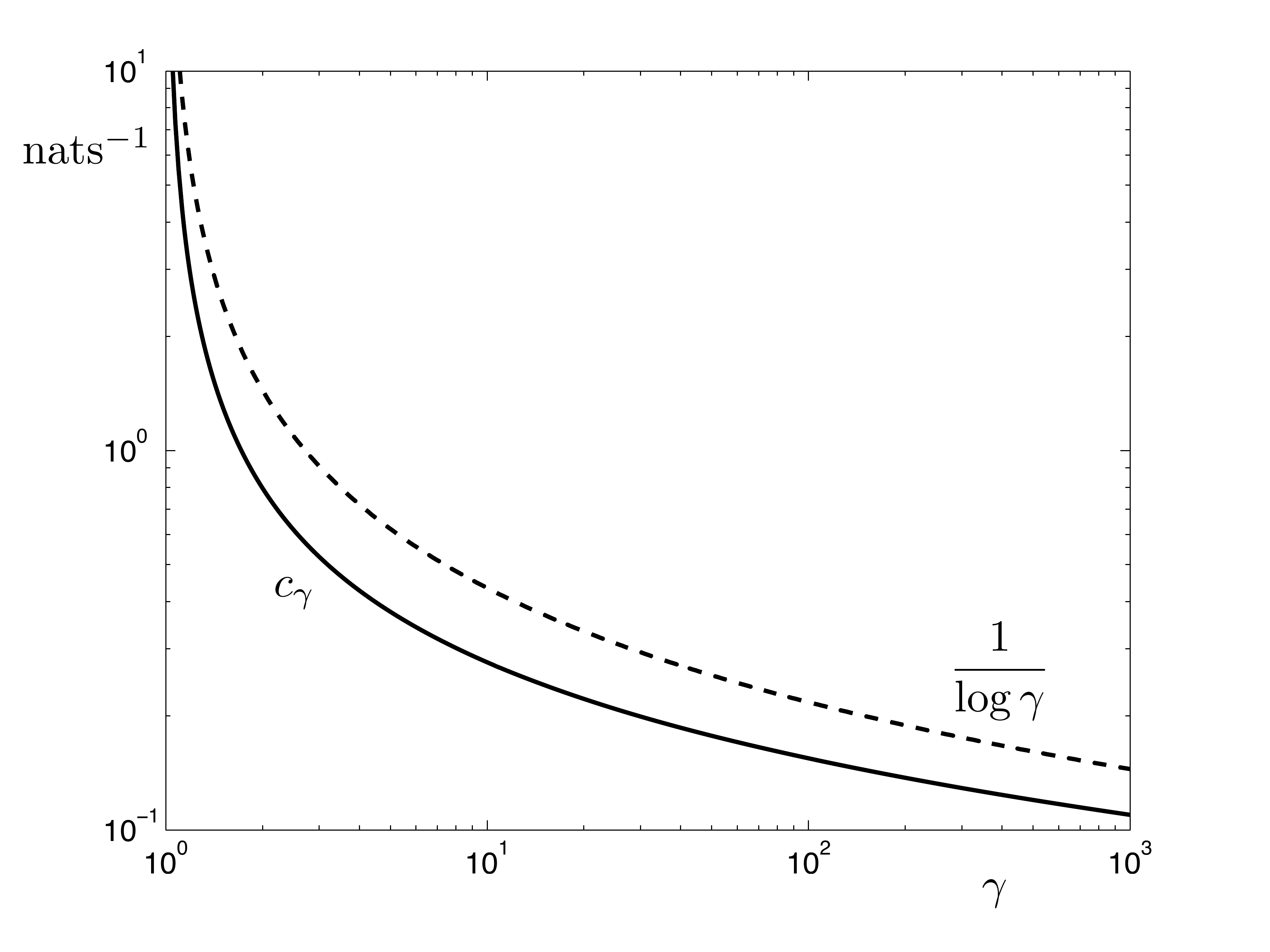}}
\caption{\label{figure:bound on EG}
The coefficient $c_{\gamma}$ in \eqref{eq: c_gamma} (solid line) compared to
$\frac1{\log \gamma}$ (cf. \eqref{eq: loose_egre}) (dashed line).}
\end{figure}
It can be verified that the bound in
Theorem~\ref{thm:EG vs. RE} is tighter than \eqref{eq: loose_egre} since
$c_\gamma < \frac1{\log \gamma}$ for $\gamma > 1$, and the additional
positive summand $\frac{2 \log e}{e \log \gamma}$ in the right side of
\eqref{eq: loose_egre} further loosens the bound \eqref{eq: loose_egre} in
comparison to \eqref{eq:sup-EG and RE}. According to the approximation of $c_\gamma$
in \eqref{eq1: approx c_gamma} and \eqref{eq2: approx c_gamma}, we have for large
values of $\gamma$
\begin{align} \label{eq: asym_approx_c}
c_\gamma \approx \frac1{\log\left(\frac{\alpha \gamma \log \gamma}{e \log e}\right)}.
\end{align}

\begin{remark}
The impossibility of a general lower bound on $E_\gamma(P\|Q)$, for $\gamma > 1$,
in terms of the relative entropy $D(P\|Q)$ is evident from Example \ref{example:zeroEg}.
\end{remark}

\begin{remark}
The fact that $\{c_{\gamma}\}_{\gamma \geq 1}$ in \eqref{eq: c_gamma} is monotonically decreasing
in $\gamma$ (see Figure~\ref{figure:bound on EG}) is consistent with \eqref{eq:sup-EG and RE} and
the fact that the $E_{\gamma}$ divergence is monotonically decreasing in $\gamma$.
\end{remark}

\begin{remark}
The fact that  the behavior of $D(P\|Q)$ for small $|P-Q|$ is quadratic
rather than linear does not contradict Theorem~\ref{thm:EG vs. RE} because
$\lim_{\gamma \downarrow 1} c_{\gamma} = +\infty$ (see
Figure~\ref{figure:bound on EG}).
\end{remark}

In view of  \eqref{eq2: EG} and
\eqref{eq:sup-EG and RE} we obtain
\par
\begin{corollary} \label{corollary: egtore}
If $P \ll Q$, $\gamma > 1$ and $\set{F} \in \mathscr{F}$, then
\begin{align} \label{eq: egtore1}
P(\set{F}) \leq \gamma \, Q(\set{F}) + c_\gamma \, D(P\|Q).
\end{align}
\end{corollary}

\begin{corollary} \label{corollary3: EG}
If $P \ll Q$ and  $\gamma > 1$, then
\begin{align} \label{eq: EG, TV, RE}
E_{\gamma}(P\|Q)
\leq \min_{\lambda \in [0,1]} \left\{\tfrac{\lambda}{2} \, |P-Q| +
c_\gamma \, D\bigl((1-\lambda)P + \lambda Q \, \| \, Q \bigr) \right\}.
\end{align}
\end{corollary}

\begin{proof}
For $\lambda \in [0,1]$, let $R = (1-\lambda) P + \lambda Q$.
Then, we have for $\gamma > 1$,
\begin{align}
\label{eq: ub001}
E_{\gamma}(P\|Q) & \leq \tfrac12 \, |P-R| + E_{\gamma}(R\|Q) \\[0.1cm]
\label{eq: ub002}
& =  \tfrac{\lambda}{2} \, |P-Q| + E_{\gamma}\bigl((1-\lambda)P + \lambda Q \, \|Q \bigr) \\[0.1cm]
\label{eq: ub003}
& \leq\tfrac{\lambda}{2} \, |P-Q| + c_\gamma \, D\bigl((1-\lambda)P + \lambda Q \, \| \, Q \bigr)
\end{align}
where \eqref{eq: ub001} is \eqref{eq2: triangleEG};
and \eqref{eq: ub003} follows from \eqref{eq:sup-EG and RE}.
\end{proof}

\begin{remark}
Note that the upper bounds in the right sides of \eqref{eq10b: EG} and \eqref{eq:EG and RE}
follow from \eqref{eq: EG, TV, RE} by setting $\lambda = 1$ or $\lambda = 0$, respectively.
\end{remark}

\begin{remark}
Further upper bounding the right side of \eqref{eq: ub003} by invoking Pinsker's
inequality and the convexity of relative entropy, followed by an optimization over
the free parameter $\lambda \in [0,1]$, does not lead to an improvement
beyond the minimum of the bounds in \eqref{eq: Pinsker3} and \eqref{eq:EG and RE}.
\end{remark}

\subsection{Lower Bound on $\mathds{F}_{P\|Q}$ as a Function of $D(P\|Q)$}
\label{subsec: re-ris}
The $E_\gamma$ divergence proves to be instrumental in the proof of the
following bound on the complementary relative information spectrum
for positive arguments.
\begin{theorem} \label{thm: reris}
If $P \ll Q$, $P \neq Q$, and $\beta > 1$, then
\begin{align}
%\label{eq: reris}
\frac{1-\mathds{F}_{P\|Q}(\log \beta)}{D(P\|Q)} &\leq u(\beta)
\label{eq:u}
\triangleq \min_{\gamma \in (1, \beta)}
\left( \frac{\beta \, c_\gamma}{\beta-\gamma} \right),
\end{align}
where $c_\gamma$ is given in \eqref{eq: c_gamma}. Furthermore,
the function $u \colon (1, \infty) \to \Reals^+$  is  monotonically
decreasing with
\begin{align} \label{eq: u-bound}
u(\beta) \leq \frac{2}{\log \left(\frac{\beta}{2e}\right)}, \quad \forall \, \beta > 2e.
\end{align}
\end{theorem}

\begin{proof}
For $\beta>1$, denote the event
\begin{align} \label{eq: set}
\set{F}_\beta \triangleq \left\{x \in \set{A} \colon
\imath_{P\|Q}(x) > \log \beta \right\}
\end{align}
which satisfies
\begin{align} \label{eq: PQ_ineq}
P(\set{F}_\beta) > \beta \, Q(\set{F}_\beta).
\end{align}
Then,
\begin{align}
\label{eq: reris1}
1-\mathds{F}_{P\|Q}(\log \beta)
&= P(\set{F}_\beta) \\
\label{eq: reris3}
& \leq \inf_{\gamma \in (1, \beta)}\frac{P(\set{F}_\beta)
- \gamma \, Q(\set{F}_\beta)}{1-\frac{\gamma}{\beta}} \\
\label{eq: reris4}
& \leq \inf_{\gamma \in (1, \beta)}
\frac{\beta \, E_\gamma(P\|Q)}{\beta-\gamma} \\
\label{eq: reris5}
& \leq \inf_{\gamma \in (1, \beta)} \left(\frac{\beta \,
c_\gamma}{\beta-\gamma}\right) \, D(P\|Q)
\end{align}
where \eqref{eq: reris1} holds by Definition~\ref{def:RIS};
\eqref{eq: reris3} follows from \eqref{eq: PQ_ineq};
\eqref{eq: reris4} is satisfied by \eqref{eq2: EG}, and
\eqref{eq: reris5} is due to \eqref{eq:sup-EG and RE}.
Note that the infimum in \eqref{eq: reris5}
is attained because $c_\gamma$ is continuous and for $\beta>1$,
$\frac{\beta \, c_\gamma}{\beta-\gamma}$ tends to $+ \infty$
at both extremes of the interval $(1, \beta)$.
The monotonicity of $u(\beta)$ and the bound in
\eqref{eq: u-bound} are proved in Appendix~\ref{appendix:u}.
\end{proof}

%%%%%%%%%%%%%%%%%%%%%%%%%%%%%%%%%%%%%%%%%%%%%%%
\section{R\'{e}nyi Divergence}
\label{sec:RD}

The R\'{e}nyi divergence (Definition~\ref{def:RD}) admits a variational representation
in terms of the relative entropy \cite[Theorem~1]{Shayevitz_ISIT11}. Let $\PU \ll \PZ$
then, for $\alpha > 0$,
\begin{align}
& (1-\alpha) \, D_{\alpha}(\PU \| \PZ)
= \min_{P \ll \PU} \bigl\{\alpha \, D(P \| \PU) + (1-\alpha) \, D(P \| \PZ) \bigr\}.
\label{eq: variational representation RD}
\end{align}

In this section, integral expressions for the R\'{e}nyi divergence
are derived in terms of the relative information spectrum (Definition~\ref{def:RIS}).
These expressions are used to obtain bounds on the R\'{e}nyi divergence as a function of
the variational distance under the assumption of bounded relative information.

\subsection{Expressions in Terms of the Relative Information Spectrum} \label{subsec: RD-RIS}
To state the results in this section, it is convenient to introduce
$\zeta_\alpha \colon (0,\infty) \to [0, \infty)$
\begin{align} \label{eq: zeta}
\zeta_{\alpha}(\beta) = \beta^{\alpha-2} \; \left( 1-\mathds{F}_{P\|Q}\bigl(\log \beta \bigr) \right).
\end{align}
The R\'{e}nyi divergence admits the following representation in terms of the relative information
spectrum and the relative information bounds $(\beta_1, \beta_2) \in [0,1]^2$ in \eqref{eq: beta1}--\eqref{eq: beta2}.
\begin{theorem}\label{thm:RDRIS}
Let $P \ll Q$.
\begin{itemize}
\item
If $\beta_1 > 0$ and $\alpha \in (0,1)\cup (1, \infty)$, then
\begin{align}
D_{\alpha}(P\|Q)
& =  \frac1{\alpha-1} \; \log \biggl( \beta_1^{1 - \alpha} + (1 - \alpha) \int_{\beta_2}^{\beta_1^{-1}}
\bigl(\beta^{\alpha-2} - \zeta_{\alpha}(\beta)\bigr) \text{d}\beta \biggr).
\label{dalpha--ub}
\end{align}
\item
If $\beta_1 = 0$ and $\alpha \in (0,1)$, then
\begin{align}
D_{\alpha}(P\|Q) & = \frac1{\alpha-1} \; \log \left( (1-\alpha) \dint_{\beta_2}^{\infty}
\bigl(\beta^{\alpha-2} - \zeta_{\alpha}(\beta)\bigr) \; \text{d}\beta \right).
\label{dalpha--ud}
\end{align}
\item
If  $\alpha \in (1,\infty)$, then
\begin{align}
D_{\alpha}(P\|Q) & = \frac1{\alpha-1} \; \log \left( (\alpha-1) \dint_0^{\infty}
\zeta_{\alpha}(\beta) \; \text{d}\beta \right)
\label{ssd}\\
& = \frac1{\alpha-1} \; \log \left(  \beta_2^{\alpha-1}+ (\alpha-1 ) \dint_{\beta_2}^{\infty}
\zeta_{\alpha}(\beta) \; \text{d}\beta \right).
\label{dalpha--uf}
\end{align}
\end{itemize}
\end{theorem}

\begin{proof}
If $\alpha > 1$,  \eqref{eq:RD3} implies that $D_{\alpha}(P\|Q)$ is given by
\begin{align}
\frac1{\alpha-1} \; \log \Bigl( \mathbb{E} \bigl[ \exp \bigl((\alpha-1) \,
\imath_{P\|Q}(X)\bigr)\bigr] \Bigr)
& = \frac1{\alpha-1} \; \log \left( \int_0^{\infty} \hspace*{-0.2cm}
\prob\left[\exp\bigl((\alpha-1) \, \imath_{P\|Q}(X)\bigr) > t \right] \text{d}t \right)
\label{ssa2}\\[0.1cm]
& = \frac1{\alpha-1} \; \log \left( \int_0^{\infty} \prob \left[\imath_{P\|Q}(X) >
\frac{\log t}{\alpha-1} \right] \, \text{d}t \right)
\label{ssb}
\end{align}
where \eqref{ssa2} follows from \eqref{eq: expectation}
for an arbitrary non-negative random variable $V$, and we use
$\alpha > 1$ to write \eqref{ssb}. Then, \eqref{ssd} holds by
the definition of the relative information spectrum in
\eqref{eq:RIS} and by changing the integration variable
$t = \beta^{\alpha-1}$. If $\beta_1 >0$, the integrand in the
right side of \eqref{ssd} is zero in $[\beta_1^{-1},\infty)$
and the expression in \eqref{dalpha--ub} is readily verified
(for $\alpha >1$). More generally (without requiring $\beta_1 >0$),
we split the integral in the right side of \eqref{ssd} into
$[0, \beta_2) \cup [\beta_2, \infty)$, and \eqref{dalpha--uf}
follows since the integral over the leftmost interval is
$\beta_2^{\alpha -1 }$ considering that
$\mathds{F}_{P\|Q}(\log \beta) = 0$ therein.

If $\alpha \in (0, 1)$, we write $D_{\alpha}(P\|Q) $ as
\begin{align}
\frac1{\alpha-1} \; \log \Bigl( \mathbb{E} \bigl[ \exp \bigl((\alpha-1) \,
\imath_{P\|Q}(X)\bigr)\bigr] \Bigr)
& = \frac1{\alpha-1} \; \log \left( \int_0^{\infty} \hspace*{-0.2cm} \prob\left[ \exp\bigl((\alpha-1)
\, \imath_{P\|Q}(X)\bigr) \geq t \right] \text{d}t \right) \\[0.1cm]
& = \frac1{\alpha-1} \; \log \left( \int_0^{\infty} \prob\left[ \imath_{P\|Q}(X) \leq
\frac{\log t}{\alpha-1} \right] \, \text{d}t \right) \\[0.1cm]
& = \frac1{\alpha-1} \; \log \left( \int_{\infty}^0 \prob\left[ \imath_{P\|Q}(X) \leq
\log \beta \right] \, (\alpha-1) \, \beta^{\alpha-2} \, \text{d}\beta \right) \nonumber \\[0.1cm]
%& = \frac1{\alpha-1} \; \log \left( (1-\alpha) \int_0^{\infty} \beta^{\alpha-2} \;
%\mathds{F}_{P\|Q}\bigl(\log \beta\bigr) \; \text{d}\beta \right) \\
& = \frac1{\alpha-1} \; \log \left( (1-\alpha) \int_{\beta_2}^{\infty} \beta^{\alpha-2} \;
\mathds{F}_{P\|Q}\bigl(\log \beta\bigr) \; \text{d}\beta \right)
\label{manon}
\end{align}
which is the expression in \eqref{dalpha--ud}. If $\beta_1 >0$, then we can further
split the integral in the right side of \eqref{manon} into the intervals
$[\beta_2, \beta_1^{-1}) \cup [\beta_1^{-1}, \infty)$. Over the rightmost interval,
$\mathds{F}_{P\|Q}\bigl(\log \beta\bigr)  = 1$ and the integral is seen to be
$\beta_1^{1 - \alpha}$, thereby verifying \eqref{dalpha--ub} for $\alpha \in (0,1)$.
\end{proof}

The close relationship between the R\'enyi and Hellinger divergences in \eqref{renyimeetshellinger}
results is the following integral representations for the Hellinger divergence.
\begin{corollary} \label{cor: Hellinger-RIS}
Let $P \ll Q$.
\begin{itemize}
\item
If $\beta_1 > 0$ and $\alpha \in (0,1)\cup (1, \infty)$, then
\begin{align}  \label{eq1: Hellinger-RIS}
\mathscr{H}_{\alpha}(P\|Q) &= \frac{\beta_1^{1-\alpha}-1}{\alpha-1}
- \int_{\beta_2}^{\beta_1^{-1}} \bigl(\beta^{\alpha-2} - \zeta_{\alpha}(\beta)\bigr) \text{d}\beta.
\end{align}
\item
If $\beta_1 = 0$ and $\alpha \in (0,1)$, then
\begin{align}  \label{eq2: Hellinger-RIS}
\mathscr{H}_{\alpha}(P\|Q) &= \frac1{1-\alpha} - \int_{\beta_2}^{\infty}
\bigl(\beta^{\alpha-2} - \zeta_{\alpha}(\beta)\bigr) \text{d}\beta.
\end{align}
\item
If $\alpha \in (1,\infty)$, then
\begin{align}  \label{eq3: Hellinger-RIS}
\mathscr{H}_{\alpha}(P\|Q)
&= \int_0^{\infty} \zeta_\alpha (\beta)
\; \text{d}\beta - \frac1{\alpha-1} \\[0.1cm]
\label{eq4: Hellinger-RIS}
&= \frac{\beta_2^{\alpha-1}-1}{\alpha-1} + \dint_{\beta_2}^{\infty} \zeta_\alpha (\beta)\; \text{d}\beta.
\end{align}
\end{itemize}
\end{corollary}

\begin{proof}
Combining \eqref{renyimeetshellinger} with \eqref{dalpha--ub}, \eqref{dalpha--ud},
\eqref{ssd}, \eqref{dalpha--uf} yields \eqref{eq1: Hellinger-RIS}--\eqref{eq4: Hellinger-RIS}, respectively.
\end{proof}

Particularizing \eqref{eq: zeta}, \eqref{ssd} and \eqref{eq3: Hellinger-RIS} to $\alpha=2$, we obtain
\begin{align}
\label{eq: RD2}
D_2(P\|Q) &= \log \left( \int_0^{\infty} \left( 1-\mathds{F}_{P\|Q}(\log \beta) \right)
\; \text{d}\beta \right),
\end{align}
\begin{align}
\label{eq: chi RES}
\chi^2(P \| Q) &= \int_0^{\infty} \left( 1 - \mathds{F}_{P\|Q}(\log \beta) \right)
\, \text{d}\beta - 1 \\
\label{eq2: chi RES}
& = \int_1^{\infty} \left( 1 - \mathds{F}_{P\|Q}(\log \beta) \right)
\, \text{d}\beta - \int_0^1 \mathds{F}_{P\|Q}(\log \beta) \, \text{d}\beta.
\end{align}
Note the resemblance of the integral expressions in \eqref{eq2: RE} and \eqref{eq2: chi RES}
for $D(P\|Q)$ and $\chi^2(P\|Q)$, respectively.

We conclude this subsection by proving three properties of
the Hellinger divergence as a function of its order. The first
two monotonicity properties are analogous to
\cite[Theorems~3 and~16]{ErvenH14} for the R\'{e}nyi divergence;
these monotonicity properties have been originally stated in
\cite[Proposition~2.7]{LieseV_book87}, though the following
alternative proof is more transparent.
\begin{theorem} \label{thm: monotonicity of Hel div in alpha}
The Hellinger divergence satisfies the following properties:
\begin{enumerate}[a)]
\item
$\mathscr{H}_{\alpha}(P \| Q)$ is monotonically
increasing in $\alpha \in (0, \infty)$;
\item
$\left(\frac1\alpha-1\right) \mathscr{H}_{\alpha}(P \| Q)$
is monotonically decreasing in $\alpha \in (0, 1)$;
\item $\frac1{\alpha} \, \mathscr{H}_{\alpha}(P \| Q)$ is log-convex in $\alpha \in (0, \infty)$,
which implies that for every $\alpha, \beta > 0$
\begin{align} \label{eq: log-convexity}
\mathscr{H}^2_{\frac{\alpha+\beta}{2}}(P \| Q) \leq \frac{(\alpha+\beta)^2}{4 \alpha \beta}
\; \mathscr{H}_{\alpha}(P \| Q) \, \mathscr{H}_{\beta}(P \| Q).
\end{align}
\end{enumerate}
\end{theorem}
\begin{proof}
\begin{enumerate}[a)]
\item From \eqref{eq: fD4} and \eqref{eq: Hel-divergence}, we have
\begin{align} \label{eq2: Hel-divergence}
\mathscr{H}_{\alpha}(P \| Q) = D_{f_\alpha}(P \| Q)
\end{align}
with
\begin{align} \label{eq2: H as fD}
f_\alpha(t) = \frac{t^\alpha-\alpha(t-1)-1}{\alpha-1}, \quad t > 0,
\end{align}
whose derivative is
\begin{align} \label{eq: partial derivative of f}
\frac{\partial}{\partial \alpha} \, f_\alpha(t) = \frac{t \, r(t^{\alpha-1})}{(\alpha-1)^2 \, \log e} > 0
\end{align}
where the function $r \colon (0, \infty) \to \Reals$ is defined in \eqref{eq: r}.
Since it is strictly positive except at $t=1$,  $f_\alpha$ is monotonically
increasing in $\alpha \in (0, \infty)$. Hence, Part~a) follows from \eqref{eq2: Hel-divergence}.
\item
From \eqref{eq2: Hel-divergence}, for $\alpha \in (0,1)$, we have
\begin{align} \label{eq: scaled Hel}
\left(\frac1\alpha-1\right) \, \mathscr{H}_{\alpha}(P \| Q) = D_{g_\alpha}(P \| Q)
\end{align}
where $g_\alpha \colon (0, \infty) \to \Reals$ is the convex function
\begin{align}
g_\alpha(t) &= t-1 - \frac{t^\alpha-1}{\alpha}, \quad t>0.
\end{align}
with derivative
\begin{align} \label{eq: partial derivative of g}
\frac{\partial}{\partial \alpha} \, g_\alpha(t) = -\frac{r(t^\alpha)}{\alpha^2 \, \log e} < 0,
\end{align}
so $g_\alpha$ is monotonically decreasing
in $\alpha \in (0,1)$. Hence, Part~b) follows from \eqref{eq: scaled Hel}.
\item To prove the log-convexity of $\frac1{\alpha} \, \mathscr{H}_{\alpha}(P \| Q)$
in $\alpha \in (0, \infty)$, we rely on \cite[Theorem~2.1]{Simic07} which states
that if $W$ is a non-negative random variable, then
$\lambda_\alpha$ in \eqref{eq: log-convex lambda} is log-convex in
$\alpha$. The claim now follows from \eqref{eq: log-convex lambda}
by setting $W = \frac{\mathrm{d}P}{\mathrm{d}Q} \, (Y)$ with $Y \sim Q$,
which yields that $\lambda_{\alpha} = \frac1{\alpha} \, \mathscr{H}_{\alpha}(P \| Q) \log e$
for $\alpha \in (0, \infty)$.
\end{enumerate}
\end{proof}

\subsection{Bounds as a Function of the Total Variation Distance} \label{subsec: RD-TV}

Just as with Pinsker's inequality,  for any $\varepsilon \in (0,2]$, the minimum
value of $D_{\alpha}(P\|Q)$ compatible with $|P - Q | \geq \varepsilon$, is achieved with
distributions on a binary alphabet \cite[Proposition~1]{Sason_IT2016}:
\begin{align} \label{eq: min RD s.t. TV}
\min_{P,Q \colon |P-Q| \geq \varepsilon} D_\alpha(P\|Q)
= \min_{p, q \colon |p-q| \geq \frac{\varepsilon}{2}} d_{\alpha}(p\|q)
\end{align}
where  the binary order-$\alpha$  R\'enyi divergence is defined as
\begin{align}
\label{eq: binary RD}
d_{\alpha}(p \| q) \triangleq
\left\{ \begin{array}{ll}
\frac1{\alpha-1}
\log \Bigl(p^\alpha q^{1-\alpha}+(1-p)^\alpha (1-q)^{1-\alpha} \Bigr),
& \mbox{if $\alpha \neq 1$}  \\[0.2cm]
p \, \log\frac{p}{q} + (1-p) \, \log\frac{1-p}{1-q},
& \mbox{if $\alpha = 1$.}
\end{array}
\right.
\end{align}

We proceed to use Theorem~\ref{thm:RDRIS} to get an upper bound
on $D_{\alpha}(P\|Q)$ expressed in terms of $|P-Q|$.
\begin{theorem} \label{thm: ub-Rd}
If $\beta_1 \in (0,1)$ and $\alpha \in (0,1) \cup (1, \infty)$, then
\begin{align}  \label{eq: ub-Rd}
D_{\alpha}(P\|Q) \leq \frac1{\alpha-1} \; \log\left(1 + \frac{|P-Q|}{2} \,
\frac{\beta_1^{1-\alpha}-1 }{1-\beta_1} \right).
\end{align}
\end{theorem}

\begin{proof}
Regardless of whether $\alpha <1$ or $\alpha >1$, we can only get an upper bound
if, in view of \eqref{eq: zeta}, in the integral in \eqref{dalpha--ub} we drop
the interval $[\beta_2, 1]$:
\begin{align}
D_{\alpha}(P\|Q)
& \leq \frac1{\alpha-1} \; \log \left( \beta_1^{1 - \alpha} + (1 - \alpha) \int_1^{\beta_1^{-1}}
\beta^{\alpha-2} \; \mathds{F}_{P\|Q}(\log \beta) \; \text{d}\beta \right) \nonumber \\
& = \frac1{\alpha-1} \; \log \left( 1 - (1 - \alpha) \int_1^{\beta_1^{-1}}
\beta^{\alpha-2} \; \bigl(1 - \mathds{F}_{P\|Q}(\log \beta) \bigr) \; \text{d}\beta \right) \nonumber \\
\label{rd4}
& = \frac1{\alpha-1} \; \log \Bigl(1 + (\alpha-1) \delta \, \mathbb{E}[W^\alpha] \Bigr)
\end{align}
where  \eqref{rd4} holds with $\delta = \tfrac12 |P-Q|$ and  $W \sim p_1$ where $p_1$ is the probability
density function supported on $[1, \beta_1^{-1}]$:
\begin{align} \label{eq: pdf p1}
p_1(\beta) =
\frac1{\delta \beta^2} \left(1 - \mathds{F}_{P\|Q}(\log \beta)\right).
\end{align}
Note that $p_1$ is indeed a probability density function due to \eqref{eq: TV100}.
In order to proceed, we derive an upper bound on $\mathbb{E}[W^\alpha]$
expressed in terms of $|P-Q|$ by invoking
Lemma~\ref{lemma: 2 pdfs} in Appendix~\ref{appendix: two pdfs}. To that end, denote
the monotonically increasing and non-negative  function  $g(x) = x^\alpha$ for
$x \geq 0$, and let $p_2$ be the probability density function supported
on $[1, \beta_1^{-1}]$:
\begin{align} \label{eq: pdf p2}
p_2(\beta) = \frac1{1-\beta_1} \, \frac1{\beta^2}.
\end{align}
Note that, on their support,
$\beta^2 p_1( \beta) $ is monotonically decreasing while $\beta^2 p_2( \beta) $ is constant.
Therefore, we can apply Lemma~\ref{lemma: 2 pdfs} to $W \sim p_1$ and $V \sim p_2$ to obtain
\begin{align}
\label{expectations1}
\mathbb{E}\bigl[W^\alpha\bigr] & \leq \mathbb{E}\bigl[V^\alpha\bigr]
 = \frac{\beta_1^{1-\alpha}-1}{(1-\beta_1)(\alpha-1)}.
\end{align}
which
gives the desired result upon substituting in \eqref{rd4}.
\end{proof}

\begin{corollary} \label{cor: Hellinger-TV}
If $\beta_1 \in (0,1)$ and $\alpha \in (0,1) \cup (1, \infty)$, then
\begin{align} \label{eq: Hellinger-TV}
\mathscr{H}_{\alpha}(P\|Q) \leq \frac{\beta_1^{1-\alpha}-1}{2(\alpha-1)(1-\beta_1)} \cdot |P-Q|.
\end{align}
\end{corollary}

\begin{proof}
Combining \eqref{renyimeetshellinger} and
\eqref{eq: ub-Rd} yields  \eqref{eq: Hellinger-TV}.
\end{proof}

Particularizing \eqref{eq: Hellinger-TV} to $\alpha=2$ yields
\begin{align} \label{eq2: chi square - TV}
\chi^2(P\|Q) \leq \tfrac12 \, \beta_1^{-1} \, |P-Q|
\end{align}
which improves the bound in \eqref{eq1: chi square - TV} if either
$\beta_1 \leq \tfrac12$ or $\beta_2=0$.

The combination of \eqref{eq: Pinsker}, \eqref{eq: BretagnolleH79}
and \eqref{eq2: chi square - TV} yields the following bound:
\begin{corollary}
If  $\beta_1 > 0$, then
\begin{align} \label{eq: chi square - RE}
\chi^2(P \| Q) \leq \frac1{\beta_1} \sqrt{ \min \left\{\frac{D(P\|Q)}{2\log e},
\, 1-\exp\bigl(-D(P\|Q)\bigr)\right\} }.
\end{align}
\end{corollary}

\begin{example}
Let $P$, $Q$ be defined on $\set{A}=\{0,1\}$ with $P(0)=Q(1)=\tfrac1{100}$,
which implies that $\beta_1 = \tfrac1{99}$. Then $\chi^2(P\|Q)=97.01$,
and the bound in \eqref{eq: chi square - RE} is equal to~98.45 in contrast
to the upper bound in \eqref{eq: RE and chi-square} whose value is 121.17.
\end{example}

\begin{remark}  \label{remark: tight ub Dinf}
By letting $\alpha \to \infty$ in \eqref{eq: ub-Rd}, we obtain
$D_{\infty}(P\|Q) \leq \log \frac1{\beta_1}$,
which shows that the bound in \eqref{eq: ub-Rd} is asymptotically tight (cf.
\eqref{def:dinf}).
\end{remark}

\begin{remark} \label{remark: tight ub D1}
By letting $\alpha \to 1$ in \eqref{eq: ub-Rd}, we get \eqref{eq: SV-ITA14}.
Therefore, Theorem~\ref{thm: ub-Rd} generalizes \cite[Theorem~7]{Verdu_ITA14}.
\end{remark}

\begin{remark} \label{remark: tight ub D0}
By letting $\alpha \to 0$, it follows from \eqref{eq: ub-Rd} that
\begin{align} \label{eq: ubD0}
D_0(P\|Q) \leq \log \left(\frac1{1- \tfrac12 |P-Q|}  \right),
\end{align}
a bound which, in view of \eqref{eq: d0}, is achieved
with equality in the case of a finite alphabet with
\begin{align} \label{eq: P}
P(a) = \left\{ \begin{array}{ll}
\frac{Q(a)}{1-\delta},
&  a \in  \set{F} \\
0, &  a \in \set{F}^c
\end{array}
\right.
\end{align}
with the event $\set{F}$ selected to satisfy
$Q(\set{F}) = 1-\delta$.
\end{remark}

\begin{remark}
Another upper bound on the R\'{e}nyi divergence can be obtained by
the simpler bound
\begin{align} \label{eq: ub1 W}
\mathbb{E}[W^\alpha] \leq \beta_1^{-\alpha},
\end{align}
which holds because $W \in [1, \beta_1^{-1}]$.
Combining \eqref{rd4} and \eqref{eq: ub1 W} yields
\begin{align}  \label{eq: ref1 ub RD}
D_\alpha(P\|Q) \leq \frac1{\alpha-1} \; \log\bigl(1+(\alpha-1)\delta \beta_1^{-\alpha}\bigr).
\end{align}
Note that, in the limit $\alpha \to 0$, the bounds in \eqref{eq: ub-Rd} and \eqref{eq: ref1 ub RD}
coincide and are equal to the tight bound $-\log(1-\delta)$.
\end{remark}
\begin{remark}
Alternatively, we have the bound
\begin{align}
D_\alpha(P\|Q)
\leq \frac1{\alpha-1} \, \log\left((1-\delta)^{1-\alpha} + \delta(\alpha-1)
\int_{\frac1{1-\delta}}^{\frac1{\beta_1}} \frac{\beta^{\alpha-1}}{\beta-1} \, \text{d}\beta \right)
\label{eq: ref2 ub RD}
\end{align}
obtained from \eqref{rd4} and
\begin{align}
\mathbb{E}[W^\alpha]
& \label{eq: ub2 Wb} \leq \frac{(1-\delta)^{1-\alpha}-1}{\delta (\alpha-1)} +
\int_{\frac1{1-\delta}}^{\frac1{\beta_1}} \frac{\beta^{\alpha-1}}{\beta-1} \, \text{d}\beta.
\end{align}
which holds since, in view of \eqref{eq: pdf p1} and \eqref{eq: 2lbtv},
\begin{align}  \label{eq: ubpdf p1}
p_1(\beta) \leq \left\{ \begin{array}{ll}
\frac1{\delta \beta^2}, &  \beta \in \bigl[0, \, \frac1{1-\delta} \bigr] \\[0.2cm]
\frac1{\beta (\beta-1)},
& \beta \in \bigl[\frac1{1-\delta}, \, \beta_1^{-1} \bigr] \\[0.2cm]
0 & \mbox{otherwise.}
\end{array}
\right.
\end{align}

The upper bounds in \eqref{eq: ub-Rd} and \eqref{eq: ref2 ub RD} asymptotically coincide in the limit
where $\alpha \to \infty$, giving the common limit of $\log \left(\frac1{\beta_1}\right)$ which
is a tight upper bound (cf. Remark~\ref{remark: tight ub Dinf}).
\end{remark}

\subsection{Bounds as a Function of the Relative Entropy}
\label{subsec: RD-RE}
In this section, we provide upper and lower bounds on the R\'enyi
divergence $D_{\alpha}(P \| Q)$, for an arbitrary order $\alpha \in (0, 1)
\cup (1, \infty)$, expressed in terms of the relative entropy
$D(P \| Q)$ and $\beta_1, \beta_2$.

\begin{theorem} \label{theorem: RD-RE}
Let $(\beta_1, \beta_2) \in [0,1)^2$,
$\alpha \in (0, 1) \cup (1, \infty)$, and
$u_{\alpha} \colon [0, \infty] \to [0, \infty]$ be
\begin{align} \label{eq: u RD/RE}
u_{\alpha} = \frac{\alpha-1}{\kappa_{\alpha} (t)}
\end{align}
with $\kappa_\alpha$ defined in \eqref{eq: kappa RE/HD}.
\begin{enumerate}[a)]
\item\label{RDRE:item:a}
If $\alpha \in (0,1)$, then
\begin{align}
& \hspace*{-0.3cm} \frac{1}{\alpha-1} \, \log\Bigl(1 + u_{\alpha}(\beta_1^{-1}) \, D(P\|Q)\Bigr) \nonumber \\
\label{eq1.2: RD-RE}
& \hspace*{-0.3cm} \leq D_\alpha(P \| Q) \\
\label{eq1.3: RD-RE}
& \hspace*{-0.3cm} \leq \min \left\{D(P \| Q), \; \frac{1}{\alpha-1} \,
\log \Bigl(1 + u_{\alpha}(\beta_2)
\, D(P\|Q)\Bigr)^+ \right\}.
\end{align}
\item \label{RDRE:item:b}
If $\alpha \in (1, \infty)$, then
\begin{align}
& \max \left\{ D(P\|Q), \; \frac{1}{\alpha-1} \,
\log\Bigl(1 + u_{\alpha}(\beta_2) \, D(P\|Q)\Bigr) \right\} \nonumber \\
\label{eq2.2: RD-RE}
& \leq D_\alpha(P \| Q) \\
\label{eq2.3: RD-RE}
& \leq \min \left\{ \log \frac1{\beta_1}, \; \frac{1}{\alpha-1}
\, \log\Bigl(1 + u_{\alpha}(\beta_1^{-1}) \, D(P\|Q)\Bigr) \right\}.
\end{align}
\item \label{RDRE:item:c}
Furthermore, if $\alpha \in (0,1) \cup (1, \infty)$, then
\begin{align}  \label{eq3: RD-RE}
D_{\alpha}(P \| Q) \leq \frac1{\alpha-1} \, \log\left(1+
\frac{\overbar{\delta} \, (\beta_1^{1-\alpha}-1)}{1-\beta_1} \right)
\end{align}
where
\begin{align} \label{eq: delta bar}
\overbar{\delta}^2 = \min \left\{\frac{D(P\|Q)}{2\log e},
\, 1-\exp\bigl(-D(P\|Q)\bigr)\right\} .
\end{align}
\end{enumerate}
\end{theorem}

\begin{proof}
Parts~\ref{RDRE:item:a}) and~\ref{RDRE:item:b}) follow from
Theorem~\ref{thm: improved HausslerO} in view of \eqref{renyimeetshellinger},
from the fact that $D_{\alpha}(P \| Q)$ is monotonically increasing in $\alpha > 0$, and
from $D_{\infty}(P \| Q) = \log \frac1{\beta_1}$.
\par
Part~\ref{RDRE:item:c}) follows from Theorem~\ref{thm: ub-Rd} replacing
$\delta \triangleq \tfrac12 \, |P-Q|$  by its upper bound
$\overbar{\delta}$ obtained from \eqref{eq: Pinsker} and \eqref{eq: BretagnolleH79}.
\end{proof}

The next three remarks address the tightness of the bounds
\eqref{eq1.2: RD-RE}--\eqref{eq1.3: RD-RE} and
\eqref{eq2.2: RD-RE}--\eqref{eq2.3: RD-RE}.

\begin{remark} \label{remark1: tightness}
The constants $u_{\alpha}(\beta_1^{-1})$ and $u_{\alpha}(\beta_2)$ in
\eqref{eq1.2: RD-RE}--\eqref{eq1.3: RD-RE} and \eqref{eq2.2: RD-RE}--\eqref{eq2.3: RD-RE}
are the best possible among all probability measures $P,Q$ with given
$(\beta_1, \beta_2) \in [0,1)^2$. This follows from \eqref{renyimeetshellinger},
and in view of the tightness of the constants in Theorem~\ref{thm: improved HausslerO}
(Remark~\ref{remark: tight constants}).
\end{remark}

\begin{remark} \label{remark2: tightness}
Let $P$, $Q = Q_\varepsilon$ be defined on a binary alphabet
with $P(0) = \tfrac12$ and $Q(0) = \tfrac12 - \varepsilon$.
Then,  it is easy to verify
that the ratio of the upper to lower bounds in Parts~\ref{RDRE:item:a}) and \ref{RDRE:item:b})
converges to 1 as $\varepsilon \to 0$.
\end{remark}

\begin{remark} \label{remark3: tightness}
Let $P$ and $Q$ be defined on a binary alphabet
with $P(0) = \tfrac12$, $Q(0) = \varepsilon \in (0,1)$.
Then, in the limit $\varepsilon \to 0$, the ratio of $D_\alpha(P \| Q)$
to the left side of \eqref{eq1.2: RD-RE} is equal to
$\frac{\alpha}{\log_2\left(\frac{2}{2-\alpha}\right)}\in (1, \log_e(4))$ for $\alpha \in (0,1)$.
Moreover, if $\varepsilon \to 0$, the ratio of $D_\alpha(P \| Q)$ and the right side of
\eqref{eq2.3: RD-RE} tends to~1 for $\alpha \in (1, \infty)$.

\par To prove the first part of Remark~\ref{remark3: tightness}, in view of
\eqref{eq1.2: RD-RE}, one needs to show that for $\alpha \in (0,1)$
\begin{align} \label{eq: bounded limit}
& \lim_{\varepsilon \downarrow 0} \frac{(\alpha-1) \,
d_{\alpha}(\tfrac12 \| \varepsilon)}{\log \Bigl(
1 + u_\alpha\bigl(\frac{1}{2\varepsilon}\bigr) \,
d(\tfrac12 \| \varepsilon)\Bigr)} =
\frac{\alpha}{\log_2 \left(\frac{2}{2-\alpha}\right)}
\end{align}
where $d(\cdot \| \cdot) \triangleq d_1(\cdot \| \cdot)$ and
$d_{\alpha}(\cdot \| \cdot)$ are given in \eqref{eq: binary RD}.
This can be verified by using \eqref{eq: binary RD} and \eqref{eq: u RD/RE}
to show that if $\alpha \in (0,1)$, then in the limit $\varepsilon \to 0$
\begin{align}
\label{a1}
& d\bigl(\tfrac12 \| \varepsilon\bigr)
= \tfrac12 \, \bigl(1+o(1)\bigr) \, \log\left(\tfrac1{2\varepsilon}\right), \\
\label{a2}
& (\alpha-1) \, d_{\alpha}\bigl(\tfrac12 \| \varepsilon\bigr)
= -\alpha \bigl(1+o(1)\bigr) \, \log 2, \\
\label{a3}
& u_{\alpha}\bigl(\tfrac1{2\varepsilon}\bigr)
= -\frac{\alpha \bigl(1+o(1)\bigr)}{\log\bigl(\frac1{2\varepsilon}\bigr)}.
\end{align}
Assembling \eqref{a1}--\eqref{a3} yields \eqref{eq: bounded limit}
whose right side is monotonically decreasing in $\alpha \in (0,1)$, and bounded
between~1 (by letting $\alpha \to 1$) and $\log_e(4)$ (by letting $\alpha \to 0$).

\par
To prove the second part of Remark~\ref{remark3: tightness}, in view of
\eqref{eq2.3: RD-RE}, one needs to show that for $\alpha \in (1,\infty)$
\begin{align} \label{eq: limit is 1}
\lim_{\varepsilon \downarrow 0} \frac{(\alpha-1)
\, d_{\alpha}(\tfrac12 \| \varepsilon)}{\log
\Bigl(1 + u_\alpha\bigl(\frac{1}{2\varepsilon}\bigr)
\, d(\tfrac12 \| \varepsilon)\Bigr)} = 1.
\end{align}
From \eqref{eq: binary RD} and \eqref{eq: u RD/RE}, it follows that in the limit
where $\varepsilon$ tends to zero
\begin{align}
\label{a4}
& (\alpha-1) \, d_{\alpha}\bigl(\tfrac12 \| \varepsilon\bigr)
= \log\left(2^{-\alpha} \, \varepsilon^{1-\alpha} \right)
\, \bigl(1+o(1)\bigr), \\[0.1cm]
\label{a5}
& u_{\alpha}\bigl(\tfrac1{2\varepsilon}\bigr)
= \frac{1+o(1)}{(2\varepsilon)^{\alpha-1}
\, \log\bigl(\frac1{2\varepsilon}\bigr)}.
\end{align}
Assembling \eqref{a1}, \eqref{a4} and \eqref{a5} yields
\eqref{eq: limit is 1} for $\alpha \in (1, \infty)$.
\end{remark}

\begin{example}
\begin{figure}[h]
\hspace*{0.2cm} \centerline{\includegraphics[width=10cm]{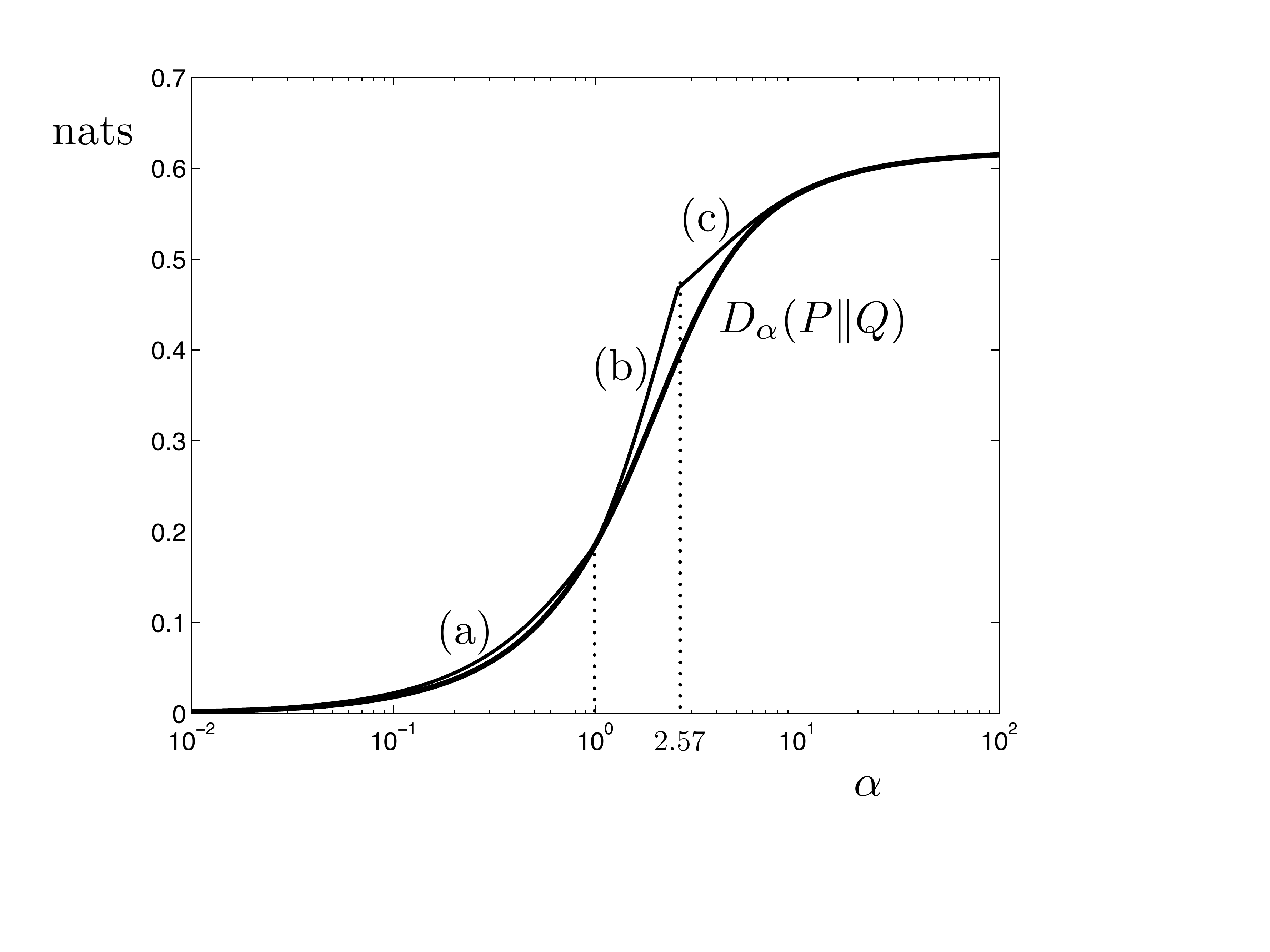}}
\vspace*{-1.5cm}
\caption{\label{figure:RE_RD}
The R\'{e}nyi divergence $D_{\alpha}(P\|Q)$ for $\set{A} = \{0,1\}$
with $P(0)=Q(1)=0.65$, compared to the tightest upper bound in
Theorem~\ref{theorem: RD-RE}: (a) \eqref{eq1.3: RD-RE} for $\alpha \in (0,1)$;
(b): \eqref{eq2.3: RD-RE} for $\alpha \in [1, 2.57]$; (c): \eqref{eq3: RD-RE}
for $\alpha > 2.57$.}
\end{figure}
Figure~\ref{figure:RE_RD} illustrates the upper bounds on
$D_{\alpha}(P \| Q)$ in Theorem~\ref{theorem: RD-RE} for
binary alphabets.
\end{example}

\begin{remark}
\cite[Proposition~11]{FongT_MMNs} shows an upper bound
on $D_\alpha(P\|Q)$ for $\alpha \in [1, \tfrac54]$, which is
expressed in terms of $D(P\|Q)$ and the finite cardinalities
of the alphabets over which $P$ and $Q$ are defined. Although
the bound in \cite[(9)]{FongT_MMNs} is not tight, it leads to
a strong converse for a certain class of discrete memoryless
networks.
\end{remark}

\section{Summary} \label{sec: summary}

Since many distance measures of interest fall under the
common paradigm of an $f$-divergence, it is not surprising that
bounds on the ratios of various $f$-divergences are useful
in many instances such as proving convergence of probability
measures according to various metrics, analysis of rates of
convergence and concentration of measure bounds
\cite{boucheron2013concentration,GibbsSu02,OrdentlichW_IT2005,
raginsky2012concentration,samson2000concentration,wongshen95},
hypothesis testing \cite{DeGroot62}, testing goodness of fit
\cite{HarremoesV_2011, ReadC88}, minimax risk in estimation and modeling
\cite{Guntuboyina11,HausslerO97,ReidW11,Vapnik98}, strong data processing
inequality constants and maximal correlation
\cite{ahlswede1976spreading,PW15,Raginsky16}, transportation-cost inequalities
\cite{boucheron2013concentration,Marton13b,raginsky2012concentration,raginsky2015concentration},
contiguity \cite{lecamyoung,LieseV_book87}, etc.

While the derivation of $f$-divergence inequalities has received considerable
attention in the literature, the proof techniques have been tailored to the
specific instances. In contrast, we have proposed several systematic approaches
to the derivation of $f$-divergence inequalities.
Introduced in Section~\ref{subsec: basic tool of functional domination},
functional domination emerges as a basic tool to obtain $f$-divergence inequalities.
Another basic tool that capitalizes on many cases of interest (including the finite
alphabet one) is introduced in Section~\ref{subsec:bounds among fD}, where
not only one of the distributions is absolutely continuous with respect to
the other but their relative information is almost surely bounded.
\par
Section~\ref{subsec: bounds-RE} illustrates the use of moment inequalities and the
log-convexity property, while the utility of Lipschitz constraints in deriving bounds
is highlighted in Section~\ref{subsec:Lipschitz}.
\par
In addition, new $f$-divergence inequalities (frequently with optimal constants) arise from:

\begin{itemize}
\item
integral representation of $f$-divergences, expressed in terms of the $E_\gamma$ divergence (Section~\ref{saens});
\item
extension of Pinsker's inequality to $E_\gamma$ divergence (Section~\ref{subsec: Pinsker for EG});
\item
a relation between the relative information and the relative entropy (Section~\ref{subsec: re-ris});
\item
exact expressions of R\'enyi divergence in terms of the relative information spectrum (Section~\ref{subsec: RD-RIS});
\item
the exact locus of the entropy and the variational distance from the equiprobable
probability mass function (Section~\ref{subsec: Distance From the Equiprobable Distribution}).
\end{itemize}

\appendices

\section{Completion of the Proof of Theorem~\ref{thm: bounds RE and dual}}
\label{appendix: properties of kappa}
\begin{lemma}  \label{lemma: properties of kappa}
The function $\kappa \colon (0, \infty) \to (0, \infty)$ which is the continuous extension
of the function in \eqref{eq: kappa RE and dual} with $\kappa(1)=1$ is strictly
monotonically increasing.
\end{lemma}
\begin{proof}
From \eqref{eq: kappa RE and dual} and \eqref{1g},
\begin{align}
\kappa'(t) = \frac{(t-1)^2 \log^2 e - t \log^2 t }{t \, g^2(t)},
\end{align}
if $ t \in (0, 1) \cup (1,\infty)$, while $\kappa'(1) = \tfrac13$.
To show
\begin{align}
(t-1)^2 \, \log^2 e -t \log^2 t  > 0, \quad \forall \, t \in (0,1) \cup (1, \infty).
\label{eq: kappa-step 1}
\end{align}
we substitute $t = \exp(x)$ to obtain that if $x \neq 0$, then
\begin{align}
s(x) = \exp(2x)-(2+x^2)\exp(x)+1 > 0, \;
\label{eq: kappa-step 2}
\end{align}
which holds since $s(0)=0$ and the derivative
\begin{align}
s'(x) = 2 \exp(x) \left[ \exp(x) - \left(1+x+\frac{x^2}{2}\right) \right]
\end{align}
is negative on $(-\infty, 0)$ and positive
on $(0, \infty)$.
\end{proof}

%%%%%%%%%%%%%%%%%%%%%%%%%%%%%%%%%%%%%%%%%%%%%%%%%%%%
\section{Proof of the Monotonicity of $\kappa_{\alpha}$ in \eqref{eq: kappa RE/HD}}
\label{appendix: monotonicity of kappa_alpha}
To show that
 the function $\kappa_\alpha \colon [0, \infty] \to [0, \infty]$ in \eqref{eq: kappa RE/HD}
 is monotonically increasing on $[0, \infty]$ if
$\alpha \in (0,1)$ and monotonically decreasing on $[0, \infty]$ if $\alpha \in (1,\infty)$
it is sufficient to show that
\begin{align} \label{eq1: diff kappa}
\frac{\text{d}}{\text{d}t} \left(\frac{r(t)}{1-t^\alpha + \alpha (t-1)} \right) > 0
\end{align}
for $(\alpha,t)\in \set{F} = ((0,1)\cup(1 , \infty))^2 $.
From \eqref{eq: r}, straightforward calculus gives
\begin{align}
& \left[ 1 - t^\alpha + \alpha (t-1) \right]^2 \; \frac{\text{d}}{\text{d}t}
\left(\frac{r(t)}{1-t^\alpha + \alpha (t-1)} \right)  \\
& = (1-\alpha) (1-t^\alpha) \log t - \alpha (1-t) (1 - t^{\alpha-1}) \log e  \\
\label{eq2: diff kappa}
& \triangleq g_{\alpha}(t)
\end{align}
so the desired result will follow upon showing
\begin{align} \label{eq: g inequality}
& g_{\alpha}(t) > 0,   \quad (\alpha,t)\in \set{F}.
\end{align}
Note that $g_\alpha(1)=0$. For $(\alpha,t)\in \set{F}$,
it is easy to verify that
\begin{align} \label{eq: positive}
(1-\alpha)(t-1)(1-t^{\alpha-1}) > 0.
\end{align}
A division of \eqref{eq: g inequality} by the positive left side of \eqref{eq: positive}
gives the following equivalent inequality:
\begin{align} \label{eq: equiv. ineq.}
& \phi_\alpha(t) > \frac{\alpha}{\alpha-1}, \quad
\forall \, (\alpha,t)\in \set{F}
\end{align}
where, for $\alpha \in (0,1) \cup (1, \infty)$,
\begin{align} \label{eq: phi}
\phi_\alpha(t) \triangleq \left\{ \begin{array}{ll}
               \frac{(1-t^\alpha) \log_e t}{(t-1)(1-t^{\alpha-1})}
               & \mbox{if $t \in (0,1) \cup (1, \infty)$} \\
               \frac{\alpha}{\alpha-1} & \mbox{if $t=1$.}
               \end{array}
               \right.
\end{align}
We aim to prove \eqref{eq: equiv. ineq.}. Note that, for $\alpha \in (0,1) \cup (1, \infty)$,
\begin{align}
& \lim_{t \to 1} \phi_\alpha(t) = \frac{\alpha}{\alpha-1} = \phi_\alpha(1) \label{eq2: phi}
\end{align}
so, \eqref{eq: equiv. ineq.} is implied by proving that $\phi_\alpha \colon (0, \infty) \to \Reals$
is monotonically decreasing on $(0,1)$, and it is monotonically increasing on $(1, \infty)$. For
this purpose, we rely on the following lemmas.

\begin{lemma} \label{lemma: phi/psi}
For every $t>0$ and $\alpha \in (0,1) \cup (1, \infty)$
\begin{align} \label{eq3: phi}
\phi_\alpha(t) = \psi(t) - \left(\frac1{1-\alpha}\right) \, \psi(t^{1-\alpha})
\end{align}
where
\begin{align} \label{eq: psi}
\psi(t) \triangleq \left\{ \begin{array}{ll}
                   \frac{\log_e t}{t-1} & \mbox{if $t \in (0,1) \cup (1, \infty)$} \\
                   1 & \mbox{if $t=1$.}
                   \end{array}
                   \right.
\end{align}
\end{lemma}

\begin{proof}
For $\alpha, t \in (0,1) \cup (1, \infty)$
\begin{align}
\label{eq2: psi}
\psi(t) - \left(\frac1{1-\alpha}\right) \, \psi(t^{1-\alpha})
& = \frac{(t^{1-\alpha}-t) \, \log_e t}{(t-1)(t^{1-\alpha}-1)} \\
\label{eq3: psi}
& = \frac{(1-t^\alpha) \, \log_e t}{(t-1) (1-t^{\alpha-1})} \\
\label{eq3.5: psi}
& = \phi_{\alpha}(t)
\end{align}
where \eqref{eq2: psi} holds due to \eqref{eq: psi},
\eqref{eq3: psi} is justified by multiplying the numerator and denominator
of \eqref{eq2: psi} by $t^{\alpha-1}$, and \eqref{eq3.5: psi} is due
to \eqref{eq: phi}. Note that \eqref{eq3: phi} is also satisfied at
$t=1$ due to the continuity of $\phi_\alpha$ and $\psi$ at this point.
\end{proof}

\begin{lemma} \label{lemma: psi}
The following inequality holds for $z>0$:
\begin{align} \label{eq4: psi}
\psi'(z) + z \psi''(z) > 0,
\end{align}
\end{lemma}
\begin{proof}
From the power series expansion of $\log_e z$ around $z=1$, we get for $ 0 < z \leq 2$
\begin{align} \label{eq: power series}
\frac{\log_e z}{z-1} = 1 - \tfrac12 (z-1) + \tfrac13 (z-1)^2 - \tfrac14 (z-1)^3 + \ldots
\end{align}
From \eqref{eq: psi} and \eqref{eq: power series}, $\psi'(1) = -\tfrac12$ and
$\psi''(1) = \frac{2}{3}$; at $z=1$, the left side of \eqref{eq4: psi} is equal to
\begin{align} \label{eq4.1: psi}
\psi'(1)+\psi''(1) = \tfrac1{6} > 0.
\end{align}
For $z \in (0,1) \cup (1, \infty)$, the left side of
\eqref{eq4: psi} satisfies
\begin{align} \label{eq5: psi}
\psi'(z) + z \psi''(z) = \frac{m(z)}{z(z-1)^3}
\end{align}
where
\begin{align} \label{eq1: m}
m(z) \triangleq (z+1) \log_e z - 2(z-1), \quad \forall \, z > 0.
\end{align}
From \eqref{eq1: m}
\begin{align}
& m(1) = 0,   \label{eq2: m} \\
& m'(z) = \frac1{z} - 1 - \log_e \left(\frac1{z}\right) > 0,
\quad \forall \, z \in (0,1) \cup (1, \infty)  \label{eq3: m}
\end{align}
which implies that $m \colon (0, \infty) \to \Reals$ is monotonically increasing, positive
for $z>1$ and negative for $z<1$. These facts together with \eqref{eq4.1: psi} and
\eqref{eq5: psi} yield that \eqref{eq4: psi} holds for all $z > 0$.
\end{proof}

We proceed now with the proof of \eqref{eq: equiv. ineq.}.
For $\alpha \in (0,1) \cup (1, \infty)$, we have
\begin{align}
\label{eq6: psi}
\frac{\partial}{\partial \alpha} \phi'_\alpha(t)
& = \frac{\partial^2}{\partial \alpha \; \partial t}
\left[ \psi(t) - \left(\frac1{1-\alpha}\right) \, \psi(t^{1-\alpha}) \right] \\
\label{eq6.1:psi}
& = \frac{\partial}{\partial \alpha} \Bigl( \psi'(t) - t^{-\alpha} \, \psi'(t^{1-\alpha}) \Bigr) \\
\label{eq7: psi}
& = - \frac{\partial}{\partial \alpha} \Bigl( t^{-\alpha} \, \psi'(t^{1-\alpha}) \Bigr) \\
\label{eq8: psi}
& = t^{-\alpha} \; \log_e(t) \; \psi'(t^{1-\alpha}) - t^{-\alpha} \; \frac{\partial}{\partial \alpha}
 \psi'(t^{1-\alpha})  \\
\label{eq9: psi}
& = t^{-\alpha} \; \log_e(t) \Bigl[ \psi'(t^{1-\alpha}) + t^{1-\alpha} \; \psi''(t^{1-\alpha}) \Bigr]
\end{align}
where \eqref{eq6: psi} follows from \eqref{eq3: phi}. From Lemma~\ref{lemma: psi}, for $t>0$,
\begin{align} \label{eq10: psi}
\psi'(t^{1-\alpha}) + t^{1-\alpha} \; \psi''(t^{1-\alpha}) > 0.
\end{align}
From \eqref{eq9: psi} and \eqref{eq10: psi}, for $\alpha \in (0,1) \cup (1, \infty)$, and  $t \in (0,1)$
\begin{align} \label{eq90: phi}
\frac{\partial}{\partial \alpha}  \phi'_\alpha(t) < 0,
\end{align}
and for $ t \in (1,\infty)$,
\begin{align} \label{eq91: phi}
\frac{\partial}{\partial \alpha}  \phi'_\alpha(t)  > 0.
\end{align}
From \eqref{eq3: phi} and the continuity of $\psi'$ on $(0, \infty)$ (see \eqref{eq: psi}
and \eqref{eq: power series}), for all $t > 0$,
\begin{align} \label{eq:phi99}
\lim_{\alpha \downarrow 0} \phi'_\alpha(t)
& = \lim_{\alpha \downarrow 0} \bigl(\psi'(t) - t^{-\alpha} \, \psi'(t^{1-\alpha}) \bigr) = 0.
\end{align}
Combining \eqref{eq90: phi} and \eqref{eq:phi99} gives that, for $\alpha \in (0,1)$ and $t \in (0,1)$
\begin{align}  \label{eq:phi100}
\phi'_{\alpha}(t) < 0,
\end{align}
and, combining \eqref{eq91: phi} and \eqref{eq:phi99} gives that, for $\alpha \in (0,1)$ and $t \in (1,\infty)$,
\begin{align}  \label{eq:phi101}
\phi'_{\alpha}(t) > 0.
\end{align}
Hence, for $\alpha \in (0,1)$,  $\phi_\alpha \colon (0, \infty) \to \Reals$ is monotonically
decreasing on $(0,1)$, and it is monotonically increasing on $(1, \infty)$.

We now consider the case where $\alpha \in (1, \infty)$.
Since, from \eqref{eq3: phi} and \eqref{eq: power series},
\begin{align} \label{eq:phi102}
\lim_{\alpha \to 1} \phi'_\alpha(t) & = \lim_{\alpha \to 1} \bigl(\psi'(t) - t^{-\alpha} \, \psi'(t^{1-\alpha}) \bigr) \\
& = \psi'(t) + \frac{\psi'(1)}{t} \\
& = \psi'(t) + \frac1{2t}
\end{align}
then, the existence of this limit in \eqref{eq:phi102} yields that its one-sided
limits are equal, i.e.,
\begin{align} \label{eq:phi103}
\lim_{\alpha \downarrow 1} \phi'_\alpha(t) = \lim_{\alpha \uparrow 1} \phi'_\alpha(t).
\end{align}
Consequently, \eqref{eq:phi100}, \eqref{eq:phi101} and \eqref{eq:phi103} yield that
\begin{align}
\begin{split}
\label{eq:phi104}
& \lim_{\alpha \downarrow 1} \phi'_\alpha(t) \geq 0, \quad \forall \, t \in (0,1), \\
& \lim_{\alpha \downarrow 1} \phi'_\alpha(t) \leq 0, \quad \forall \, t \in (1,\infty)
\end{split}
\end{align}
and, from \eqref{eq90: phi}, \eqref{eq91: phi} and \eqref{eq:phi104}, we conclude
that \eqref{eq:phi100} and \eqref{eq:phi101} also hold for $\alpha \in (1, \infty)$.
The property that $\phi_\alpha \colon (0, \infty) \to \Reals$ is
monotonically decreasing on $(0,1)$ and monotonically increasing on $(1, \infty)$
is therefore extended also to $\alpha \in (1, \infty)$.
As explained after \eqref{eq2: phi}, this implies
the satisfiability of \eqref{eq: equiv. ineq.}. Consequently, also \eqref{eq: g inequality} holds,
which implies that $\kappa_\alpha \colon [0, \infty] \to [0, \infty]$, defined in
\eqref{eq: kappa RE/HD}, is monotonically increasing for all $\alpha \in (0,1)$,
and it is monotonically decreasing for all $\alpha \in (1,\infty)$.
%%%%%%%%%%%%%%%%%%%%%%%%%%%%%%%%%%%%%%%%%%%%%%%%%%%%

\section{Proof of \eqref{eq: 2nd condition}}
\label{appendix: proof of 2nd condition}
To verify that \eqref{eq: 2nd condition} follows from \eqref{eq: 1st condition},
fix arbitrarily small $\varepsilon > 0$ and $\rho > 0$.
Consider the partition
$\set{A} = \set{A}_1^{(n)} \cup \set{A}_2^{(n)} \cup \set{A}_3^{(n)}$ with
\begin{align}
\label{eq: A1}
& \set{A}_1^{(n)} \triangleq \left\{ a \in \set{A} \colon
\frac{\text{d}P_n}{\text{d}Q} \, (a) \in [0, 1-\rho] \right\}, \\[0.1cm]
\label{eq: A2}
& \set{A}_2^{(n)} \triangleq \left\{ a \in \set{A} \colon
\frac{\text{d}P_n}{\text{d}Q} \, (a) \in (1-\rho, 1+\varepsilon] \right\}, \\[0.1cm]
\label{eq: A3}
& \set{A}_3^{(n)} \triangleq \left\{ a \in \set{A} \colon
\frac{\text{d}P_n}{\text{d}Q} \, (a) \in (1+\varepsilon, \infty) \right\},
\end{align}
then
\begin{align}
\label{eq2: intergal1}
 I_1^{(n)} + I_2^{(n)} + I_3^{(n)} = 1
\end{align}
where
\begin{align}
\label{eq: I123}
& I_j^{(n)} \triangleq \int_{\set{A}_j^{(n)}}
\frac{\text{d}P_n}{\text{d}Q} \, (a) \; \text{d}Q(a).
\end{align}

From the assumption in \eqref{eq: 1st condition}, $I_3^{(n)} \to 0$
when $n \to \infty$ since for all sufficiently large~$n$
\begin{align} \label{eq: QA3}
Q(\mathcal{A}_3^{(n)})=0.
\end{align}
Let
\begin{align} \label{eq: dn}
d_n \triangleq Q(\set{A}_1^{(n)})
\end{align}
then, from \eqref{eq: QA3}, for all sufficiently large $n$
\begin{align}
\label{eq: 1-dn}
1-d_n = Q(\set{A}_2^{(n)}).
\end{align}
Consequently, from \eqref{eq: A1}, \eqref{eq: A2}, \eqref{eq2: intergal1},
\eqref{eq: QA3}, \eqref{eq: dn} and \eqref{eq: 1-dn}, it follows
that for all sufficiently large $n$,
\begin{align}
\label{eq1: mu}
1 &= I_1^{(n)} + I_2^{(n)} \\
\label{eq2: mu}
& \leq d_n (1-\rho) + (1-d_n) (1+\varepsilon) \triangleq \mu_n.
\end{align}
If $\liminf d_n = 0$ for an arbitrarily small $\rho > 0$ then
\eqref{eq: 2nd condition} holds by the definition in \eqref{eq: dn}.
Assuming otherwise, namely,
\begin{align}
\label{eq: false assumption}
\liminf d_n = \theta \in (0,1)
\end{align}
leads to the following contradiction:
\begin{align}
\label{eq: contradict1}
1 & \leq \liminf \mu_n \\
\label{eq: contradict2}
& \leq (1-\rho) \, \liminf d_n + (1 + \varepsilon) \, \limsup \, (1-d_n) \\
\label{eq: contradict3}
& = \theta (1-\rho) + (1-\theta) (1 + \varepsilon) \\
\label{eq: contradict4}
& = 1 - \frac{\theta \rho}{2}
\end{align}
where $\varepsilon = \frac{\theta \rho}{2(1-\theta)}$;
\eqref{eq: contradict1} follows from \eqref{eq1: mu}, \eqref{eq2: mu};
\eqref{eq: contradict2} holds by \eqref{eq2: mu}; \eqref{eq: contradict3}
is due to \eqref{eq: false assumption}.

%%%%%%%%%%%%%%%%%%%%%%%%%%%%%%%%%%%%%%%%%%%%%%%%%%%%%%%%%%%%%%%%%%%%%%%%%%%%%%%%%%%%%
\section{Proof of Theorem~\ref{thm: tv}}
\label{appendix: tv}
Eq.~\eqref{eq: TV1} follows from the definitions in \eqref{eq:RI} and
\eqref{eq2: TV distance}. Since $z^+ = \tfrac12 \bigl(|z|+z\bigr)$
and $z^- = \tfrac12 \bigl(|z|-z\bigr)$, for all $z \in \Reals$,
\eqref{eq: TV2} and \eqref{eq: TV3} follow from \eqref{eq: TV1} and
\begin{align}
\mathbb{E} \bigl[ 1 - \exp(\imath_{P\|Q}(Y)) \bigr] =
\int \left( 1 - \frac{\mathrm{d}P}{\mathrm{d}Q} \right) \, \mathrm{d}Q = 0.
\end{align}

By change of measure, for every measurable function $f \colon \set{A} \to \Reals$
with $\mathbb{E} \bigl[f(X)\bigr] < \infty$ and $\mathbb{E} \bigl[f(Y)\bigr] < \infty$,
\begin{align}
\label{eq: cm}
\mathbb{E} \bigl[f(X)\bigr] &= \mathbb{E} \Bigl[\frac{\text{d}P}{\text{d}Q}(Y) \, f(Y)\Bigr]
= \mathbb{E} \bigl[ \exp\bigl(\imath_{P\|Q}(Y)\bigr) \, f(Y)\bigr].
\end{align}
Hence, it follows from \eqref{eq: cm} that
\begin{align}
& \hspace*{-0.2cm} \mathbb{P} \bigl[\imath_{P\|Q}(X) > 0\bigr]
= \mathbb{E} \bigl[1\bigl\{\imath_{P\|Q}(X) > 0\bigr\}\bigr]  \\
\label{eq: cm1}
& \hspace*{2.25cm} = \mathbb{E} \bigl[\exp(\imath_{P\|Q}(Y)) 1\bigl\{\imath_{P\|Q}(Y) \hspace*{-0.05cm} > \hspace*{-0.05cm} 0\bigr\} \bigr]
\end{align}
and
\begin{align}
& \hspace*{-0.2cm} \mathbb{P} \bigl[\imath_{P\|Q}(X) \leq 0\bigr]
\hspace*{-0.01cm} = \mathbb{E} \bigl[1\bigl\{\imath_{P\|Q}(X) \leq 0\bigr\}\bigr]  \\
\label{eq: cm2}
& \hspace*{2.21cm} = \mathbb{E} \bigl[\exp(\imath_{P\|Q}(Y)) 1\bigl\{\imath_{P\|Q}(Y)
\hspace*{-0.08cm} \leq \hspace*{-0.08cm} 0 \bigr\} \bigr].
\end{align}

To show \eqref{eq: TV4}, note that from \eqref{eq: TV3} and the change of measure in
\eqref{eq: cm}, we get
\begin{align}
\tfrac12 \, |P-Q|
& = \mathbb{E} \bigl[ \bigl( 1 - \exp(\imath_{P\|Q}(Y)) \bigr)^- \bigr]  \\
& = \mathbb{E} \bigl[ \bigl( \exp(\imath_{P\|Q}(Y)) - 1 \bigr)
\; 1\bigl\{ \imath_{P \| Q}(Y) > 0 \bigr\} \bigr] \\
& = \mathbb{E} \bigl[ \bigl( 1 - \exp(-\imath_{P\|Q}(X)) \bigr)
\; 1\bigl\{ \imath_{P \| Q}(X) > 0 \bigr\} \bigr] \\
& = \mathbb{E} \bigl[ \bigl( 1 - \exp(-\imath_{P\|Q}(X)) \bigr)^{+} \bigr].
\end{align}

To show \eqref{eq: TV4a} and \eqref{eq: TV4b}, we get from \eqref{eq: TV3} and \eqref{eq: cm1}
\begin{align}
\tfrac12 \, |P-Q|
\label{eq1: TV4a}
& = \mathbb{E} \bigl[ \bigl( 1 - \exp(\imath_{P\|Q}(Y)) \bigr)^- \bigr] \\
\label{eq2: TV4a}
& = \mathbb{E} \Bigl[ \bigl(\exp(\imath_{P\|Q}(Y)) - 1 \bigr) \; 1\bigl\{\imath_{P\|Q}(Y) > 0\bigr\} \Bigr] \\
\label{eq3: TV4a}
& = \mathbb{P} \bigl[\imath_{P\|Q}(X) > 0\bigr] - \mathbb{P} \bigl[\imath_{P\|Q}(Y) > 0\bigr]
\end{align}
where \eqref{eq3: TV4a} is \eqref{eq: TV4a}, and \eqref{eq: TV4b} is equivalent to \eqref{eq: TV4a}.

To show \eqref{eq: TV5}, we use \eqref{eq: TV3} and the notation in
\eqref{eq: Z} in order to write
\begin{align}
\label{zza}
\tfrac12 \, |P-Q| &= \mathbb{E} \bigl[ (1-Z)^- \bigr]  \\
& = \mathbb{E} \bigl[ (Z-1) \, 1\{Z>1\} \bigr] \\
\label{zzb}
& = \int_0^{\infty} \mathbb{P} \bigl[ (Z-1) \, 1\{Z>1\} \geq \beta \bigr] \, \text{d} \beta \\
& = \int_1^{\infty} \mathbb{P} \bigl[ Z \geq \beta \bigr] \, \text{d} \beta \\
\label{zzc}
& = \int_0^1 \mathbb{P} \bigl[ Z < \beta \bigr] \, \text{d} \beta
\end{align}
where \eqref{zza} follows from \eqref{eq: TV3} with $Z$ in \eqref{eq: Z};
\eqref{zzb} exploits the fact that the expectation of a non-negative random
variable is the integral of its complementary cumulative distribution function;
and \eqref{zzc} is satisfied since $Z$ is non-negative with $\mathbb{E}[Z] = 1$.

To show \eqref{eq: TV6}, we use \eqref{eq: TV4} to write
\begin{align}
\tfrac12 |P-Q| \nonumber
& = \mathbb{E} \bigl[\bigl( 1 -
\exp(-\imath_{P\|Q}(X)) \bigr)^{+} \bigr] \\
& = \int_0^\infty \mathbb{P} \bigl[ \bigl( 1 -
\exp(-\imath_{P\|Q}(X)) \bigr)^{+} > \beta \bigr] \,
\text{d} \beta \\
& = \int_0^1 \mathbb{P} \bigl[ \bigl( 1 -
\exp(-\imath_{P\|Q}(X)) \bigr)^{+} > \beta \bigr] \,
\text{d} \beta \\
& = \int_0^1 \mathbb{P} \left[ \imath_{P\|Q}(X) >
\log \frac1{1-\beta} \right] \, \text{d} \beta \\
& = \int_0^1 \mathbb{P} \left[ \imath_{P\|Q}(X) >
\log \frac1{\beta} \right] \, \text{d} \beta.
\end{align}

To prove \eqref{eq: TV100},
a change of variable of integration in \eqref{eq: TV6},
and the fact that $\mathds{F}_{P\|Q} ( \log \beta ) = 1$
for $\beta > \beta_1^{-1}$ give
\begin{align}
\tfrac12 |P-Q| & = \int_0^1 \mathbb{P}\Bigl[\imath_{P \| Q}(X) >
\log \frac1{t} \Bigr] \, \text{d}t \\[0.1cm]
& = \int_0^1 \Bigl[ 1 - \mathds{F}_{P\|Q}\Bigl(\log
\frac1{t}\Bigr) \Bigr] \text{d}t \\[0.1cm]
& = \int_1^{\infty} \frac{1-\mathds{F}_{P\|Q}(\log
\beta)}{\beta^2} \; \text{d}\beta \\[0.1cm]
& = \int_1^{\beta_1^{-1}} \frac{1-\mathds{F}_{P\|Q}(\log
\beta)}{\beta^2} \; \text{d}\beta
\end{align}
with the convention that $\beta_1^{-1} = \infty$ if $\beta_1 = 0$.

Assume that $P \ll \gg Q$. To show \eqref{eq: TV7} simply note that
\eqref{eq: TV1}, the symmetry of the total variation distance,
and the anti-symmetry of the relative information where
$\imath_{Q\|P} = -\imath_{P\|Q}$ enable to conclude that
\begin{align}
|P-Q| & = \mathbb{E} \bigl[ \bigl| 1 - \exp(\imath_{Q\|P}(X)) \bigr| \bigr]  \\
& = \mathbb{E} \bigl[ \bigl| 1 - \exp(-\imath_{P\|Q}(X)) \bigr| \bigr].
\end{align}
Similarly, switching $P$
and $Q$ in \eqref{eq: TV3} results in
\begin{align}
\tfrac12 \, |P-Q| & = \mathbb{E} \bigl[ \bigl( 1 -
\exp\bigl(\imath_{Q \| P}(X)\bigr) \bigr)^{-} \bigr] \\
& = \mathbb{E} \bigl[ \bigl( 1 -
\exp\bigl(-\imath_{P \| Q}(X)\bigr) \bigr)^{-} \bigr]
\end{align}
which proves \eqref{eq: TV8}.

%%%%%%%%%%%%%%%%%%%%%%%%%%%%%%%%%%%%%%%%%%%%%%%%%%%%
\section{\eqref{eq: improved SV-ITA14} vs. \eqref{eq: SV-ITA14}}
\label{appendix:improvement}

\subsection{Example for the Strengthened Inequality in Theorem~\ref{thm: improved SV-ITA14}}
\label{subsec: example-strengthened RPI}
We exemplify  the improvement obtained by \eqref{eq: improved SV-ITA14},
in comparison to~\eqref{eq: SV-ITA14}, due to the introduction of the additional parameter
$\beta_2$ in \eqref{eq: beta2}. Note that when $\beta_2$ is replaced by zero (i.e., no
information on the infimum of $\frac{\text{d}P}{\text{d}Q}$ is available or $\beta_2=0$),
inequalities \eqref{eq: improved SV-ITA14} and \eqref{eq: SV-ITA14} coincide.

Let $P$ and $Q$ be two probability measures, defined on $(\set{A}, \mathscr{F})$, $P \ll Q$,
and assume that
\begin{align} \label{eq: condition for 1st example}
1-\eta \leq \frac{\text{d}P}{\text{d}Q} \, (a) \, \leq 1+\eta, \quad \forall \, a \in \set{A}
\end{align}
for a fixed $\eta \in (0,1)$.

In \eqref{eq: improved SV-ITA14}, one can replace $\beta_1$ and $\beta_2$ with
lower bounds on these constants. Since
$\beta_1 \geq \frac1{1+\eta}$ and $\beta_2 \geq 1-\eta$ it follows from
\eqref{eq: improved SV-ITA14} that
\begin{align}
D(P\|Q)
& \leq \tfrac12 \left( \frac{(1+\eta) \, \log(1+\eta)}{\eta} + \frac{(1-\eta) \,
\log(1-\eta)}{\eta} \right) \, |P-Q|  \\
& \leq \eta \log e \cdot |P-Q|.
\label{eq: relative entropy - bound1}
\end{align}
From \eqref{eq: condition for 1st example}
\begin{align}
\bigl| \exp\bigl(\imath_{P\|Q}(a)\bigr) - 1 \bigr| \leq \eta, \quad \forall \, a \in \set{A}
\end{align}
so, from \eqref{eq: TV1}, the total variation distance satisfies (recall that $Y \sim Q$)
\begin{align} \label{igalast2}
|P-Q| =  \mathbb{E}\Bigl[ \bigl| \exp\bigl(\imath_{P\|Q}(Y)\bigr) - 1 \bigr| \Bigr]\leq \eta.
\end{align}
Combining \eqref{igalast2} with \eqref{eq: relative entropy - bound1} yields
\begin{align}  \label{eq: ubresv2}
D(P\|Q) \leq \eta^2 \, \log e, \quad \forall \, \eta \in (0,1).
\end{align}
For comparison, it follows from \eqref{eq: SV-ITA14} (see \cite[Theorem~7]{Verdu_ITA14}) that
\begin{align}  \label{eq: ubresv}
D(P\|Q) & \leq \frac{\log \frac1{\beta_1}}{2(1-\beta_1)} \cdot |P-Q|  \\[0.1cm]
& \leq \frac{(1+\eta) \, \log(1+\eta)}{2 \eta} \cdot |P-Q|  \\
& \leq \tfrac12 \, (1+\eta) \log(1+\eta)  \\
& \leq \tfrac12 \, \eta (1+\eta) \log e.
\end{align}
The upper bound on the relative entropy in \eqref{eq: ubresv} scales like $\eta$, for small $\eta$,
whereas the tightened bound in \eqref{eq: ubresv2} scales like $\eta^2$, which is tight according
to Pinsker's inequality \eqref{eq: Pinsker}. For example, consider the probability measures defined
on a two-element set $\set{A} = \{a,b\}$ with
\begin{align}
P(a) = Q(b) = \tfrac12 - \tfrac{\eta}{4}, \quad P(b) = Q(a) = \tfrac12 + \tfrac{\eta}{4}.
\end{align}
Condition \eqref{eq: condition for 1st example} is satisfied for $\eta \approx 0$, and Pinsker's inequality
\eqref{eq: Pinsker} yields
\begin{align} \label{eq: lbrepi}
D(P\|Q) \geq  \tfrac12  \eta^2 \log e
\end{align}
so the ratio of the upper and lower bounds in \eqref{eq: ubresv2} and \eqref{eq: lbrepi}
is~2, and both provide the true quadratic scaling in $\eta$ whereas the weaker upper bound
in \eqref{eq: ubresv} scales linearly in $\eta$ for $\eta \approx 0$.

%%%%%%%%%%%%%%%%%%%%%%%%%%%%%%%%%%%%%%%%%%%%%%%%%%%%
\section{Derivation of \eqref{eq: ub RE RPI}--\eqref{eq: ub RE-TV-chi}}
\label{appendix: ub RE-TV-chi}
Similarly to the proof of Theorem~\ref{thm: improved SV-ITA14}, let $X \sim P$, $Y \sim Q$,
and $Z = \exp\bigl( \imath_{P\|Q}(Y) \bigr)$.
We rely on the concavity of $\varphi \colon [0, \infty) \to [0, \infty)$, defined to be the continuous
extension of $\frac{t \log t}{t-1}$, for tightening the upper bound in \eqref{eq: bound on 1st summand}.
The combination of this tightened bound with \eqref{eq: relative entropy} and \eqref{eq: bound on 2nd summand}
serves to derive a tighter bound on the relative entropy in comparison to \eqref{eq: improved SV-ITA14}.

Since $Z \leq \beta_1^{-1}$, and $\varphi$ is concave, monotonically increasing
and differentiable, we can write
\begin{align}  \label{eq: iubphi}
\varphi(Z) \leq \varphi(\beta_1^{-1}) - \varphi'(\beta_1^{-1}) \, (\beta_1^{-1}-Z) \leq \varphi(\beta_1^{-1})
\end{align}
which improves the upper bound on $\varphi(Z)$ in \eqref{eq: range of values for phi of Z}. Consequently, from
\eqref{eq: iubphi}, the first summand in the right side of \eqref{eq: relative entropy} is upper bounded as follows:
\begin{align}
\mathbb{E} \bigl[ \varphi(Z) \, (Z-1) \, 1\{Z>1\} \bigr]
& \leq \mathbb{E} \Bigl[ \Bigl( \varphi(\beta_1^{-1}) - \varphi'(\beta_1^{-1}) \, (\beta_1^{-1}-Z) \Bigr) \,
(Z-1) \, 1\{Z>1\} \Bigr]  \\[0.1cm]
& = \Bigl( \varphi(\beta_1^{-1}) - \varphi'(\beta_1^{-1}) \, \beta_1^{-1} \Bigr) \;
\mathbb{E}\bigl[(Z-1) \, 1\{Z>1\} \bigr] \nonumber \\
& \hspace*{0.3cm} + \varphi'(\beta_1^{-1}) \, \mathbb{E} \bigl[Z(Z-1) \, 1\{Z>1\}\bigr]  \\[0.1cm]
& = \tfrac12 \Bigl( \varphi(\beta_1^{-1}) - \varphi'(\beta_1^{-1}) \, \beta_1^{-1} \Bigr) \; |P-Q| \nonumber \\
& \hspace*{0.3cm} + \varphi'(\beta_1^{-1}) \, \mathbb{E} \bigl[Z(Z-1) \, 1\{Z>1\}\bigr]
\label{eq: ib 1st summand}
\end{align}
where \eqref{eq: ib 1st summand} follows from  \eqref{eq: Z} and  \eqref{eq: TV2}.
Combining \eqref{eq: relative entropy}, \eqref{eq: bound on 2nd summand} and \eqref{eq: ib 1st summand}
gives the upper bound on the relative entropy in \eqref{eq: ub RE RPI}.

The second term in the right side of \eqref{eq: ib 1st summand} depends on the distribution of the relative
information. To circumvent this dependence, we derive upper and lower bounds in terms of $f$-divergences.
\begin{align} \label{eq: 2nd term RI}
\mathbb{E} \bigl[ Z(Z-1) \, 1\{Z>1\} \bigr]
& = \mathbb{E} \bigl[(Z-1)^2 \, 1\{Z>1\} \bigr] + \mathbb{E} \bigl[(Z-1) \, 1\{Z>1\} \bigr]  \\
& = \mathbb{E} \bigl[(Z-1)^2 \, 1\{Z>1\} \bigr] + \tfrac12 \, |P-Q|
\label{igalsecond}
\end{align}
where \eqref{igalsecond} follows from \eqref{eq: TV2}, and consequently the following upper and lower
bounds on \eqref{eq: 2nd term RI} are derived:
\begin{align}  \label{eq: ub 2nd term RI}
\mathbb{E} \bigl[ Z(Z-1) \, 1\{Z>1\} \bigr]
&\leq \mathbb{E} \bigl[(Z-1)^2 \bigr] + \tfrac12 \, |P-Q| \\
&= \chi^2(P\|Q) + \tfrac12 \, |P-Q|
\label{igalast}
\end{align}
where \eqref{igalast} follows from \eqref{eq: Z} and \eqref{eq: chi-square 1}. Furthermore,
from \eqref{eq: TV3}, \eqref{eq: range of values for phi of Z} and \eqref{eq: 2nd term RI}
\begin{align}
\mathbb{E} \bigl[ Z(Z-1) \, 1\{Z>1\} \bigr]
& = \mathbb{E} \bigl[(Z-1)^2 \, 1\{Z>1\} \bigr] + \tfrac12 \, |P-Q|  \\
& = \mathbb{E} \bigl[(Z-1)^2 \bigr] - \mathbb{E} \bigl[ (Z-1)^2 \, 1\{\beta_2 \leq Z \leq 1\} \bigr]
+ \tfrac12 \, |P-Q|  \\
& = \chi^2(P\|Q) + \tfrac12 \, |P-Q|
- \mathbb{E} \bigl[ (Z-1)^2 \, 1\{\beta_2 \leq Z \leq 1\} \bigr]  \\
& \geq \chi^2(P\|Q) + \tfrac12 \, |P-Q|
- (1-\beta_2) \, \mathbb{E} \bigl[ (1-Z) \, 1\{\beta_2 \leq Z \leq 1\} \bigr]  \\
& = \chi^2(P\|Q) + \tfrac12 \, |P-Q|  - (1-\beta_2) \, \mathbb{E} \bigl[ (Z-1)^- \bigr]  \\
& = \chi^2(P\|Q) + \tfrac{\beta_2}{2} \, |P-Q|.
\label{eq: lb 2nd term RI}
\end{align}
Combining \eqref{eq: ub 2nd term RI} and \eqref{eq: lb 2nd term RI} gives the inequality in
\eqref{eq: issv15}, and combining \eqref{eq: relative entropy}, \eqref{eq: ib 1st summand}
and \eqref{eq: ub 2nd term RI} gives the upper bound on the relative entropy in \eqref{eq: ub RE-TV-chi}.

%%%%%%%%%%%%%%%%%%%%%%%%%%%%%%%%%%%%%%%%%%%%%%%%%%%%
\section{Proof of Theorem~\ref{thm: exact locus}} \label{appendix: exact locus}
\subsection{Proof of Theorem~\ref{thm: exact locus}\ref{thm: exact locus: parta})}
The concavity of the entropy functional
implies that given a probability mass function $P$ on a finite set $\{1, \ldots, |\set{A}|\}$,
and given any subset $\set{S} \subset \set{A}$,  $H(P) \leq H (P_{\set{S}} )$
with
\begin{align}
P_{\set{S}} ( k) = \left\{
\begin{array}{ll}
\frac
{P( \set{S} )}
{|\set{S}|}
&k \in \set{S}, \\
\frac
{P( \set{S}^c )}
{|\set{S}^c|}
&k \not\in \set{S}.
\end{array}
\right.
\end{align}
Applying this fact with $\set{S}$ given by the indices
of the masses below $|\set{A}|^{-1}$,
we conclude that $H(P) \leq H(\bar{P})$ with
\begin{align}
\bar{P}( k) = \left\{
\begin{array}{ll}
\frac
{\sum_{a \in \set{A}} P(a) 1\{P(a) \geq |\set{A}|^{-1}\} }
{\sum_{a \in \set{A}}         1\{P(a) \geq |\set{A}|^{-1}\}}
&k\colon P(k) \geq |\set{A}|^{-1},\\
\frac
{\sum_{a \in \set{A}} P(a) 1\{P(a) < |\set{A}|^{-1}\} }
{\sum_{a \in \set{A}}         1\{P(a) < |\set{A}|^{-1}\}}
&k\colon P(k) < |\set{A}|^{-1}.
\end{array}
\right.
\end{align}
Moreover, if $\mathsf{U}$ is the equiprobable distribution on $\set{A}$, then
\begin{align}
| P - \mathsf{U} | = | \bar{P} - \mathsf{U} |.
\end{align}
Consequently, in order to maximize entropy subject to a given (positive) total variation distance
from the equiprobable distribution on $\set{A}$, it is enough to restrict attention to distributions
whose masses take two distinct values only, i.e., of the form \eqref{eq: P for maximizing the TV}.
The only remaining optimization is to determine $m_\Delta$, the number of masses larger than
$|\set{A}|^{-1}$. The requirement that  $m_\Delta$ satisfy \eqref{conditionmdelta}
is made so that \eqref{eq: P for maximizing the TV} is a valid probability distribution.
The solution is as given in Part~\ref{thm: exact locus: parta})
since $H(P) = \log |\set{A}| - D ( P \| \mathsf{U} )$, and
\begin{align}
D( P_{\Delta} \| \mathsf{U} ) =
d \left( \tfrac{m_\Delta}{|\set{A}|} + \tfrac{\Delta}{2} \big\|  \tfrac{m_\Delta}{|\set{A}|} \right).
\end{align}

\subsection{Proof of Theorem~\ref{thm: exact locus}\ref{thm: exact locus: partb})}
The minimizing probability measure in \eqref{eq: P for minimizing the TV} is a special case of
\cite[Theorem~3]{HoY_IT2010}, which gives the general solution of minimizing the entropy subject
to a constraint on the maximal total variation distance from a fixed discrete distribution $Q$
(here, $Q = \mathsf{U}$).

%%%%%%%%%%%%%%%%%%%%%%%%%%%%%%%%%%%%%%%%%%%%%%%%%%%%
\section{Proof of \eqref{eq: limit}}  \label{appendix: limit}
\begin{align}
& \lim_{\varepsilon \to 0} \left\{\frac1{D(P_\varepsilon\|Q_\varepsilon)} \cdot
r\left(\frac{P_\varepsilon(1)}{Q_\varepsilon(1)}\right) \right\} \nonumber \\[0.15cm]
\label{lim2}
& = \lim_{\varepsilon \to 0} \frac{\left(\frac{1-\varepsilon}{1-\frac{\varepsilon}{t_\gamma}} \right)
\, \log\left(\frac{1-\varepsilon}{1-\frac{\varepsilon}{t_\gamma}} \right)
+ \left(1-\frac{1-\varepsilon}{1-\frac{\varepsilon}{t_\gamma}}\right)
\, \log e}{\varepsilon \, \log(t_\gamma) +
(1-\varepsilon) \log\left(\frac{1-\varepsilon}{1-\frac{\varepsilon}{t_\gamma}} \right)} \\[0.15cm]
\label{lim3}
& = \lim_{\varepsilon \to 0} \frac{(1-\varepsilon) \, \left[\log(1-\varepsilon)
- \log\left(1-\frac{\varepsilon}{t_\gamma}\right) \right]
+ \varepsilon \, \left(1-\frac1{t_\gamma}\right) \, \log e}{\varepsilon \, \log(t_\gamma) +
(1-\varepsilon) \, \left[\log(1-\varepsilon)
- \log\left(1-\frac{\varepsilon}{t_\gamma}\right) \right]} \\[0.15cm]
\label{lim4}
& = \lim_{\varepsilon \to 0} \frac{-\log(1-\varepsilon) - \log(e)
+ \log\left(1-\frac{\varepsilon}{t_\gamma}\right)
+ \frac1{t_\gamma} \frac{1-\varepsilon}{1-\frac{\varepsilon}{t_\gamma}} \, \log(e) +
\left(1-\frac1{t_\gamma}\right) \, \log e}{\log(t_\gamma) - \log(e) - \log(1-\varepsilon)
+ \log\left(1-\frac{\varepsilon}{t_\gamma}\right)
+ \frac1{t_\gamma} \frac{1-\varepsilon}{1-\frac{\varepsilon}{t_\gamma}} \, \log(e)} \\[0.1cm]
\label{lim5} & = 0
\end{align}
where \eqref{lim2} follows from \eqref{eq: r}, and the definition of $P_\varepsilon, Q_\varepsilon$;
\eqref{lim4} is due to L'H\^{o}pital's rule; and \eqref{lim5} holds since the numerator
in \eqref{lim4} converges to zero as $\varepsilon \to 0$ while its denominator converges to
$\log(t_\gamma) - \left(1-\frac1{t_\gamma}\right) \log(e) > 0$ (recall that
$t_\gamma \in (\gamma, \infty)$, for $\gamma > 1$, so $t_\gamma > 1$).

%%%%%%%%%%%%%%%%%%%%%%%%%%%%%%%%%%%%%%%%%%%

\section{Completion of the Proof of Theorem~\ref{thm: reris}}
\label{appendix:u}

{\em Proof of monotonicity and boundedness of \eqref{eq:u}}:
Substituting $\gamma = \beta x$
into the right side of \eqref{eq:u} gives that, for $\beta > 1$,
\begin{align}
\label{eq:u0}
u(\beta) = \min_{x \in \bigl(\frac1{\beta}, 1\bigr)} \left( \frac{c_{\beta x}}{1-x} \right).
\end{align}
The function $u$ in \eqref{eq:u0} is indeed monotonically decreasing
on $(1, \infty)$ since, if $\beta_2 > \beta_1 > 1$,
\begin{align}
\label{eq:u1}
u(\beta_1) &= \min_{x \in \bigl(\frac1{\beta_1}, 1\bigr)} \frac{c_{\beta_1 x}}{1-x}  \\
\label{eq:u2}
& \geq \min_{x \in \bigl(\frac1{\beta_1}, 1\bigr)}  \frac{c_{\beta_2 x}}{1-x}  \\
\label{eq:u3}
& \geq \min_{x \in \bigl(\frac1{\beta_2}, 1\bigr)}  \frac{c_{\beta_2 x}}{1-x}  \\
\label{eq:u4}
&= u(\beta_2)
\end{align}
where \eqref{eq:u2}
holds since $c_\gamma$ is monotonically decreasing in $\gamma \in (1, \infty)$ (see
Theorem~\ref{thm:EG vs. RE}).

{\em Proof of \eqref{eq: u-bound}}:
From \eqref{eq:sup-EG and RE} and  \eqref{eq: c_gamma}, we obtain
$t_\gamma > \gamma$ for $\gamma > 1$.
Furthermore, since $t/r(t)$ is monotonically decreasing on $(1, \infty)$,
if $\gamma>e$, then
\begin{align}\label{goder}
c_\gamma = \frac{t_\gamma-\gamma}{r(t)}< \frac{\gamma}{r(\gamma)}
= \frac{\gamma}{\log e+\gamma \log\frac{\gamma}{e}} < \frac{1}{\log \frac{\gamma}{e}}.
\end{align}
Hence, for $\beta > 2e$,
\begin{align}
\label{eq:u5}
u(\beta) &= \min_{\gamma \in (1, \beta)}
\frac{\beta \, c_\gamma}{\beta-\gamma} \\
\label{eq:u6}
& \leq 2 c_{\beta/2} \\
\label{eq:u7}
& < \frac{2}{\log\frac{\beta}{2e}}
\end{align}
where  \eqref{eq:u6} follows by choosing $\gamma = \frac{\beta}{2}$ in the minimization, and \eqref{eq:u7}
follows from \eqref{goder}.

%%%%%%%%%%%%%%%%%%%%%%%%%%%%%%%%%%%%%%%%%%%
\section{A Lemma Used for Proving \eqref{expectations1}}
\label{appendix: two pdfs}
\begin{lemma} \label{lemma: 2 pdfs}
Let $g$ be a monotonically increasing and non-negative function on $[a,b]$, and let $p_1, p_2$
be probability density functions supported on $[a,b]$. Assume that there exists
$c \in (a,b)$ such that
\begin{align} \label{eq: conditions on 2 pdfs}
\begin{split}
& p_1(\beta) \geq p_2(\beta), \quad \forall \, \beta \in [a,c], \\
& p_1(\beta) < p_2(\beta), \quad \forall \, \beta \in (c,b].
\end{split}
\end{align}
Let $W \sim p_1$ and $V \sim p_2$, then
\begin{align} \label{eq: expectations}
\mathbb{E}\bigl[g(W)\bigr] \leq \mathbb{E}\bigl[g(V)\bigr].
\end{align}
\end{lemma}
\begin{proof}
The function $d \triangleq p_2 - p_1$, defined on $[a,b]$, satisfies
\begin{align}
\label{d1}
& d(\beta) \leq 0, \quad \forall \, \beta \in [a,c] \\
\label{d2}
& d(\beta) \geq 0, \quad \forall \, \beta \in [c,b] \\
\label{d3}
& \int_a^b d(\beta) \, \text{d}\beta = 0.
\end{align}
Consequently, we get
\begin{align}
& \mathbb{E}\bigl[g(V)\bigr] - \mathbb{E}\bigl[g(W)\bigr]  \nonumber\\
\label{d6}
& = \int_a^c d(\beta) \, g(\beta) \, \mathrm{d}\beta + \int_c^b d(\beta) \, g(\beta) \, \mathrm{d}\beta \\
\label{d7}
& \geq g(c) \int_a^c d(\beta) \, \mathrm{d}\beta + g(c) \int_c^b d(\beta) \, \mathrm{d}\beta \\
\label{d8}
& = 0
\end{align}
where \eqref{d7} follows from \eqref{d2}, \eqref{d3} and the monotonicity of $g$, and
\eqref{d8} is due to \eqref{d3}.
\end{proof}

\vspace*{0.2cm}
\subsubsection*{\bf{Acknowledgment}}
Discussions with Jingbo Liu, Vincent Tan and Mark Wilde are gratefully acknowledged.

%\eject

\end{document}